\documentclass[11pt]{article}
\usepackage{fullpage}
\usepackage{authblk}
\usepackage{amsthm}
\usepackage{amsmath}
\usepackage[table,xcdraw]{xcolor}
\usepackage{multirow}
\usepackage{wrapfig}
\usepackage[leftcaption]{sidecap}
\usepackage{multicol}

\usepackage{thm-restate}
\usepackage{nicefrac}
\usepackage[boxed,lined,linesnumbered, ruled]{algorithm2e}

\usepackage{comment,amsfonts,amsmath,amsthm,graphicx,subcaption}

\definecolor{Darkblue}{rgb}{0,0,0.4}
\definecolor{Brown}{cmyk}{0,0.61,1.,0.60}
\definecolor{Purple}{cmyk}{0.45,0.86,0,0}
\definecolor{Darkgreen}{rgb}{0.133,0.543,0.133}

\usepackage[colorlinks,linkcolor=Darkblue,filecolor=blue,citecolor=blue,urlcolor=Darkblue,pagebackref]{hyperref}
\usepackage[nameinlink]{cleveref}

\newcommand{\commentout}[1]{}

\def\epsilon{\varepsilon}
\def\eps{\varepsilon}

\newcommand{\alert}[1]{\textbf{\color{red}
		[[[#1]]]}\marginpar{\textbf{\color{red}**}}\typeout{ALERT:\@
		\the\inputlineno: #1}}

\usepackage[colorinlistoftodos,prependcaption,textsize=tiny]{todonotes}
\newcommand{\atodo}[1]{}
\newcommand{\atodoin}[1]{}

\newcommand{\gtodo}[1]{}
\newcommand{\gtodoin}[1]{}

\def\eps{\epsilon}

\newcommand{\cost}{{\rm cost}}

\newcommand{\dcost}{{\rm diamcost}}

\newcommand{\opt}{{\rm OPT}}

\newcommand{\polylog}{{\rm polylog}}

\newcommand{\poly}{{\rm poly}}

\newcommand{\diam}{{\rm diam}}
\newcommand{\rad}{{\rm rad}}

\newcommand{\R}{\mathbb{R}}

\newcommand{\etal}{\emph{et.\ al. }}

\newcommand{\mommit}[1]{}
\newcommand{\namedref}[2]{\hyperref[#2]{#1~\ref*{#2}}}

\newcommand{\ball}{{\sf ball}}

\newcommand{\FGDB}{\texttt{FIND-GOOD-DENSE-BALL}}
\newcommand{\SC}{\texttt{SETTLE-CLUSTER}\xspace}
\newcommand{\rchble}{{\textsf{reachable}}}
\newcommand{\pof}{{\textsf{Partitions-of}}}
\newcommand{\guess}{\emph{\texttt{guess}}\xspace}

\theoremstyle{plain}
\newtheorem{theorem}{Theorem}[]
\newtheorem{lemma}{Lemma}[]
\newtheorem{claim}[lemma]{Claim}

\newtheorem{remark}[lemma]{Remark}
\newtheorem{definition}{Definition}

\newtheorem{assumption}{Assumption}
\newtheorem{observation}{Observation}

\newtheorem{hypothesis}{Hypothesis}[]

\newcommand{\cC}{\mathcal{C}}
\newcommand{\cD}{\mathcal{D}}

\newcommand{\cG}{\mathcal{G}}

\newcommand{\cI}{\mathcal{I}}
\newcommand{\cJ}{\mathcal{J}}

\newcommand{\cL}{\mathcal{L}}
\newcommand{\cM}{\mathcal{M}}

\newcommand{\cR}{\mathcal{R}}
\newcommand{\cS}{\mathcal{S}}
\newcommand{\cT}{\mathcal{T}}

\newcommand{\csrfull}{\textsc{Capacitated Sum of Radii}\xspace}
\newcommand{\csr}{\textsc{CapSoR}\xspace}

\newcommand{\sor}{\textsf{Sum of Radii}\xspace}
\newcommand{\sod}{\textsf{Sum of Diameters}\xspace}

\newcommand{\kcen}{\textsc{$k$-Center}\xspace}
\newcommand{\ckcen}{Capacitated \textsc{$k$-Center}\xspace}

\newcommand{\kdia}{\textsc{Max-$k$-Diameter}\xspace}
\newcommand{\ckdia}{Capacitated \textsc{Max-$k$-Diameter}\xspace}

\newcommand{\csdfull}{\textsc{Capacitated Sum of Diameters}\xspace}
\newcommand{\csd}{\textsc{CapSoD}\xspace}

\newcommand{\maxcov}{\textsc{Max $k$-Cov}\xspace}

\newcommand{\sat}{\textsc{$3$-SAT}\xspace}

\newcommand{\tp}[1]{Type-$#1$}
\newcommand{\tps}{special\xspace}
\newcommand{\tpr}{regular\xspace}

\newcommand{\pp}{\textsf{P}\xspace}
\newcommand{\np}{\textsf{NP}\xspace}
\newcommand{\fpt}{\textsf{FPT}}
\newcommand{\gapeth}{\textsf{gap-ETH}}
\newcommand{\eth}{\textsf{ETH}}

\newcommand{\nonblind}[1]{}

\author[1]{Arnold Filtser\thanks{This research was supported by the Israel Science Foundation (grant No. 1042/22).}}
\author[2]{Ameet Gadekar\thanks{Work partly done while at Bar-Ilan University.}}

\affil[1]{Bar-Ilan University, Israel. }
\affil[2]{CISPA Helmholtz Center for Information Security, Saarbr\"{u}cken, Germany.\qquad\qquad\qquad Emails:  \texttt{arnold.filtser@biu.ac.il},~ \texttt{ameet.gadekar@cispa.de}}

\title{FPT Approximations for Capacitated Sum of Radii and Diameters}
\date{}
\begin{document}
\pagenumbering{gobble}

\maketitle

\begin{abstract}
The \textsf{Capacitated Sum of Radii} problem involves partitioning a set of points $P$, where each point $p\in P$ has capacity $U_p$, into $k$ clusters that minimize the sum of cluster radii, such that the number of points in the cluster centered at point $p$ is at most $U_p$.
We begin by showing that the problem is APX-hard, and that under \textsf{gap-ETH} there is no parameterized approximation scheme (FPT-AS).
We then construct a $\approx5.83$-approximation algorithm in \textsf{FPT} time (improving a previous $\approx7.61$ approximation in \textsf{FPT} time).
Our results also hold when the objective is a general monotone symmetric norm of radii.  
We also improve the approximation factors for the uniform capacity case, and for the closely related problem of \textsf{Capacitated Sum of Diameters}.
 \end{abstract}

    	\vfill
	\begin{multicols}{2}
		{
			\setcounter{secnumdepth}{5}
			\setcounter{tocdepth}{2} \tableofcontents
		}
	\end{multicols}
	\newpage

\pagenumbering{arabic}


\section{Introduction}
Clustering is a fundamental problem in several domains of computer science, including, data mining, operation research, and computational geometry, among others. In particular, center based clustering problems such as $k$-median, $k$-means, and $k$-center have received significant attention from the research community for more than half a century~\cite{llyod,HS85, MEGIDDO1990327, badoiu-etal:approximate-clustering-coresets,AryaGKMMP04,kumar2010linear,dasgupta08hardness-2-means,har2004coresets,matouvsek2000approximate,KanungoMNPSW04,ByrkaPRST17,bhattacharya2018faster,cohen2019inapproximability,Cohen-AddadG0LL19, Cohen-AddadGHOS22,CKL22,byrka2018constant,10353074,GowdaPST23}. In these problems, we are given a set $P$ of $n$ points together with a distance function (metric) and a positive integer $k$. The goal is to partition $P$ into $k$ parts called \textit{clusters} and choose a \textit{center} point for each cluster, so that to minimize a clustering objective that is a function of the point distances to their centers. 
A related fundamental problem that helps reduce \textit{dissection effect} due to $k$-center~\cite{MS89,CHARIKAR2004417} is called \sor. 
Here the goal is to choose $k$-size subset $X$ of $P$ (called centers, as before) and assign every point to an element in $X$. This partitions the set $P$ into possible $k$ clusters, $C_1,\cdots, C_k$, where cluster $C_i$ corresponds to the set of points assigned to $x_i \in X$. The radius of  cluster $C_i$ centered at $x_i \in X$ is the maximum distance of a point in $P_i$ to $x_i$. The objective is the sum of radii of the clusters $C_1,\cdots, C_k$. \footnote{Alternatively, in the \sor problem we choose centers $x_1,\dots,x_k\in X$, and radii $r_1,\dots,r_k\in\R_{\ge 0}$ so that $P\subseteq\cup_{i=1}^k\ball(x_i,r_i)$. The objective is to minimize $\sum_{i=1}^{k}r_i$.\label{foot:SORdef}}

In the recent years \sor received a great share of interest in all aspects~\cite{CHARIKAR2004417, 10.5555/1347082.1347172,GibsonKKPV12,GibsonKKPV12,  BehsazS15,friggstad_et_al, BuchemERW24, ChenXXZ24, inamdar_et_al:LIPIcs.ESA.2020.62, bandyapadhyay_et_al:LIPIcs.SoCG.2023.12,jaiswal_et_al:LIPIcs.ITCS.2024.65,BandyapadhyayV16,AhmadianS16,BhowmickIV21,FotakisK14,BiloCKK05}. 
Nevertheless, its  computational landscape is not yet fully understood. While, the problem is NP-hard~\cite{10.5555/1347082.1347172} even in weighted planar graphs and in metrics of bounded doubling dimensions, it is known to admit a QPTAS (quasi polynomial time approximation scheme)~\cite{10.5555/1347082.1347172} in general metrics, thus prompting a possibility 
of PTAS (polynomial time approximation scheme).~\footnote{If \sor is APX-hard (equivalently, does not admit a PTAS) then NP$\subseteq$QP (in particular there is a quasi-polynomial time algorithm solving SAT). It is widely believed that  NP$\not\subseteq$QP.}
This is in contrast to related clustering problems like $k$-center, $k$-median, and $k$-means which are all known to be APX-hard~\cite{HS85,JainMS02,awasthi2015hardness}.
Currently, the present best known approximation factor in polynomial time is $(3+\eps)$ due to Buchem \etal \cite{BuchemERW24} (improving over previous results of \cite{CHARIKAR2004417,friggstad_et_al}).
Additionally, there is a recent $(2+\eps)$-approximation algorithm \cite{ChenXXZ24} that runs in \fpt~time\footnote{In this paper, by \fpt\ time, we mean \fpt\  w.r.t. the parameter $k$.}   (fixed parameter tractable).

\paragraph*{Capacitated Clustering.} 
We are interested in a much more challenging generalization where each point $p$ has an inherent capacity $U_p \ge 0$, indicating an upper bound on the number of points it can serve as a center of a cluster in the solution.%
\footnote{Formally, following \cref{foot:SORdef}, a valid solution to the capacitated version also contains an assignment $\sigma:P\rightarrow\{x_1,\dots,x_k\}$ such that if $\sigma(p)=x_i$ then $\delta(p,x_i)\le r_i$, and for every $i$, $|\sigma^{-1}(x_i)|\le U_{x_i}$.}
This is known as   the \textit{Non-Uniform Capacitated} clustering problem. If all the points have the same capacity $U$, the problem is referred to as the \textit{Uniform Capacitated} clustering problem.
Capacitated clustering naturally models many applications: in \emph{load balancing}, cluster centers correspond to servers with limited service capacity, the radius captures latency, and minimizing the sum of radii models reducing aggregate latency subject to serve load; in \emph{wireless sensors}, each cluster head can accommodate only bounded number of devices due to bandwidth, memory, or battery limitations, the radius corresponds to communication range, and optimizing the sum of radii reflects minimizing total communication energy.
Capacitated clustering is thoroughly studied: 
Capacitated $k$-center admits a constant factor approximation in polynomial time for both uniform~\cite{KhullerS00} and non-uniform capacities~\cite{CyganHK12, AnBCGMS15}.
There are several bi-criteria polytime approximations for $k$-median \cite{ChuzhoyR05,BFRS15,ByrkaRU16,DemirciL16,Li16,Li17}.

There are \fpt\ approximation algorithms for capacitated $k$-median and capacitated $k$-means with approximation factors $(3+\eps)$ and $(9+\eps)$ respectively \cite{Cohen-AddadL19} (improving over \cite{AdamczykBMM019,XZZ19}). From the other hand, assuming \gapeth\footnote{Informally, \gapeth\ says that there exists an $\eps>0$ such that no sub-exponential time algorithm for \sat  can distinguish whether a given \sat formula has a satisfying assignment or every assignment satisfies at most $(1-\eps)$ fraction of the clauses.\label{foot:gapETH}}, in \fpt~ time it is impossible to approximate $k$-median and $k$-means with factors better than $(1+\frac2e)$ and $(1+\frac8e)$, respectively (even without capacities) \cite{Cohen-AddadG0LL19}.
There is a $4+\eps$ \fpt-approximation for uniform capacitated $k$-center \cite{jaiswal_et_al:LIPIcs.ITCS.2024.65,GOYAL2023190}\footnote{In fact, \cite{GOYAL2023190} obtained a $2$-approximation for capacitated $k$-center under \emph{soft} assignments (where different centers can be co-located). In contrast, in this paper we consider only \emph{hard} assignments (where we can open only a single center at each point). One can use \cite{GOYAL2023190} to obtain a $4$-approximation for uniform capacitated $k$-center w.r.t. hard assignments. See \Cref{subsec:related} for a further discussion.\label{footnote:GJ23}}. 

For capacitated \sor (\csr) (see \Cref{tb:results} for a summary of previous and new results) Inamdar and Varadarajan \cite{inamdar_et_al:LIPIcs.ESA.2020.62} constructed an \fpt~time algorithm providing an $28$-approximation for  \sor with uniform capacities. This was later improved to $4+\eps$ by Bandyapadhyay, Lochet, and Saurabh
\cite{bandyapadhyay_et_al:LIPIcs.SoCG.2023.12}, and finally to $3+\eps$ by  Jaiswal, Kumar, and Yadav \cite{jaiswal_et_al:LIPIcs.ITCS.2024.65}. For non-uniform capacities, Jaiswal \etal obtained $(4+\sqrt{13}+\eps)\approx 7.61$ approximation in \fpt~ time, improving a previous approximation of $15+\eps$ \cite{bandyapadhyay_et_al:LIPIcs.SoCG.2023.12}.

From the lower bound side, Jaiswal \etal \cite{jaiswal_et_al:LIPIcs.ITCS.2024.65} showed that assuming \eth\footnote{Informally, \eth\ (Exponential Time Hypothesis), says that there is no subexponential time algorithm for $3$-SAT.\label{foot:eth}}, there is some constant $\beta>0$, such that any $\beta$-approximation for \csr  (with non-uniform capacities) requires $2^{\Omega(k/\poly \log k)} n^{O(1)}$ time.
The first contribution of this paper is to show that \csr is APX hard for an explicit factor and even with uniform capacities. That is compared to \cite{jaiswal_et_al:LIPIcs.ITCS.2024.65} we removed the \eth~assumption,  used only uniform capacities, and showed an explicit factor of inapproximability. 
This is in contrast to the uncapacitated version which is believed to admit a PTAS (assuming $NP\not\subseteq QP$).
\begin{restatable}[APX-hard]{theorem}{SORAPX}\label{thm: nphard}
For every constant $\epsilon >0$, it is NP-hard to approximate Uniform \csr to within a factor of $(1+1/e-\epsilon)$.
\end{restatable}   

Earlier, Bandyapadhyay \etal \cite{bandyapadhyay_et_al:LIPIcs.SoCG.2023.12} showed that assuming \eth$^{\ref{foot:eth}}$, no \fpt \ algorithm can solve \csr exactly (even with uniform capacities).
However, until this point, nothing ruled out the existence of an \fpt-Approximation Scheme (\fpt-AS) or PAS for Parameterized Approximation Scheme\footnote{Such algorithms find $(1+\eps)$-approximation in time $f(k,\eps) n^{g(\eps)}$ for some fixed functions $f$ and $g$.}.
The second contribution of this paper is to show that assuming \gapeth$^{\ref{foot:gapETH}}$, no such \mbox{FPT-AS} exists.

\begin{restatable}[No FPT-AS]{theorem}{FPTLB}
	\label{thm:gapethhard}
	Assuming \gapeth, for every $\eps>0$ and any function $f$, there is no $f(k)n^{o(k)}$ time algorithm that approximates Uniform \csr to within a factor of $(1+1/e-\epsilon)$.
\end{restatable}
Note that if we allow $n^{O(k)}$ time, then we can exactly solve \sor~using brute-force, even with non-uniform capacities (see~\cref{foot: bruteforce for sor}).
Therefore, \Cref{thm:gapethhard} not only rules out \fpt-approximation algorithms for achieving approximation better than factor of $(1+1/e-\epsilon)$ for this problem but also implies that in this context, essentially the best algorithm is the naive brute-force algorithm, which runs in  $n^{O(k)}$ time.

Next, as the best we can hope for is a constant factor approximation in \fpt~time, we turn to improving this factor. The main result of the paper is a $(3+\sqrt{8}+\eps) \approx 5.83$ \fpt-approximation for non-uniform \csr, significantly improving the present best~\cite{jaiswal_et_al:LIPIcs.ITCS.2024.65} factor of $(4+\sqrt{13}+\eps)\approx 7.61$. 
The theorem also mentions cluster capacities and other objectives. These will be explained in the following sub-sections.

\begin{restatable}[Main Theorem]{theorem}{MainFPTUB}\label{thm: nuni csr}
	There is a (deterministic) \fpt-approximation algorithm that finds a $(3+2\sqrt{2} + \epsilon)$-approximation for Non-uniform \csr for any $\epsilon>0$, and runs in time $2^{O(k^2\log k + k \log (k/\epsilon))} n^{O(1)}$. Furthermore,  the algorithm yields  $(3+2\sqrt{2} + \epsilon)$-approximation for non-uniform capacities even when the objective is a monotone symmetric norm of the radii. \\Both results hold also w.r.t. cluster capacities.
\end{restatable}

\begin{table}[ht!]
	\setlength\doublerulesep{0.4cm} 
	\centering
	\begin{tabular}{|>{\centering}m{0.2\textwidth}|c|c|c|c|}
		\hline
		\textbf{Capacities}                                    & \textbf{Norm}                           & \textbf{Approx. factor  }                                                    & \textbf{Run time}        & \textbf{Reference}                                                     \\ \hline
		\multicolumn{5}{|c|}{\sor}                                                                                                                                                                                                                                             \\ \hline
		\multicolumn{1}{|c|}{}                               & \multicolumn{1}{c|}{}                               & \multicolumn{1}{c|}{$28$}                    & $2^{O(k^2)} n^{O(1)}$                               &~\cite{inamdar_et_al:LIPIcs.ESA.2020.62} \\ \cline{3-5} 
		\multicolumn{1}{|c|}{}                               & \multicolumn{1}{c|}{}                               & \multicolumn{1}{c|}{$4+\epsilon$}                    & $2^{O(k \log (k/\epsilon))} n^{O(1)}$                               &~\cite{bandyapadhyay_et_al:LIPIcs.SoCG.2023.12} \\ \cline{3-5} 
		\multicolumn{1}{|c|}{}                               & \multicolumn{1}{c|}{\multirow{-3}{*}{$\ell_1$}}          & \multicolumn{1}{c|}{{ $3+\epsilon$}}            & \multicolumn{1}{c|}{}               & \multicolumn{1}{c|}{}           \\ \cline{2-3} 
		\multicolumn{1}{|c|}{}                               & \multicolumn{1}{c|}{}  &                      \multicolumn{1}{c|}{$(1+\epsilon)(2 \cdot 4^{p-1}+1)^{1/p}$} &\multirow{-2}{*}{$2^{O(k \log (k/\epsilon))} n^{O(1)}$}                       & \multirow{-2}{*}{~\cite{jaiswal_et_al:LIPIcs.ITCS.2024.65}} \\ \cline{3-5} 
		\multicolumn{1}{|c|}{}                      & \multicolumn{1}{c|}{\multirow{-2}{*}{$\ell_p$}} & {\cellcolor[HTML]{EFEFEF} $(1+\epsilon)(2\cdot 3^{p-1}+1)^{1/p}$} & {\cellcolor[HTML]{EFEFEF}$2^{O(k^2+k \log (k/\epsilon))} n^{O(1)}$} & {\cellcolor[HTML]{EFEFEF} \Cref{thm: uni}}                             \\ \cline{2-5} 
		\multicolumn{1}{|c|}{\multirow{-8}{*}{Uniform}}                    & \multicolumn{1}{c|}{{general}} & {\cellcolor[HTML]{EFEFEF} $3+\epsilon$} &{\cellcolor[HTML]{EFEFEF}$2^{O(k^2+k \log (k/\epsilon))} n^{O(1)}$}  & {\cellcolor[HTML]{EFEFEF} \Cref{thm: uni}}                             \\ \hline
        {\footnotesize(Non-Uni-)}Hardness            & general & $\beta>1$ & $2^{\Omega(\frac{k}{\polylog k})}\cdot\poly(n)$ (\eth)  & \cite{jaiswal_et_al:LIPIcs.ITCS.2024.65}\\\hline
        
		\multirow{3}{*}{\begin{tabular}[c]{@{}c@{}} Hardness\\ Uniform \end{tabular}}                            & \multicolumn{1}{c|}{\multirow{3}{*}{general}}& exact & $n^{\Omega(k)}$ (\eth)  & \cite{bandyapadhyay_et_al:LIPIcs.SoCG.2023.12}\\\cline{3-5}
		& &{ \cellcolor[HTML]{EFEFEF}} & {\cellcolor[HTML]{EFEFEF}NP-hard}    &{\cellcolor[HTML]{EFEFEF} \Cref{thm: nphard}}\\\cline{4-5}

		& & {\cellcolor[HTML]{EFEFEF} \multirow{-2}{*}{$(1+1/e-\eps)$}} & {\cellcolor[HTML]{EFEFEF}$n^{\Omega(k)}$ (\gapeth)}  & {\cellcolor[HTML]{EFEFEF}\Cref{thm:gapethhard}}\\\hline
		
		\multirow{3}{*}{\begin{tabular}[c]{@{}c@{}}Non-Uniform \end{tabular}}                                                 & $\ell_1$                            & $15+\epsilon$                &$2^{O(k^2 \log (k/\epsilon))} n^{O(1)}$                                 & ~\cite{bandyapadhyay_et_al:LIPIcs.SoCG.2023.12}             \\ \cline{2-5} 
		& $\ell_p$                       & $4+\sqrt{13}+\epsilon \approx 7.61$     &$2^{O(k^3+k \log (k/\epsilon))} n^{O(1)}$                                   &~\cite{jaiswal_et_al:LIPIcs.ITCS.2024.65}                   \\ \cline{2-5} 
		& {general} & {\cellcolor[HTML]{FBE4E2} $3+\sqrt{8}+\epsilon \approx 5.83$}       &{\cellcolor[HTML]{FBE4E2}$2^{O(k^2 \log k+k \log (k/\epsilon))} n^{O(1)}$}          & {\cellcolor[HTML]{FBE4E2} \Cref{thm: nuni csr}}                             \\ \hline

		\hline \hline    
		\multicolumn{5}{|c|}{\sod}                                                                                                                                                                                                         \\ \hline

		\multicolumn{1}{|c|}{}                               & \multicolumn{1}{c|}{}                               & \multicolumn{1}{c|}{$6+\epsilon$}                                    &{$2^{O(k \log (k/\epsilon))} n^{O(1)}$}                  &~\cite{jaiswal_et_al:LIPIcs.ITCS.2024.65} \\ \cline{3-5} 
		\multicolumn{1}{|c|}{}                               & \multicolumn{1}{c|}{\multirow{-2}{*}{$\ell_1$}}          & \multicolumn{1}{c|}{{\cellcolor[HTML]{EFEFEF} $4+\epsilon$}}                  &{\cellcolor[HTML]{EFEFEF}$2^{O(k \log (k/\epsilon))} n^{O(1)}$}             & {\cellcolor[HTML]{EFEFEF} \Cref{thm: uni dia}}           \\ \cline{2-4} 
		\multicolumn{1}{|c|}{}                               & \multicolumn{1}{c|}{}                               & \multicolumn{1}{c|}{$(2+\epsilon)(2 \cdot 4^{p-1}+1)^{1/p}$}                  &{$2^{O(k \log (k/\epsilon))} n^{O(1)}$}        &~\cite{jaiswal_et_al:LIPIcs.ITCS.2024.65}                                       \\ \cline{3-5} 
		\multicolumn{1}{|c|}{}                               & \multicolumn{1}{c|}{\multirow{-2}{*}{$\ell_p$}}     & \multicolumn{1}{c|}{{\cellcolor[HTML]{EFEFEF} $(2+\epsilon)(2^{p-1}+1)^{1/p}$}}   &{\cellcolor[HTML]{EFEFEF}$2^{O(k \log (k/\epsilon))} n^{O(1)}$}     & {\cellcolor[HTML]{EFEFEF} \Cref{thm: uni dia}}           \\ \cline{2-5} 
		
		\multicolumn{1}{|c|}{\multirow{-5}{*}{Uniform}}      & \multicolumn{1}{c|}{{general}} & \multicolumn{1}{c|}{{\cellcolor[HTML]{EFEFEF} $4+\epsilon$}}    &{\cellcolor[HTML]{EFEFEF}$2^{O(k \log (k/\epsilon))} n^{O(1)}$}                          & {\cellcolor[HTML]{EFEFEF} \Cref{thm: uni dia}}           \\ \hline
		{Non-Uniform} & \multicolumn{1}{c|}{{general}} & \multicolumn{1}{c|}{{\cellcolor[HTML]{EFEFEF} $7+\epsilon$}}    &{\cellcolor[HTML]{EFEFEF}$2^{O(k^2\log k+ k\log (k/\epsilon))} n^{O(1)}$}                          & {\cellcolor[HTML]{EFEFEF} \Cref{thm: nuni csd}}           \\ \hline
	\end{tabular}
	\caption{\small Summary of our results and previous results for \csr and \csd. 
		The results for \csd are implicit in the corresponding references. Furthermore, \csd makes sense only for the cluster capacities. While \csr with the $\ell_\infty$ objective corresponds to the \ckcen problem, \csd with the $\ell_\infty$ objective corresponds to the \ckdia problem. While our algorithms are deterministic for non-uniform capacities, the algorithms for the uniform capacity case are randomized.
        The inapproximability result, assuming \gapeth,  rules out any $f(k)n^{o(k)}$ time algorithm   that can approximate  uniform \csr (even with $\ell_1$ objective) to a factor better than $(1+1/e)$.
        }
	\label{tb:results}
\end{table}

\subsection{Other Norm Objectives}\label{subsec:Norms}
In the \sor problem the goal is to choose $k$ balls covering all the metric points such that the sum of radii $\sum_ir_i$ is minimized. A natural generalization is to optimize alternate objectives. Specifically, given a norm $\|\cdot\|:\R^k\rightarrow\R_{\ge0}$, the task is to choose balls covering all metric points so as to minimize the norm of the radii vector $\|(r_1,\dots,r_k)\|$.
A canonical example is the $\ell_p$ norm where $\|(r_1,\dots,r_k)\|_p=\left(\sum_{i=1}^kr_i^p\right)^{1/p}$. Note that the $\ell_1$ norm objective is exactly the \sor problem, while the $\ell_\infty$ norm objective ($\|(r_1,\dots,r_k)\|_\infty=\max_{i\in[1,k]}r_i$) is the $k$-center problem. The other $\ell_p$-norm objectives interpolate between these two fundamental problems in clustering. In particular, they capture objectives such as the sum of squared radii,\footnote{Sum of squared radii is essentially equivalent to the $\ell_2$ objective: an $\alpha$-approximation algorithm for the $\ell_2$ objective yields an $\alpha^2$-approximation to the sum of squared radii.} which has  received attention in practice, especially in wireless network applications where it naturally models power consumption~\cite{alt2006minimum,liu2022primal}. The strength of such generalizations has been recently studied~\cite{chakrabarty2019approximation,10353074}, which not only  unified the existing \fpt-AS for several different problems such as $k$-median, $k$-means, and $k$-center, but also lead to \fpt-AS for advanced problems such as priority $k$-center~\cite{li2003asymmetry,bajpai2021revisiting,plesnik1987heuristic}, $\ell$-centrum~\cite{LocationTheory,LocationScience}, ordered $k$-median\cite{byrka2018constant,braverman2019coresets}, and Socially Fair $k$-median~\cite{anthony2010plant,bhattacharya2014new,pmlr-v134-makarychev21a,ghadiri2021socially}, many of which were previously unresolved.

Jaiswal \etal \cite{jaiswal_et_al:LIPIcs.ITCS.2024.65} constructed an algorithm for the uniform capacitated \sor with $\ell_p$ norm of radii with approximation factor  $(1+\eps)(2\cdot4^{p-1}+1)^{1/p}$ in \fpt~time. In particular, this implies a $(3+\eps)$ \fpt-approximation for uniform capacitated \sor $(p=1)$ and $(4+\eps)$ \fpt-approximation$^{\ref{footnote:GJ23}}$ for uniform capacitated $k$-center $(p=\infty)$.
In our work, we generalize and improve these results. 
Specifically, we obtain a $(3+\eps)$ \fpt-approximation w.r.t. any monotone symmetric norm (see \Cref{sec:prelims} for definitions) of radii, thus generalizing the $\ell_p$ norm objective of \cite{jaiswal_et_al:LIPIcs.ITCS.2024.65}. Further, we also improve the approximation ratio for the $\ell_p$ norm objective to  $(1+\epsilon) (2\cdot 3^{p-1}+1)^{1/p}$.

\begin{restatable}[Uniform]{theorem}{UnifromMonSymm}\label{thm: uni}
	There is a randomized algorithm for Uniform \csr, oblivious to the objective, that runs in $2^{O(k^2+k \log (k/\epsilon))} n^{O(1)}$ time, and with probability at least $3/5$ returns a solution ${\rm Sol}$, such that ${\rm Sol}$ is a $(3+\epsilon)$-approximation w.r.t. any monotone symmetric norm objective (simultaneously). 
	Furthermore, for $1< p <\infty$, ${\rm Sol}$ is a $(1+\epsilon) (2\cdot 3^{p-1}+1)^{1/p} $-approximation w.r.t. the $\ell_p$ norm objective.
\end{restatable}

For the special case of $p=\infty$, the $\ell_\infty$ norm objective is simply the $k$-center problem. Here a simple corollary  of \Cref{thm: uni} implies $(3+\eps)$ \fpt-approximation algorithm for $k$-center with uniform capacities, improving the state of the art factor of $4$ due to~\cite{jaiswal_et_al:LIPIcs.ITCS.2024.65,GOYAL2023190}.
Note that for $p \in (1,\infty]$, the approximation factor of our \Cref{thm: uni} is better than that of \cite{jaiswal_et_al:LIPIcs.ITCS.2024.65}. 
In fact, the approximation factor of our algorithm for $\ell_p$ norm objective equals $\max_{\alpha\in[0,1]}\left(\frac{(2+\alpha)^{p}+1}{1+\alpha^{p}}\right)^{\nicefrac{1}{p}}$, which is slightly better than the stated factor. 
For example, it equals $\approx 2.414$ and $\approx2.488$ for  $p=2$ and $p=3$, respectively (instead of the stated $\approx2.65$ and $\approx2.67$). 
Similarly, \cite{jaiswal_et_al:LIPIcs.ITCS.2024.65} explicitly claimed approximation factors of $3$ and $\approx 3.2$, respectively. 
However, if we optimize their final expression, it yields factors $\approx 2.92$ and $\approx3.191$ for $p=2$ and $p=3$, respectively. See~\Cref{rem: better factor for p} for further discussion.

For non-uniform Capacitated \sor, Jaiswal \etal \cite{jaiswal_et_al:LIPIcs.ITCS.2024.65} obtained approximation factor of  $(4+\sqrt{13}+\eps)\approx 7.61$ w.r.t. any $\ell_p$-norm objective in $2^{O(k^3 + k \log (k/\epsilon))} n^{O(1)}$ time.
As stated in \Cref{thm: nuni csr} above, for non-uniform Capacitated \sor, we obtain approximation factor of $(3+\sqrt{8}+\eps) \approx 5.83$ w.r.t. any monotone symmetric norm objective, in $2^{O(k^2 \log k + k \log (k/\epsilon))} n^{O(1)}$ time.
Thus we improve over \cite{jaiswal_et_al:LIPIcs.ITCS.2024.65} on three fronts (see~\Cref{tb:results}): (1) approximation factor, (2) generalizing to any monotone symmetric norm from $\ell_p$ norms, and (3) running time.

\subsection{\sod and Cluster Capacities}\label{subsec:diameters}
A closely related problem is the \sod problem, which has been studied extensively and predates \sor~\cite{HJ87,HJ97,MS89,CRW91,10.5555/642992.642997,CHARIKAR2004417,BehsazS15,10.5555/642992.642997,friggstad_et_al}. 
Here, the goal is to partition the point set $P$ into $k$ clusters $C_1,\dots,C_k$, and the objective is to minimize the sum of cluster diameters $\sum_{i=1}^k\diam(C_i)$, where the diameter of a cluster is $\diam(C_i)=\max_{x,y\in C_i}\delta(x,y)$ the maximum pairwise distance between two cluster points.
Note that, unlike \sor, in this problem, there are no centers representing the clusters.
Furthermore, this problem is NP-hard to approximate to a factor better than $2$ in polynomial time~\cite{10.5555/642992.642997} (unlike \sor, which admits a QPTAS).
A simple observation shows that any $\alpha$-factor approximation for \sor implies $2\alpha$-factor approximation for \sod in a black-box way (and vice-versa).%
\footnote{Given a set of points $P$, denote by $\cR$ and $\cD$ the value of the optimal solutions to \sor and \sod respectively. It holds that $\cR\le \cD\le2\cdot \cR$. Indeed, consider an optimal solution to \sod of cost $\cD$, then by picking arbitrary center in each cluster we obtain a solution for \sor of cost at most $\cD$ (thus $\cR\le\cD$).  On the other hand, given a solution to \sor of cost $\cR$, the clusters induced by the balls constitute a solution to \sod of cost at most $2\cdot\cR$ (thus $\cD\le2\cdot\cR$).\label{foot:SORroSOD}}
This trick has often been used to design approximation algorithms for \sod. 
In fact, the current state-of-the-art algorithms for \sod, including a polynomial-time \((6+\epsilon)\)-approximation~\cite{BuchemERW24} and a quasi-polynomial-time \((2+\epsilon)\)-approximation~\cite{10.5555/1347082.1347172}, are based on this implicit trick by applying it to the polynomial-time \((3+\epsilon)\)-approximation and QPTAS for \sor, respectively.

\paragraph*{Capacitated \sod. } We introduce the problem of Capacitated \sod (\csd). Here we are initially given $k$ capacities $U_1,\dots,U_k$, the goal is to partition the point set $P$ into $k$ clusters $C_1,\dots,C_k$, such that for every $i$, $|C_i|\le U_i$, while the objective is to minimize the sum of cluster diameters. If all the capacities are equal $U_1=\dots=U_k$, we will call the capacities uniform (and otherwise as non-uniform).
Under uniform capacities, the reduction mentioned in \cref{foot:SORroSOD} goes through. Thus an $\alpha$ approximation algorithm to uniform capacitated \sor transfers in a black-box manner into a $2\alpha$ approximation algorithm for the uniform capacitated \sod, with the same running time.
In particular, by using \cite{jaiswal_et_al:LIPIcs.ITCS.2024.65}, one can obtain a $6+\eps$-approximation for uniform capacitated \sod in \fpt \ time.
In fact, similarly to \sor, one can study Capacitated \sod w.r.t. any norm objective, and the reduction will still go through. Thus it follows from \cite{jaiswal_et_al:LIPIcs.ITCS.2024.65} that for any $\ell_p$ norm objective, uniform capacitated \sod admits $2(1+\eps)(4^{p-1}+1)^{1/p}$ \fpt\  approximation.
Similarly, using our \Cref{thm: uni}, we can obtain \fpt time $6+\eps$ approximation for uniform capacitated \sod w.r.t. to any monotone symmetric norm, or 
$2(1+\eps)(3^{p-1}+1)^{1/p}$-approximation w.r.t. $\ell_p$ norm objective.
\sod is a fundamental and important problem. Its capacitated version was not previously explicitly studied simply because there was nothing to say beyond this simple reduction.

In our work, we go beyond this reduction and directly design novel approximation algorithms for capacitated \sod with significantly better approximation factors than twice that of \sor.
\begin{restatable}[Uniform Diameters]{theorem}{uniDiameter}\label{thm: uni dia}
There is a randomized algorithm that given an instance of the Uniform \csd runs in $2^{O(k \log (k/\epsilon))} n^{O(1)}$ time, and with probability at least $3/5$ returns a solution ${\rm Sol}$, such that ${\rm Sol}$ is a $(4+\epsilon)$-approximation w.r.t. any monotone symmetric norm objective (simultaneously). 
Furthermore, for $1< p <\infty$, ${\rm Sol}$ is a $(2+\epsilon) (2^{p-1}+1)^{1/p}$-approximation w.r.t. the $\ell_p$ norm objective.
\end{restatable}

Finally, we proceed to consider the more challenging problem of non-uniform Capacitated \sod. Here we obtain a $7+\eps$ approximation w.r.t. to any monotone symmetric norm objective.
\begin{restatable}[Non-Uniform Diameters]{theorem}{GeneralCapDiameter}\label{thm: nuni csd}
For any $\epsilon>0$, there is a (deterministic) algorithm running $2^{O(k^2 + k \log (k/\epsilon))} n^{O(1)}$-time and returns a $(7 + \epsilon)$-approximation for Non-uniform \csd w.r.t. any monotone symmetric norm objective.
\end{restatable}

\paragraph{Cluster Capacities.}\label{para:clusterCapacities}
Bandyapadhyay \etal \cite{bandyapadhyay_et_al:LIPIcs.SoCG.2023.12} introduced the problem of Capacitated \sor where each point $p\in P$ has a capacity $U_p$, and a cluster centered in $p$ can contain at most $U_p$ points. This corresponds for example to a scenario where we want to construct water wells, and a well constructed at point $p$ can serve up to $U_p$ clients.
However, an equally natural problem is where one is given $k$ capacities $U_1,\dots,U_k$, and the goal is to construct $k$ clusters with arbitrary centers, such that the $i$'th cluster contains at most $U_i$ points. This problem is similar to our capacitated \sod, and can correspond to a scenario where one wants to distribute already existing $k$ water tanks (for example in a tent village during a festival).
We refer to the two versions of the problem as node capacities, and cluster capacities, respectively.
Note that for uniform capacities the two versions 
coincide.
Further, note that the reduction from \sor to \sod mentioned in \cref{foot:SORroSOD} holds in the capacitated version w.r.t. cluster capacities. 
Our results on node capacities in \Cref{thm: nuni csr} hold for cluster capacities as well ($3+2\sqrt{2}+\eps$ approximation for any monotone symmetric norm objective). No results on \sor with non-uniform cluster capacities were previously known.

\section{Overview of Techniques}\label{sec: tech}
In this section, we highlight our conceptual and technical contributions. 
Due to space constraints, 
we begin with our main theorem (\Cref{thm: nuni csr}) and present the key ideas behind our algorithm for non-uniform capacities in~\Cref{ss: overview nuni}. 
Then, in~\Cref{ss: overview uni}, we delve into the algorithmic ideas for uniform capacities.
Note however, that from a pedagogical viewpoint, it is easier to begin reading first \Cref{ss: overview uni}, followed by~\Cref{ss: overview nuni}.
We now set up basic notations required for the exposition.

\textit{Notations.}
We denote by $\cI$ an instance of capacitated \sor or capacitated \sod, depending upon the context. For uniform capacities, $\cI$ consists of metric space $(P,\delta)$, a positive integer $k$, and a uniform capacity $U > 0$. The elements in $P$ are called  points. 
For non-uniform node capacities, $U$ is replaced by corresponding node capacities $\{U_p\}_{p \in P}$, and for non-uniform cluster capacities, $U$ is replaced by corresponding cluster capacities $\{U_1,\cdots, U_k\}$.
We denote by $\cC=\{C_1,\cdots,C_k\}$  an optimal (but fixed) clustering for $\cI$. For ease of analysis, we assume that $|\cC| =k$, otherwise, we can add zero radius clusters to make it $k$, without increasing the cost of $\cI$.   
Note that the clusters in $\cC$ are disjoint.
When $\cI$ is an instance of   capacitated \sor, we denote by $r^*_i$  and $o_i$ as the radius and the center of cluster $C_i \in \cC$, respectively, and let  $O = \{o_1,\cdots, o_k\}$. In this case, we let   $\text{OPT}= r^*_1 + \cdots+ r^*_k$. 
When $\cI$ is an instance of   capacitated \sod, we denote by $d^*_i$ as the diameter of cluster $C_i \in \cC$. In this case, we let  $\text{OPT}= d^*_1 + \cdots+ d^*_k$.
It is known that we can guess in \fpt\ time in $k$, a set ($\epsilon$-approximation) $\{r_1, \cdots, r_k\}$ corresponding to $\{r^*_1, \cdots, r^*_k\}$ such that $\sum_{i \in [k]} r_i \le (1+\eps)\cdot\sum_{i \in [k]}  r^*_i$ and $r_i \ge r^*_i$. Similarly, let  $\{d_1, \cdots, d_k\}$ denote  $\epsilon$-approximation of $\{d^*_1, \cdots, d^*_k\}$. A feasible solution to $\cI$ is a partition of $P$ (along with centers for \sor) such that the capacities are respected.
For a point $p \in P$ and a positive real $r$, denote by $\ball(p,r)$ as the set of points from $P$ that are at a distance at most $r$ from $p$.

\begin{remark}[Point assignment using matching and flows]\label{rem: point assgn}
In this section (and throughout the paper), we will focus on presenting a set of centers and their corresponding radii as our solution. This approach suffices because, given such a set of centers and radii, we can determine in polynomial time whether there exists a corresponding feasible solution (i.e., feasible assignment of points to the centers/clusters) with the same cost. Moreover, we can also find such a solution by defining an appropriate flow problem (see \Cref{thm: feasibility by matching,thm: feasibility by matching dia}).
\end{remark}

\begin{remark}[\fpt\ and bounded guessing]
    We assume that our algorithms have a power to make $\poly(k)$ guesses each with success probability $p(k)$. Such algorithms can be transformed into a randomized \fpt \ algorithms that are correct with constant probability. See~\Cref{rem: bounded guess} for more details.
\end{remark}
\subsection{Non-Uniform Capacitied \sor}\label{ss: overview nuni}
For ease of exposition, in this technical overview we highlight our ideas for cluster capacities (see paragraph \ref{para:clusterCapacities}).
Transitioning to node capacities introduces several more challenges which we will not discuss here. 
Other than the basic \fpt~framework mentioned in the preliminaries above, our algorithm is fundamentally different from these of~\cite{bandyapadhyay_et_al:LIPIcs.SoCG.2023.12,jaiswal_et_al:LIPIcs.ITCS.2024.65}, and hence we do not attempt to compare them (in contrast, for uniform capacities the algorithms are similar).

Consider the following simple and natural strategy of processing clusters in $\cC$ iteratively (see~\Cref{algo:nucsr} for pseudo-code). 
Initially all the clusters in $\cC$ are unprocessed (colored red).
We will process the clusters in $\cC$ in non-decreasing order of their radii, denoted as $C_1, \cdots, C_k$. Let $U_1, \cdots, U_k$ be the corresponding cluster capacities. 
Consider the first iteration when we process $C_1$. Suppose we could find a dense ball $D_1 = \ball(y_1, r_1)$ of radius $r_1$ in $P$ such that $|D_1| \ge |C_1|$.
Then, we could create a cluster $C'_1 \subseteq D_1$ of size $|C_1|$\footnote{Even though $|C_1|$ is unknown, in retrospect, we will be able to use~\Cref{rem: point assgn} to obtain such a clustering.\label{foot: cluster D_i}} with radius $r_1$, matching the optimal cost. 
To make this cluster permanent, we delete the points from $D_1$ to prevent them from being reassigned in later iterations. 
However, this could create problems --- (i) $C'_1$ may contain points from other clusters, so in the future iterations when we process these clusters then the densest ball may not have enough points, and (ii) we need to make sure the points of $C_1$ are taken care by some cluster in our solution. 
Our algorithm is based on the following two key ideas that handle these two issues:

\textbf{Invariant: do not touch the unprocessed.} We ensure that, throughout the algorithm, every unprocessed cluster in $\cC$ has all of its points intact. In other words, during the  processing iteration of $C_i$, all its points are present. 

\textbf{Making progress: Good dense balls.} An immediate implication of the above invariant is that when we process $C_i$, there is a ball of radius $r_i$ containing at least $|C_i|$ points. 
Let $D_i = \ball(y_i, r_i)$ be the densest ball w.r.t. the remaining points during the iteration of $C_i$. 
Note that such a ball contains at least $|C_i|$ points due to  the above invariant.
Consider the smallest radius cluster $C_j \in \cC$ that intersects%
\footnote{In fact, we use a weaker notion of a \textit{reachable} set (see \Cref{def:reach}).} $D_i$, and call $C_j$ the \emph{anchor} for $D_i$. See~\Cref{fig:sfig21}, where $D_i$ is the ball centered at $y$ with radius $r_i$, $C_i=C_5$ and $C_j=C_8$.
Now, consider the extended ball $D'_i := \ball(y_i, r_i + 2r_j)$, and note that $C_j \subseteq D'_i$.  
Now, if $C_i$ intersects $D'_i$, then  $C_i$ is not far from $y_i$ (see~\Cref{fig:sfig22}). 
In this case, we call $D_i$ a good dense ball (see~\Cref{def:good dense ball}), and will be able to process $C_i$. 
Thus, our first goal when processing the cluster $C_i$ is to find a good dense ball (see~\Cref{algo:settle cluster}). 

\textit{Finding good dense ball in \fpt\ time. }
Suppose $D'_i$ does not intersect $C_i$, then we can temporarily delete $D'_i = \ball(y_i, r_i + 2r_j)$ since this ball does not contain any point of $C_i$ (see~\Cref{fig:sfig21}, where $C_i=C_5$ and $C_j = C_8$). 
Furthermore, we end up deleting at least one cluster (specifically, $C_j$), so this process of temporarily deleting balls can repeat at most $k$ times before $C_i$ intersects the extended ball $D'_i$ (see~\Cref{algo:good dense ball}). 
Once we have a good dense ball $D_i = \ball(y_i, r_i)$ for $C_i$, consider the extended ball $D'_i := \ball(y_i, r_i + 2r_j)$, where $r_j$ is the radius of the anchor $C_j$ for $D_i$. There are two cases:

\textit{Case $1$: Hop through the anchors. } Suppose we are lucky, and it turns out that $r_j \le r_i$ (\Cref{algo:settlecluster: type1} of ~\Cref{algo:settle cluster}). Then, $y_i$ can serve all the points of $C_i$ within radius $3r_i + 2r_j \le 5r_i$. See~\Cref{fig:sfig22}, where $C_j = C_3$. We now mark $C_i$ as processed (color black).
Note that, in this case, we do not delete any points.

\textit{Case $2$: Otherwise ``Exchange". } Suppose we are not lucky, and it turns out that $r_j > r_i$ (see~\cref{fig:sfig23} where $r_i=r_5 > r_8=r_j$). In this case, our basic approach is to open the cluster $C_i$ at $D_i$, and thus relieving the clusters intersecting $D_i$ from their responsibility for these points (see~\Cref{algo:exchange cluster}). In exchange, these clusters will become responsible for $C_i$ points. The crux of the argument is that as the anchor $C_j$ has the minimal radius among the clusters intersecting $D_i$, each such cluster $C_q$ has radius larger than $r_i$. Thus from $C_q$'s perspective, the points of $C_i$ are ``nearby''. Hence $C_q$ can accept responsibility for a number of $C_i$ points proportional to the number of $C_q$ points taken (and thus everybody is taken care of).

In more details: we create a new cluster $C'_i \subseteq D_i$ out of the points in $D_i$ (recall $|D_i|\ge|C_i|$ and see~\cref{foot: cluster D_i}). 
To make this assignment permanent, we delete $D_i$, and mark $C_i$ as processed (color black).
However, we  end up taking points from other clusters (and also end up deleting points from these clusters). 
The key observation is that as we process the clusters in non-decreasing order and $r_j > r_i$, $D_i$ does not contain points from clusters processed via Case $1$.
Since we maintain the invariant that the points of unprocessed clusters are not deleted, we have to mark the clusters intersecting $D_i$ as processed (color black). 
However, before marking these clusters as processed, we need to ensure that they are accounted for in our solution. 
To this end, we use a novel idea of exchanging points. 
Since we have created a new cluster $C'_i$ out of $D_i$, any cluster $C_{j'}$ that has lost, say, $n_{j'} > 0$ points in this process, can instead claim back $n_{j'}$ points from the original cluster $C_i$ (depicted by islands in $C_5$ in~\Cref{fig:sfig23}) by paying slightly more cost since, in this case, we have $r_i < r_j \le r_{j'}$. 
Here, we use the fact that $C_j$ is the anchor for $D_i$, and hence $r_j$ is the smallest radius intersecting $D_i$. 
We call such clusters partitioned clusters. Specifically, the radius of the partitioned $C_{j'}$ is at most $\delta(o_{j'}, y_i) + \delta(y_i, p) \le (r_{j'} + r_i) + (r_i + 2r_j + 2r_i) \le 7r_{j'}$, for $p \in C_i$, as $r_j \le r_{j'}$. 
See~\Cref{fig:sfig23} for an illustration. Note that a cluster $C_{j'}$ can be partitioned multiple times by different $C_i$'s (we track the clusters that partitioned $C_{j'}$ by $\pof(C_{j'})$), as shown in~\Cref{fig:sfig24}. 
However, $o_{j'}$ can still serve all the points of the modified $C_{j'}$ within cost $7r_{j'}$, since, in this case, the radii of the clusters partitioning $C_{j'}$ are strictly less than $r_{j'}$. 
While it is hard to find $o_{j'}$, in~\Cref{sec: nonuni}, we show how to find another point that can serve all the points of the partitioned cluster $C_{j'}$ within radius $7r_{j'}$ (see \textbf{for} loop at~\Cref{algo:nucsr for} of~\Cref{algo:nucsr}).

Note that, we get different approximation factors in Case $1$ and Case $2$.
In the technical section, we interpolate between these two factors to obtain an improved approximation factor for the algorithm.

\subsection{Uniform Capacities}\label{ss: overview uni}
Our algorithm builds upon the algorithm of~\cite{jaiswal_et_al:LIPIcs.ITCS.2024.65}.
The algorithm divides the clusters in $\cC$ into two categories: heavy and light.
A cluster $C \in \cC$ is \emph{heavy} if $|C| > \nicefrac{U}{k^3}$; otherwise, it is a \emph{light} cluster.
The intuition behind this partition is the following: some heavy cluster must exist, and given such heavy cluster $C \in \cC$, a randomly selected point from $P$ belongs to $C$ with probability at least $\poly(1/k)$. 
Therefore, with probability at least $k^{-O(k)}$, we can obtain a set $X \subseteq P$ containing a single point $x_i$ from every heavy cluster $C_i\in\cC$.
For every heavy cluster $C_i$, we wish to open a cluster $B_i=\ball(x_i,2r_i)$ that will take care of $C_i$ points (note $C_i\subseteq B_i$).
We will assume that the union of these balls, $\cup_{x_i \in X}B_i$, covers the point set $P$ (this assumption catches the essence of the problem, as we can greedily take care of light clusters not covered by this union). 
Let $\cL \subseteq P$ be the set of points corresponding to the light clusters in $\cC$. For each light cluster $C_t \in \cC$, arbitrarily assign it to a heavy cluster $C_i \in \cC$ such that $C_t$ intersects $B_i$.
Then, for each heavy cluster $C_i$, consider the set $S_i : = C_i \cup (\ball(x_i,2r_i + 2r_{\ell(i)}) \cap \cL)$, where $r_{\ell(i)}$ is the maximum radius of a light cluster assigned to $C_i$. 
Note that $S_i$ contains $C_i$, and all the light clusters assigned to $C_i$.
If $|S_i| \le U$, then we can open cluster $C'_i= S_i$ without violating the uniform capacity. 
For the other case, note that $|S_i| \le (1+\nicefrac{1}{k^2})U$ since there can be at most $k$ light clusters assigned to $C_i$, each of which has at most $U/k^3$ points. 
Let $I \subseteq [k]$ be the indices of heavy clusters, and let $I' \subseteq I$ be the indices of heavy clusters for whom $S_i$ violated the capacity constraint.
\cite{jaiswal_et_al:LIPIcs.ITCS.2024.65} used a very neat matching based argument to construct $|I'|$ balls $\{\hat{C}_i\}_{i\in I'}$ each with a unique radius from $\{r_i\}_{i\in I}$ (the heavy clusters) such that $\cup_{i\in I'}\hat{C}_i$ contains at least $\nicefrac{U}{k^2}$ points from each $S_i$ (for $i\in I'$).
These clusters $\{\hat{C}_i\}_{i\in I'}$ are then used to unload enough points from the overloaded heavy clusters. 
Now, $\{\ball(x_i,2r_i + 2r_{\ell(i)})\}_{i\in I}$ together with $\{\hat{C}_i\}_{i\in I'}$ is a valid feasible solution.
The $\ell_1$ norm cost (\csr cost) of this solution is bounded by
$
\sum_{i \in I} (2r_i + 2r_{\ell(i)}) + \sum_{j \in I} r_j \le 3 \sum_{i \in [k]} r_i \le (3+\eps)\opt,
$
where we used that fact that $C_{\ell(i)}$ is distinct for each $C_i$, and that in $\cup_{i\in I'}\hat{C}_i$ we used the radius of every heavy cluster at most once.
From the other hand, the $\ell_\infty$ cost (\ckcen) can only be bounded by $(4+\eps)\opt$. It is also possible to bound the approximation factor for general $\ell_p$ norm objective.

\begin{wrapfigure}{r}{0.25\textwidth}
	\begin{center}
		\vspace{-20pt}
		\includegraphics[width=0.24\textwidth]{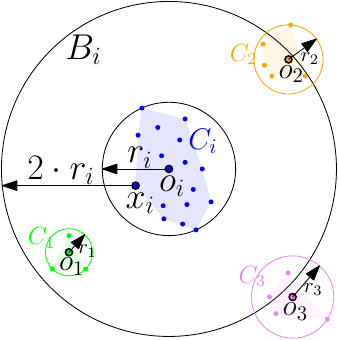}
		\vspace{-5pt}
	\end{center}
	\vspace{-15pt}
\end{wrapfigure}
\paragraph*{Improvement for $\ell_p$ norm objectives.}
Our modification to get a better factor is very small, but leads to a significant improvement.
Specifically, we fine tune the definition of $C_{\ell(i)}$ for $C_i$. 
\cite{jaiswal_et_al:LIPIcs.ITCS.2024.65} assigned each light cluster $C_j$ to an arbitrary heavy cluster $C_i$ such that $C_j$ intersects the extended ball $B_i$, and let $r_{\ell(i)}$ to be the maximum radius of a light cluster assigned to $C_i$.
Instead, we assign each light cluster $C_j$ to an arbitrary  heavy cluster $C_i$ such that $B_i$ contains its center $o_i$, while $r_{\ell(i)}$ remains the maximum radius of a light cluster assigned to $C_i$. 
The crux is that now the ball $\ball(x_i,2r_i+r_{\ell(i)})$ contains the heavy cluster $C_i$ and all the light clusters assigned to it.
In the example on the right we consider a heavy cluster $C_i$, where we sampled a point $x_i\in C_i$ and consider the ball $B_i=\ball(x_i,2r_i)$. $B_i$ intersects $3$ light clusters ($C_1,C_2,C_3$) with radii $r_1<r_2<r_3$. Now $r_{\ell(i)} = r_2$ because $o_2\in B_i$, while $o_3\notin B_i$.
As previously, we open clusters of two types: 
(1) for every heavy cluster $C_i$, a cluster $C'_i$ centered at $x_i$ of radius $2r_i+r_{\ell(i)}$,
and (2) the clusters $\{\hat{C}\}_{i\in I'}$ relieving the extra load.
Overall, we save a factor of $r_{\ell(i)}$ in the radius of the cluster centered at $x_i$. 
Note that this solution yields a $3$ approximation w.r.t. the $\ell_\infty$ norm objective (\ckcen), 
and yields a similar improvement w.r.t. other $\ell_p$ norm objectives.

\subsection{Hardness of \fpt-Approximation for Uniform Capacities} \label{ss: overview hardness}
Our hardness result is based on a polynomial time reduction from \textsc{Max $k$-Coverage} (\maxcov), where, given a universe $\cG$ of elements, a positive integer $k$, and $k$ collections $\{\cC_1,\dots,\cC_k\}$ of subsets of $G$, the task is to find $\{S_1,\dots,S_k\}$ such that $S_i\in \cC_i, i\in [k]$ that maximizes the number of elements covered. It is known that \maxcov is NP-hard to approximate to a factor better than $(1-1/e)$, even when every set in the collection has precisely $|\cG|/k$ elements. Under \gapeth, there is no $\fpt(k)$ algorithm for \maxcov that approximates to a factor better than $(1-1/e)$, even when every set in the collection has precisely $|\cG|/k$ elements

On a high level, given an instance of \maxcov, we create the set-element incidence graph  of the instance. We then add edges between the vertices corresponding to the sets of the same collection. Let $G=(V,E)$ be the resulting graph.
The points of our \csr instance corresponds to the vertices in $V$, and the distance metric is given by the shortest path distance in $G$. Finally, we set the uniform capacity $U=|\cG|/k+m$. It is easy to see that when there are $k$ sets $\{S_1,\dots,S_k\}$ such that $S_i\in \cC_i, i\in [k]$ that cover $\cG$, then we can find a solution for \csr with cost $k$, by selecting the points corresponding to the sets $\{S_1,\dots,S_k\}$ as the centers. For the No case, observe that any solution to \csr must have clusters of size precisely $U$ due to our construction. This simple observation forces the centers to belong to different collections, yielding the result.

\begin{figure}
\begin{subfigure}{.5\textwidth}
  \centering
  \includegraphics[width=\linewidth]{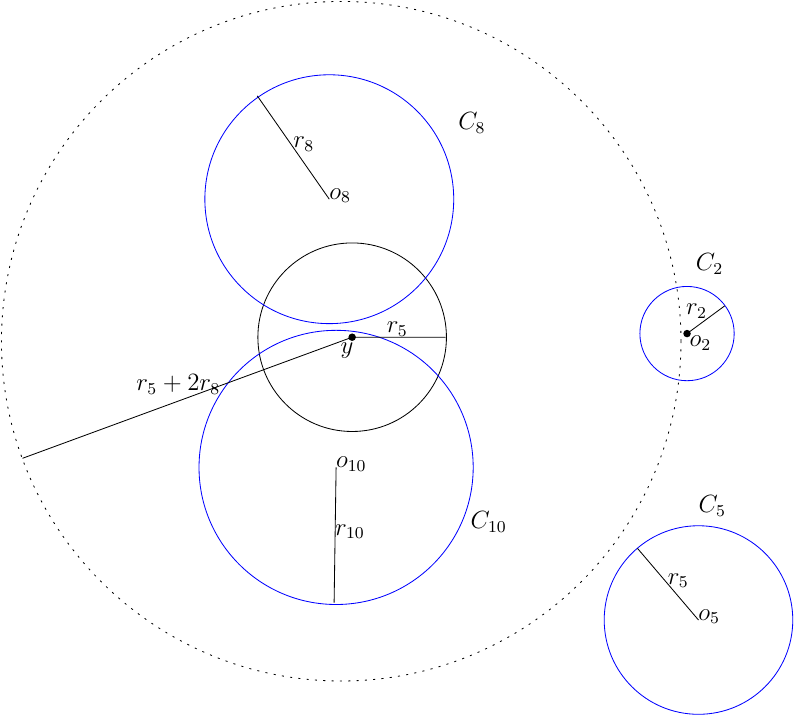}
  \caption{}
    \label{fig:sfig21}
\end{subfigure}%
\begin{subfigure}{.5\textwidth}
  \centering
  \includegraphics[width=.8\linewidth]{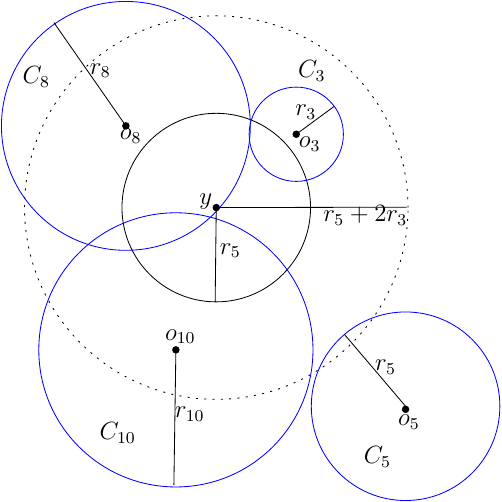}
  \caption{}
  \label{fig:sfig22}
\end{subfigure}
\begin{subfigure}{.5\textwidth}
  \centering
  \includegraphics[width=.9\linewidth]{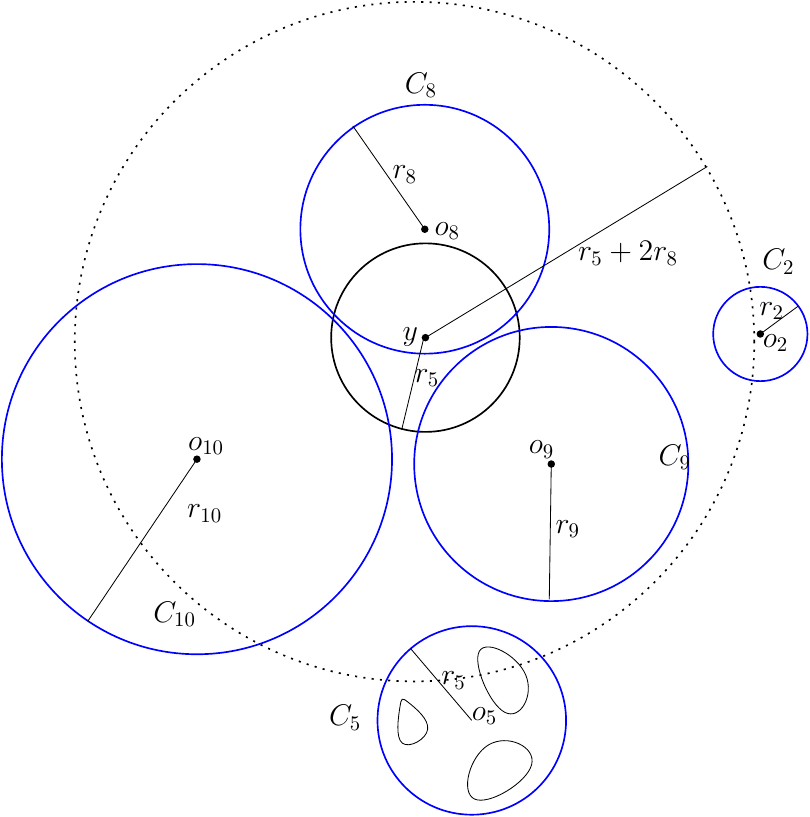}
  \caption{}
    \label{fig:sfig23}
\end{subfigure}
\begin{subfigure}{.5\textwidth}
  \centering
  \includegraphics[width=1.1\linewidth]{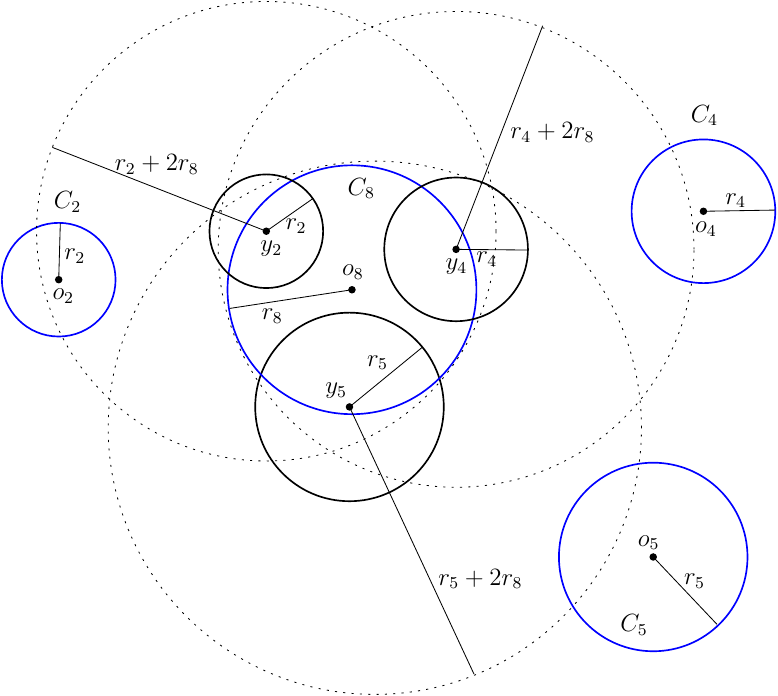}
  \caption{}
    \label{fig:sfig24}
\end{subfigure}
\caption{\footnotesize \Cref{fig:sfig21} shows a dense ball $\ball(y,r_5)$, colored in black, that is not good for cluster $C_5$, because $\ball(y,r_5+2r_8)$ is disjoint from $C_5$.
Hence, it is safe to (temporary) delete $\ball(y,r_5+2r_8)$ without affecting cluster $C_5$, and recursively search for good dense ball for $C_5$ in the remaining points. 
On the other hand, the $\ball(y,r_5)$ of~\Cref{fig:sfig22} is good dense ball for $C_5$.
Furthermore, note that in this case the radius of the anchor $r_3 \le r_5$, and hence $y$ can serve all the points of $C_5$ within distance $3r_5+2r_3 \le 5r_5$. 
In~\Cref{fig:sfig23}, the radius of the anchor is $r_8 > r_5$, and hence the previous argument fails. 
However, in this case, since $\ball(y_,r_5)$ is a good dense ball for $C_5$, so we can open a new cluster $C'_5 \subseteq \ball(y,r_5)$ of appropriate size, instead of serving the points of $C_5$. 
Note that, in this case, the new cluster $C'_5$ ends up taking points from clusters $C_8, C_9$ and $C_{10}$. 
However, these clusters can claim back the lost points from the original cluster $C_5$ (depicted by islands in $C_5$), which is not too far from them. 
In this case, we say that  $C_8, C_9$ and $C_{10}$ are partitioned by $C'_5$. 
\Cref{fig:sfig24} shows cluster $C_8$ partitioned by $C_2,C_4$ and $C_5$, and it has to collect its lost points from the respective clusters.
}
\label{fig:fig2}
\end{figure}

\section{Further Related Work}\label{subsec:related}
We begin with further background on the (uncapacitated) \sor and \sod problems. 
When the point set $P$ is from either $\ell_1$ or $\ell_2$ there is an exact poly-time algorithm \cite{10.5555/1347082.1347172}, while there is also an EPTAS for the Euclidean ($\ell_2$) case \cite{10.5555/1347082.1347172}.
Another interesting case is when the metric between the input points $P$ corresponds to the shortest path metric of an unweighted graph $G$. Here, assuming the optimal solution does not contain singeltons (clusters with only a single point), there is a polynomial time algorithm finding optimal solution \cite{BehsazS15}.
The \sod with fixed $k$ admits exact polytime algorithm \cite{BehsazS15} (see also \cite{HJ87,CRW91}). Further, there is a PTAS for \sod where the input points are from the Euclidean plane \cite{BehsazS15}.
The \sor problem was also studied in the fault tolerant regime \cite{BhowmickIV21}, considering outliers \cite{AhmadianS16}, and in an online (competitive analysis) setting \cite{FotakisK14,CEIL13}.

\textbf{\ckcen} is a special case of capacitated $\ell_p$ norm of radii when $p=\infty$. The uniform \ckcen problem has been explicitly studied by Goyal and Jaiswal~\cite{GOYAL2023190}, and is also implicitly captured by the results of Jaiswal \etal \cite{jaiswal_et_al:LIPIcs.ITCS.2024.65}. Goyal and Jaiswal design a factor $2$ \fpt-approximation for the soft assignment setting, where multiple centers can be opened at a single point, enabling it to serve more than $U$ points. For instance, the algorithm can opt to open two centers at a point $p \in P$, thereby allowing up to $2U$ points to be served from location $p$.
In contrast, similar to this paper, Jaiswal \etal focus on the hard assignment scenario, where at most one center can be opened at a given point, and they present a $(4+\epsilon)$ \fpt-approximation for this case.
We are not aware of any approximation preserving subroutines that transform a soft assignment algorithm to a hard assignment algorithm. However, by replacing the chosen center with any other point in the cluster one can obtain a solution with hard assignments while incurring a factor $2$ multiplicative loss in approximation. 
Consequently, the algorithm of~\cite{GOYAL2023190} can only yield $4$ \fpt-approximation using this approach, which matches the factor of~\cite{jaiswal_et_al:LIPIcs.ITCS.2024.65}.

\textbf{\ckdia} is a special case of capacitated $\ell_p$ norm of diameters problem when $p=\infty$.
Its  uncapacitated version turns out to be computationally very challenging.  While there is a simple $2$-approximation for \kdia~\cite{GONZALEZ1985293}, it is NP-hard to approximate the problem better than factor $2$ in $\ell_\infty$ metric  even for $k=3$~\cite{MEGIDDO1990327}. Very recently, Fleischmann~\cite{fleischmann2024inapproximabilitymaximumdiameterclustering} showed that, even for $k=3$, the problem is NP-hard to approximate better than factors $1.5$ and $1.304$ when the points come from $\ell_1$ and $\ell_2$ metrics, respectively. Given that these hardness results hold for constant $k$ (i.e., even for $k=3$), they also translate to \fpt\ lower bounds. Thus, these hardness results imply that no \fpt\ algorithm can beat these lower bounds for \kdia, unless $\pp = \np$. Of course, the corresponding capacitated version of \kdia is a generalization of \kdia, and hence these \fpt-hardness of approximations also hold for the capacitated version.
For upper bounds, it is interesting to note that the reduction from \sor to \sod mentioned in \cref{foot:SORroSOD} works even when solving \sor w.r.t. soft assignments. Thus using \cite{GOYAL2023190} one can obtain a $4$ \fpt-approximation for the uniform  \ckdia problem, matching our \Cref{thm: uni dia} for this case.

\addtocontents{toc}{\protect\setcounter{tocdepth}{1}}
\section{Preliminaries}\label{sec:prelims}

\subsection{Problems}
In this section, we define the problems of interest. We start with defining capacitated versions of \sor.
\begin{definition}[Uniform \csrfull]
    Given a set $P$ of $n$ points in a metric space with distance $\delta$ and non-negative integers $U$ and $k$, the \emph{Uniform \csrfull (Uniform \csr)} problem asks to find a set of centers $S \subseteq P$ of size $k$ and an assignment function $\sigma : P \rightarrow X$ such that  $|\sigma^{-1}(s)| \le U$ for all $s \in S$ and  the objective $\sum_{s \in S} \rad(s)$ is minimized, where $\rad(s) := \max_{p \in \sigma^{-1}(s)} \delta(s,p)$ is the radius of cluster centered at $s \in S$.
\end{definition}

\begin{definition}[Non-Uniform Node \csr]
    Given a set $P$ of $n$ points in a metric space with distance $\delta$, a non-negative integer $U_p$ for every $p \in P$ and a non-negative integer $k$, the \emph{Non-Uniform Node \csr} problem asks to find a set of centers $S \subseteq P$ of size $k$ and an assignment function $\sigma : P \rightarrow S$ such that  $|\sigma^{-1}(s)| \le U_{s}$ for all $s \in S$ and the objective
    $\sum_{s \in S} \rad(s)$ is minimized, where $\rad(s) := \max_{p \in \sigma^{-1}(s)} \delta(s,p)$.
\end{definition}

\begin{definition}[Non-Uniform Cluster \csr]
    Given a set $P$ of $n$ points in a metric space with distance $\delta$, a set of cluster capacities $\{U_1,\cdots, U_k\}$ and a non-negative integer $k$, the \emph{Non-Uniform Cluster \csr} problem asks to find a set of centers $S = \{s_1,\cdots, s_k\} \subseteq P$ of size $k$ and an assignment function $\sigma : P \rightarrow S$ such that  $|\sigma^{-1}(s_i)| \le U_{i}$ for all $i \in [k]$ and the objective
    $\sum_{i \in [k]} \rad(s_i)$ is minimized, where $\rad(s_i) := \max_{p \in \sigma^{-1}(s_i)} \delta(s_i,p)$.
\end{definition}
Next, we define capacitated version of \sod.

\begin{definition}[Uniform \csdfull]
    Given a set $P$ of $n$ points in a metric space with distance $\delta$
     and non-negative integers $U$ and $k$, the \emph{Uniform \csdfull (Uniform \csd)} problem asks to find $k$ partitions (clusters) $C_1,\cdots,C_k$ of $P$ such that  $|C_i| \le U$ for all $i \in [k]$ and the objective $\sum_{i \in [k]} \diam(C_i)$ is minimized, where $\diam(C_i) := \max_{p,q \in C_i} \delta(p,q)$, is the diameter of cluster $C_i$.
\end{definition}

\begin{definition}[Non-Uniform \csd]
    Given a set $P$ of $n$ points in a metric space with distance $\delta$,  a set of cluster capacities $\{U_1,\cdots, U_k\}$ and a positive integer $k$, the \emph{Non-Uniform \csd} problem asks to find $k$ partitions (clusters) $C_1,\cdots,C_k$ of $P$ such that  $|C_i| \le U_{i}$ for all $i \in [k]$ and the objective
    $\sum_{i \in [k]} \diam(C_i)$ is minimized, where $\diam(C_i) := \max_{p,q \in C_i} \delta(p,q)$.
\end{definition}

\subsection{Norm objectives}
\begin{wrapfigure}{r}{0.18\textwidth}
	\begin{center}
		\vspace{-20pt}
		\includegraphics[width=0.17\textwidth]{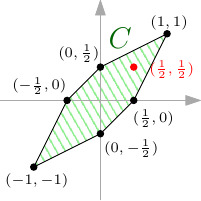}
		\vspace{-5pt}
	\end{center}
	\vspace{-15pt}
\end{wrapfigure}
A norm $\|\cdot\|:\R^k\rightarrow\R_{\ge0}$ is monotone if for every two vectors $(x_1,\dots,x_k),(y_1,\dots,y_k)$ such that $\forall i,~|x_i|\le|y_i|$, it holds that $\|(x_1,\dots,x_k)\|\le \|(y_1,\dots,y_k)\|$. 
A norm $\|\cdot\|$ is symmetric if for every permutation $\pi:[k]\rightarrow[k]$ it holds that $\|(x_1,\dots,x_k)\|=\|(x_{\pi(1)},\dots,x_{\pi(k)})\|$. Interestingly, as pointed out by \cite{chakrabarty2019approximation}, a symmetric norm is not necessarily monotone. Indeed, consider the set $C \subseteq \mathbb{R}^2$, which is the convex hull of the points $\{(1,1), (-1,-1), (0,0.5), (0.5,0), (0,-0.5), (-0.5,0)\}$ (illustrated on the right), and define $\|x\|_C$ to be the smallest $\lambda$ such that $x/\lambda \in C$. Then $\|\cdot\|_C$ is a symmetric norm over $\mathbb{R}^2$, and $\|(0, 0.5)\|_C=1$, while  $\|(0.5, 0.5)\|_C=0.5$, and thus $\|\cdot\|_C$ is not monotone.

In this paper, we generalize the above problems to objective that is a monotone symmetric norm of radii and diameters.

\subsection{Preliminary claims}
In this section, we first define feasible solutions for \csr and \csd. Then, we show how to recover a feasible solution from a set of centers and radii (diameters). Finally, we show how to guess a close approximation of optimal radii and diameters in \fpt \ time.

\subsubsection{Recovering feasible clusters}
We show how to recover feasible clusters given a set of centers and radii (diameters). For simplicity, we work with sum objective, but the same ideas extend to any monotone symmetric norm objective.
\begin{definition}[Feasible solution for \csr]
    Let $\cI=((P,\delta),k,\{U_p\}_{p \in P})$ be an instance of \csr. A \emph{solution} to $\cI$ is given by a pair $(S,\sigma)$, where  $S \subseteq P$ and $\sigma: P \rightarrow S$. Furthermore, we say solution $(S, \sigma)$ is \emph{feasible} for $\cI$ if $|S| \le k$ and $|\sigma^{-1}(s)| \le U_s$, for every $s \in S$. Moreover, the \emph{cost} of a feasible solution  $(S, \sigma)$ is given by $\cost((S, \sigma)) = \sum_{s \in S} \rad(s)$. For a monotone norm objective, $\cost((S, \sigma))$ denotes the norm of the radii.
\end{definition}

\begin{definition}
    Let  $\cI=((P,\delta),k,\{U_p\}_{p \in P})$ be an instance of \csr and let $k' \le k$. Let $S = \{s_1,\cdots, s_{k'}\} \subseteq P$ and let $\{\gamma_1,\cdots, \gamma_{k'}\}$ such that $\gamma_i \in \mathbb{R}_{\ge 0}$, for $i \in [k']$. We say $\cS = \{ (s_1,\gamma_1),\cdots, (s_{k'}, \gamma_{k'})\}$ is \emph{feasible for $\cI$} 
    if there exists an assignment $\sigma: P \rightarrow S$ such that  for all $s_i \in S$ it holds that $|\sigma^{-1} (s_i)| \le U_{s_i}$ and $\rad(s_i)  \le \gamma_i$. Furthermore,  the \emph{cost} of $\cS$ is $\cost(\cS) = \sum_{i \in [k']} \gamma_i$.
\end{definition}

The following lemma states that, given a set of centers and their corresponding radii, we can find a feasible assignment of points to centers that has the same cost without violating the capacity constraints.

\begin{lemma}\label{thm: feasibility by matching}
     Let  $\cI=((P,\delta),k,\{U_p\}_{p \in P})$ be an instance of \csr, $k' \le k$, and let  $\cS \subseteq P \times \mathbb{R}_{\ge 0}, |\cS| = k'$.
     Then,  the following can be done in time $\poly(n)$. We can verify whether or not $\cS$ is feasible for $\cI$, and if  $\cS$ is feasible for $\cI$, then we can  find a feasible solution $(S,\sigma)$ to $\cI$ such that  $\cost((S,\sigma)) \le \cost(\cS)$.
\end{lemma}
\begin{proof}
    Let $\cS=  \{ (s_1,\gamma_1),\cdots, (s_{k'}, \gamma_{k'})\}$, where for $i\in[k'], s_i \in P$ and $\gamma_i \in \mathbb{R}_{\ge 0}$ and hence, $\beta = \cost(\cS) = \sum_{i \in [k']} \gamma_i$. Denote by $S = \{s_1,\cdots, s_{k'}\}$. We create a bipartite  graph $G= (A \cup B, E)$, with left partition $A$ that contains $U_{s_i}$ many vertices $\{s^1_i,\cdots, s^{U_{s_i}}_i\}$ for each $s_i \in S$ and right partition $B=P$. 
    Then, for every $p \in B$ and $s_i \in S$, we add edges from $p$ to $s^1_i,\cdots s^{U_{s_i}}_i$ if $\delta(p,s_i) \le \gamma_i$.
    Now, note that there is an assignment $\sigma: P \rightarrow S$ such that $(S,\sigma)$ is feasible for $\cI$ if and only if there is a matching on size $|P|$. To see this, suppose there is a matching $\cM$ of size $|P|$, we define $\sigma(p) = s_i$ for $p \in B=P$ such that $(p,s^j_i) \in \cM$ for some $s^j_i \in A$. Since $s_i$ has $U_{s_i}$ many vertices in $A$, we have $|\sigma^{-1}(s_i)| \le U_{s_i}$, and hence $(S,\sigma)$ is feasible for $\cI$.
    For the other direction, suppose there exists an assignment $\sigma: P \rightarrow S$ such that $(S,\sigma)$ is feasible for $\cI$. Then, since, for every $s_i\in S$, we have $|\sigma^{-1}(s_i)| \le U_{s_i}$, we can match every $p \in \sigma^{-1}(s_i)$ to a unique copy $s^j_i$ to obtain a matching of size $|P|$. Finally, we can find maximum matching in $G$ in $\poly(n)$ time using the algorithm of Hopcroft and Karp~\cite{HopcroftKarp}.
    The cost of solution $(S,\sigma)$ thus obtained is $\cost((S,\sigma))  \le \sum_{i \in [k']}  \gamma_i= \cost(\cS)$, as desired.
\end{proof}
In fact, instead of solving a matching problem, we can solve a flow problem and use the recent result of~\cite{9996881} to improve the running time to almost-linear time, specifically  $(nk)^{1+o(1)}$.
\begin{remark}\label{rem: matching cluster cap}
    \Cref{thm: feasibility by matching} is stated for node capacities. However, it can be easily adapted to cluster capacities if we are given the mapping of the elements of $\cS$ and the cluster capacities.
\end{remark}

\begin{definition}[Feasible solution for \csd]\label{def: feasible sol csd}
    Let $\cI=((P,\delta),k,\allowbreak\{U_i\}_{i \in [k]})$ be an instance of \csd. A \emph{solution} to $\cI$ is given by a partition $\cC$ of $P$. We say that $\cC$ is \emph{feasible} for $\cI$ if $|\cC| \le k$ and $|C_i| \le U_i$ for $C_i \in \cC$.
      Moreover, the \emph{cost} of a feasible solution  $\cC$ is given by $\dcost(\cC) = \sum_{C_i \in \cC} \diam(C_i)$, where $\diam(C_i) := \max_{p, p' \in C_i} \delta(p,p')$, is the diameter of cluster $C_i \in \cC$. Furthermore, when the objective is a monotone norm of diameters, $\dcost(\cC)$ denotes the corresponding norm of diameters.
\end{definition}

\begin{definition}[Feasible set for \csd]\label{def: feasible set csd}
     Let $\cI=((P,\delta),k,\{U_i\}_{i \in [k]})$ be an instance of \csd and let $k' \le k$. 
      Let $\cC' = \{C'_1,\cdots, C'_{k'}\}$ such that $C'_i \subseteq P, i \in [k']$. Note that the sets in $\cC'$ are not necessarily disjoint. We say $\cC'$     is \emph{feasible set for $\cI$} if there exists a feasible solution $\tilde{\cC} = \{\tilde{C}_1,\cdots, \tilde{C}_{k'}\}$ to $\cI$ such that $\tilde{C}_i \subseteq C'_i$, for $i \in [k']$. Furthermore,  the \emph{cost} of $\cC'$ is $\dcost(\cC') = \sum_{i \in [k']} \diam(C'_i)$.
\end{definition}

Similar to \csr,  given a collection of subsets of points covering $P$, we can find a feasible clustering  whose sum of diameters is bounded by the sum of diameters of the subsets.

\begin{lemma}\label{thm: feasibility by matching dia}
     Let  $\cI=((P,\delta),k,\{U_i\}_{i \in [k]})$ be an instance of \csd and $k' \le k$.
     Let $\cC' = \{C'_1,\cdots, C'_{k'}\}$ such that $C'_i \subseteq P, i \in [k']$. 
     Then,  the following can be done in time $\poly(n)$. We can verify whether or not $\cC'$ is a feasible set for $\cI$, and if  $\cC'$ is a feasible set for $\cI$, then we can  find a feasible solution $\tilde{\cC}$ to $\cI$ such that  $\dcost(\tilde{\cC}) \le \dcost(\cC')$.
\end{lemma}
\begin{proof}
    Then, we create a bipartite  graph $G= (A \cup B, E)$, with left partition $A$ that contains $U_{i}$ many vertices $\{c^1_i,\cdots, c^{U_{i}}_i\}$ for each $i \in [k']$ and right partition $B=P$. 
    Then, for every $p \in C'_i$, we add edges from $p \in B$ to $c^1_i,\cdots c^{U_{i}}_i$.
    Now, note that $\cC'$ is feasible for $\cI$ if and only if there is a matching on size $|P|$. To see this, suppose there is a matching $\cM$ of size $|P|$, then for $i \in [k']$, we define, $\tilde{C}_i = \{p \in B \vert \exists c^j_i \in A, (p,c^j_i) \in \cM \}$. It is easy to see that $|\tilde{C}_i| \le U_i$ and and hence $\{\tilde{C}_1,\cdots,\tilde{C}_{k'}\}$ is feasible for $\cI$.
    For the other direction, suppose there exists a feasible solution $\tilde{\cC}=\{\tilde{C}_1,\cdots \tilde{C}_{k'}\}$ for $\cI$, i.e., $\tilde{C}_i \subseteq C'_i$ for $i \in [k']$.
    Then, since, for every $\tilde{C}_i\in \tilde{\cC}$, we have $|\tilde{C}_i| \le U_{i}$, we can match $p \in \tilde{C}_i$ to a unique copy $c^j_i$ to obtain a matching of size $|P|$. Finally, we can find maximum matching in $G$ in $\poly(n)$ time using the algorithm of Hopcroft and Karp~\cite{HopcroftKarp}.
    The cost of solution $\tilde{\cC}$ thus obtained is $\dcost(\tilde{\cC})  \le \sum_{i \in [k']}  \diam(\tilde{C}_i) \le \cost(\cC')$, as desired.
\end{proof}

Note that~\Cref{thm: feasibility by matching} and~\Cref{thm: feasibility by matching dia} also work when the objective is a monotone symmetric norm of radii and diameters, respectively.

\subsubsection{$\eps$-approximation of radii and diameters} \label{ss: eps apx}
In our algorithms, we will assume that we have a close approximations of the optimal radii/diameters. Consider a (unknown) multi-set $T^*=\{t^*_1,\dots,t^*_k\}$ of $k$ non-negative reals such that $t^*_1 \le \dots \le t^*_k$ and the largest entry $t^*_k$ is known. We say that a multi-set $T=\{t_1,\cdots,t_k\}$ of non-negative reals is an \emph{$\eps$-approximation} of $T^*$ if for all $i \in [k]$, it holds that $t_i \geq t^*_i$, and  $\|(t^*_1,\dots,t^*_k)\| \le \| (t_1,\dots,t_k) \| \le (1+\eps) \|(t^*_1,\dots,t^*_k)  \|$, where $\|\cdot\|$ is a monotone symmetric norm.
The following lemma says that we can guess an $\eps$-approximation in $2^{O(k\log(k/\eps))}$ time, when the largest entry $t^*_k$ is known.
\begin{lemma}\label{lem: eps apx}
    Suppose there is an unknown multi-set $T^*=\{t^*_1,\cdots,t^*_k\}$ of non-negative reals such that the largest entry is known. Then, there is an algorithm, that for any $\eps >0$, runs in time $2^{O(k \log (k/\eps))}$ and outputs a list $\mathcal{L}$ of $k$-multi-set non-negative reals of size $2^{O(k \log (k/\eps))}$ such that $\mathcal{L}$ contains an $\eps$-approximation of $T^*$.
\end{lemma}
\begin{proof}
Let $\delta = \eps/k$. Given $t^*_k$, it is known that one compute (see~\cite{bandyapadhyay_et_al:LIPIcs.SoCG.2023.12,jaiswal_et_al:LIPIcs.ITCS.2024.65}), in time $2^{O(k \log (k/\delta))}$, a list $\mathcal{L}$ of $k$-multi-set non-negative reals of size $2^{O(k \log (k/\delta))}$  such that there exists $(t_1,\dots,t_k) \in \mathcal{L}$ with $t^*_i \le t_i \le t^*_i + \delta \cdot t^*_k$, for $i \in [k]$. 
First, note that $\|(t^*_{k},\dots,t^*_{k})\|\le\sum_{i=1}^{k}\|t^*_{k}\cdot e_{i}\|=k\cdot\|t^*_{k}\cdot e_{k}\|\le k\cdot\|(t^*_{1},\dots,t^*_{k})\|$, where the first inequality follows from triangle inequality, the equality is because $\|\cdot\|$ is symmetric, and the last inequality follows due to monotonicity. Hence,
\begin{align*}
\|(t^*_{1},\dots,t^*_{k})\|\le\|(t_{1},\dots,t_{k})\| & \le\|(t^*_{1},\dots,t^*_{k})\|+\delta\cdot\|(t^*_{k},\dots,t^*_{k})\| \quad \text{by triangle inequality}\\
 & \le\|(t^*_{1},\dots,t^*_{k})\|+\delta\cdot k\cdot\|(t^*_{1},\dots,t^*_{k})\|\\
 & =(1+\delta\cdot k)\cdot\|(t^*_{1},\dots,t^*_{k})\|\\
 &=(1+\eps)\cdot\|(t^*_{1},\dots,t^*_{k})\|
\end{align*}

\end{proof}

\subsubsection{Algorithms with bounded guesses.} 
For ease of exposition, we assume that our algorithms can make a bounded guesses. Such algorithms can be transformed into a (randomized) \fpt \ algorithms, as mentioned in the following remark.
\begin{remark}\label{rem: bounded guess}
  We assume that our algorithms have a power to make guesses in constant time. In particular, it can guess an element correctly  with probability $p(k)$ in constant time, for some  computable function $p$.
  It can be easily verified that  an algorithm for a minimization (maximization) problem with such guesses running in  $f(k)n^{O(1)}$ time can be transformed into a randomized algorithm without guesses that runs in time  $h(k)n^{O(1)}$, for some $h$, by replacing each guess either with a procedure that enumerates all elements of the sample space when the size of the sample space   is bounded by a function of $k$ or  the corresponding sampling procedure, and returning a minimum (maximum) cost solution.  

  In our algorithms, we explicitly show subroutines that can replace corresponding guess calls in \fpt \ time.
\end{remark}

\addtocontents{toc}{\protect\setcounter{tocdepth}{2}}
\section{Non-Uniform Capacities}\label{sec: nonuni}

In this section, we provide improved \fpt-approximation algorithms for non-uniform capacities. The main theorem is restated for convenience.

\MainFPTUB*
As mentioned before, the above theorem implies, in a black-box way, a factor $2(3+2\sqrt{2} + \epsilon) \approx 11.656 + \eps$ \fpt-approximation algorithm for Non-uniform \csd.
Interestingly, our algorithmic framework \Cref{thm: nuni csr} can be adapted to Non-uniform \csd to obtain a significantly better factor of $(7+\eps)$ (\Cref{thm: nuni csd} restated bellow for convenience). 

\GeneralCapDiameter*

In the next section, we will present a detailed exposition of the algorithm of~\Cref{thm: nuni csr} for node capacities. In~\Cref{ss: non cluster csr}  and~\Cref{ss: non  csd}, we explain the changes required in the algorithm to adapt it to cluster capacities and \csd, respectively.

\subsection{Node \csr}\label{ss: non node csr}
In this section, we prove \Cref{thm: nuni csr} for node capacities. 
Let $\cI = ((P,\delta),k,\{U_p\}_{p \in P})$
be a given instance of non-uniform node \csr.
Let $\cC=\{C_1,\cdots,C_k\}$ be an optimal clustering of $\cI$ such that $C_i$ has radius $r^*_i$ with center $o_i$ and $r^*_1 \le \cdots \le r^*_k$.
Let $O = \{o_1,\cdots, o_k\}$, and
let $\text{OPT}= r^*_1+\cdots+ r^*_k$ be the  cost of clustering $\cC$. Let $r_1 \le \cdots \le r_k$ be an $\epsilon$-approximation of the optimal radii of $\cC$ (see~\Cref{ss: eps apx}).
 \Cref{algo:nucsr} describes the pseudo-code of our algorithm for \Cref{thm: nuni csr} that is based on the ideas presented in~\Cref{sec: tech}. 
 \Cref{algo:nucsr} makes guesses (using the command \guess), and hence, we make the following assumption for brevity. Later, when we proof~\Cref{thm: nuni csr} in~\Cref{ss: proof of nuni}, we show how to remove this assumption.
\begin{assumption}\label{as:nuni}
    All the guesses of \Cref{algo:nucsr} are correct.
\end{assumption}

As mentioned in~\Cref{sec: tech}, we process the clusters in $\cC$ in non-decreasing order of their radii. The first thing we do when we process a cluster $C_i \in \cC$ (\Cref{algo:settle cluster}) is to find a good dense ball, which we define next. However, we need the following definitions first.
\begin{definition}[reachable set, anchor]\label{def:reach}
    Given a ball $B := \ball(y,r)$, for $y \in P$ and radius $r$, we say that a cluster $C_i \in \cC$ is \emph{reachable} from $B$ if $\delta(y,p) \le r + 2r_i$, for all $p \in C_i$. Let \emph{\rchble}$(B)$ denote the set of all clusters in $\cC$ reachable from $B$. Furthermore, by breaking the ties arbitrarily, the cluster with smallest radius in $\emph{\rchble}(B)$ is called \emph{anchor} for $B$.
\end{definition}

See~\Cref{fig:sfig21} for an illustration, where clusters $C_8$ and $C_{10}$ are reachable from $\ball(y,r_5)$, while $C_2$ and $C_5$ are not. Furthermore, cluster $C_8$ is anchor for $\ball(y,r_5)$.
\begin{definition}\label{def: available point}
    $y \in P$ is said to be \emph{available for $C_i$} if $(i)\ U_y \ge |C_i|$ and  $(ii)\ y \notin O \setminus \{o_i\}$.
\end{definition}

Now we are ready to define good dense ball for $C_i$.
For simplicity, we assume we have a copy $F$ of $P$ from which we pick our solution. Our algorithm sometimes works on subsets of $P$ and $F$. Furthermore, for $P' \subseteq P, y\in P$ and a positive real $r$, let $\ball_{P'}(y,r) = \ball(y,r) \cap P'$.
\begin{definition}[Good Dense ball for $C_i$]\label{def:good dense ball}
For $P',F' \subseteq P$,
we say a ball $D_i= \ball_{P'}(y_i,r_i), y_i \in F'$ is a \emph{good dense ball} for $C_i$ in $(P',F')$ if
\begin{itemize}
    \item $|\ball_{P'}(y_i,r_i)| \ge |C_i|$ and
    \item $y_i$ is available for $C_i$  and
    \item $C_i \in \rchble(\ball(y_i,r_i+2r_j))$, where $r_j$ is the radius of anchor $C_j$ for $D_i$. 
\end{itemize}
\end{definition}
Note that, for every $C_i \in \cC$, we have that $\ball_{P'}(o_i,r_i)$ is a good dense ball for $C_i$ in every $(P',F')$ such that $C_i \subseteq P' \subseteq P$ and $F' \subseteq F$ containing $o_i$. See~\Cref{fig:sfig22} for an illustration, where $\ball(y,r_5)$ is a good dense ball for $C_5$ (assuming $|\ball(y,r_5)| \ge |C_5|$). On the other hand, in~\Cref{fig:sfig21}, the $\ball(y,r_5)$ is not a good dense ball for $C_5$.


\begin{algorithm}[h]
\caption{$(3+2\sqrt{2}+\epsilon)$ \fpt approximation for Non Uniform \csr (node capacities)} \label{algo:nucsr}
  \KwData{Instance $\cI=((P,\delta),k,\{U_p\}_{p \in P})$ of Non Uniform \csr, $\epsilon>0$}
  \KwResult{A $(3+2\sqrt{2} + \epsilon)$ approximate solution}

    $F \leftarrow P, Q \leftarrow P$ \tcp*{$F$ is a set of potential centers}
    $\cS \leftarrow (\emptyset,\emptyset)$\;
    \guess $R=(r_1,\cdots,r_k)$, an $\epsilon$-approximation of optimal radii such that $r_i \le r_{i+1}$ for $1 \le i \le k-1$, and let $(C_1,\dots,C_k)$ be the names of the corresponding clusters in optimal\;
    Let $\alpha \leftarrow 1 + 2\sqrt{2}$\label{algo:nucsr alpha}\tcp*{constant for extended balls}
    Color each cluster $C_i \in \cC$ red\;
    Set $\pof(C_t) = \emptyset$ for $C_t \in \cC$\;
    \While{$\exists$ a red cluster\label{algo:nucsr while}}{
        $C_i \leftarrow$ smallest radius red cluster according to R\;
        \SC$(C_i)$\;
    }

    \ForEach{black cluster $C_t$ such that $\emph{\pof}(C_t) \ne \emptyset$\label{algo:nucsr for}}{
        \guess available $y_t \in \cap_{r_j \in \pof(C_t)} \ball_{F}(y_j,r_j+r_t)$ for $C_t$\tcp*{$U_{y_t} \ge |C_t|$}  \label{algo:nucsr tp3 yt}
        Add $(y_t,(\alpha + 2)r_t)$ to $\cS$\label{algo:nucsr add cen tp3}\;

        Delete $y_t$ from $F$\;
        
    }
    \lIf{$\cS$ is feasible  for $\cI$}{
        \Return $(S,\sigma)$ obtained from \Cref{thm: feasibility by matching}\label{algo: nuni find feasible sol}}
   \textbf{fail}\;
\end{algorithm}

\begin{algorithm}[h]
\caption{\SC$(C_i)$} \label{algo:settle cluster}

    Let $D_i = \ball_Q(y_i,r_i)$ be the ball returned by \FGDB$(r_i,(Q,F))$\label{algo:settle cluster findgoodenseball} \;
    \guess \eIf{$C_i \in \rchble(\ball(y_i,\alpha r_i))$\label{algo:settlecluster: type1}}{
         Color $C_i$ black\label{algo:settlecluster type1 black}\;
         Delete $y_i$ from $F$\label{algo: nuni sc del yi}\;
         Add $(y_i,(\alpha + 2)r_i)$ to $\cS$\label{algo: sc add cen2}\;
    }
    { \texttt{EXCHANGE}$(D_i)$\tcp*{$r_j > \frac{(\alpha-1)r_i}{2} = \sqrt{2}r_i$}}
    
    \Return
\end{algorithm}

\begin{algorithm}[h]
\caption{\FGDB$(r_i,(Q,F))$} \label{algo:good dense ball}

    $F' \leftarrow F$,
    $P' \leftarrow Q$\;
    \For{$2k$ times}{
        Let $D_i = \ball_{P'}(y_i,r_i), y_i \in F'$ be the argument that maximizes the function $\max_{y_j \in F'} \min \{U_{y_j},|\ball_{P'}(y_j,r_i)|\}$\label{algo: nuni fgdb max}\;

        \guess the radius $r_j$ of the anchor of $D_i$ \tcp*{see Definition~\ref{def:reach}}
        \guess \lIf{$C_i \notin \rchble(\ball(y_i,r_i+2r_j))$}{
            $P' \leftarrow P' \setminus \ball(y_i,r_i+2r_j))$\label{algo: nuni fgdb del ball}}
        \Else{
            \guess \lIf{$y_i$ is not available for $C_i$}{
                $F' \leftarrow F' \setminus \{y_i\}$\label{algo: nuni fgdb del cen}}
        \lElse{
         \Return $D_i$}
    }
    }
     \textbf{fail}\;
\end{algorithm}

\begin{algorithm}[h]
\caption{\texttt{EXCHANGE}$(D_i)$} \label{algo:exchange cluster}

    Delete $D_i$ from $Q$\label{algo: nuni exch del di}\; 
    Color $C_i$ black\label{algo:exchangecluster type2 black}\;
    Add $(y_i,r_i)$ to $\cS$ and delete $y_i$ from $F$\label{algo: nuni exch del yi}\;
    \guess $\rchble(D_i)$\; 
    \ForEach{$C_t \in \rchble(D_i)$}{
        Add $r_i$ to \pof$(C_t)$\;
        Color $C_t$ black\;
    }
    \Return
\end{algorithm}

\subsubsection{Analysis}
For the analysis, we suppose that \Cref{as:nuni} is true. The following lemma says that the subroutine \FGDB (\Cref{algo:good dense ball}) finds a good dense ball for every $C_i \in \cC$ that is processed by the algorithm.
\begin{lemma}\label{lm:good dense ball} \sloppy
Let $\cC_p  = \{C_{i_1},C_{i_2},\cdots, C_{i_t}\}, t \le k'$ be the red clusters processed by \Cref{algo:nucsr} in Step~\ref{algo:nucsr while} in that order, which is also the same order in which the subroutine \FGDB\ is invoked.
Then, \Cref{as:nuni} implies the following.     For iteration $s \in[t]$, consider the cluster $C_{i_s} \in \cC_p$ for which \FGDB\ is invoked on $(r_{i_s}, (Q_{i_s}, F_{i_s}))$. Then, the following is true.
    \begin{enumerate}
        \item At the beginning of iteration $s$, the ball $G_{i_s}:= \ball_P(o_{i_s}, r_{i_s}) \cap C_{i_s}$ is a good dense ball for $C_{i_s}$ in $(Q_{i_s}, F_{i_s})$. Therefore,  \FGDB\ returns a good dense ball $D_{i_s}$  for $C_{i_s}$ in  $(Q_{i_s}, F_{i_s})$ within $2k$ iterations.
        \item Iteration $s$ neither deletes any point from $C_{i_{\hat{s}}}$ for $\hat{s} >s$ nor deletes a center from $O \setminus \{o_{i_s}\}$.
    \end{enumerate}
\end{lemma}
\begin{proof}
We prove the claim by induction on $s \in [t]$. For the base case, note that the first cluster in $\cC_p$ is $C_1$ and hence $Q_{i_1}=P$ and $F_{i_1}=P$. Thus, $G_{i_1} = \ball_P(o_1,r_1) \cap C_1$ is a good dense ball for $C_1$ in $(P,P)$, as required.
Furthermore,  \FGDB\ either deletes one cluster from $\cC \setminus \{C_{1}\}$ in $P$ at Step~\ref{algo: nuni fgdb del ball}  or deletes one center from $O \setminus \{o_{i_s}\}$ at Step~\ref{algo: nuni fgdb del cen}. 
Thus, the for loop \FGDB\ runs at most $2k$ times. However, in both the cases, neither any point in  $C_1 = G_{i_1}$ is deleted   nor $o_{1}$ is deleted. 
Hence, $G_{i_1}$ is a candidate for the maximizer of the function at Step~\ref{algo: nuni fgdb max} of \FGDB\ for  $F' \subseteq F, P' \subseteq Q$ considered in any iteration. Therefore,  $\min\{U_{y_{1}}, |\ball_{P'}(y_{1},r_{1})|\} \ge \min\{U_{o_1}, |G_{i_1}|\}$. However, note that  $\min\{U_{o_{1}},  |G_{i_1}|\} \ge |C_{1}|$. 
 Thus, we have  $U_{y_{1}} \ge |C_{1}| \text{ and }  |\ball_{P'}(y_{1},r_{1})| \ge |C_1|$.
Hence, \FGDB\ returns a good dense ball for $C_{1}$ within $2k$ iterations. Finally, note that if iteration $1$ deletes $D_1$ from $Q$ at Step~\ref{algo: nuni exch del di} of subroutine \texttt{EXCHANGE}, then it colors all the clusters in $\rchble(D_1)$ black and hence $\rchble(D_1) \cap \cC_p = \emptyset$. Hence, no point from $C_{i_{s'}}$ is deleted for $s'>1$. Similarly, whenever iteration $1$ deletes the center $y_1$ of $D_1$ (in Step~\ref{algo: nuni sc del yi} of \SC and Step~\ref{algo: nuni exch del yi} of \texttt{EXCHANGE}), it is the case that $y_1 \notin O \setminus \{o_1\}$ since $D_1$ is good dense ball for $C_1$.

Now, assume that the lemma is true for all iterations $s'< s$, and consider the iteration $s$ and the ball $G_{i_s} = \ball_P(o_{i_s}, r_{i_s}) \cap C_{i_s}$. Induction hypothesis implies that none of the previous iteration $s' < s$ deleted any point from $C_{i_s}$ nor did they delete $o_{i_s}$. Hence, all the points of $G_{i_s}$ remain as they were at the start of the algorithm. Using the above arguments, since every iteration \FGDB\ either deletes a cluster or a optimal center, the for loop of \FGDB\ runs at most $2k$ times. Again, using same arguments, since, in both the cases, neither any point in  $C_{i_s} = G_{i_{s}}$ is deleted   nor $o_{i_{s}}$ is deleted, we have that  $G_{i_s}$ is a candidate for the maximizer of the function at Step~\ref{algo: nuni fgdb max} of \FGDB\ for  $F' \subseteq F, P' \subseteq Q$ considered in any iteration.
 Therefore,  $\min\{U_{y_{1}}, |\ball_{P'}(y_{1},r_{1})|\} \ge \min\{U_{o_1}, |G_{i_1}|\} \ge |C_{i_s}|$. 
Thus, we have  $U_{y_{i_s}} \ge |C_{i_s}| \text{ and }  |\ball_{P'}(y_{i_s},r_{i_s})| \ge |C_{i_s}|$.
Hence, \FGDB\ returns a good dense ball for $C_{i_s}$ within $2k$ iterations. Finally, note that if iteration $i_s$ deletes $D_{i_s}$ from $Q$ at Step~\ref{algo: nuni exch del di} of subroutine \texttt{EXCHANGE}, then it colors all the clusters in $\rchble(D_{i_s})$ black and hence $\rchble(D_{i_s}) \cap \cC_p = \emptyset$. Hence, no point from $C_{i_{\hat{s}}}$ is deleted for $\hat{s}>1$. Similarly, whenever iteration $i_s$ deletes the center $y_{i_s}$ of $D_{i_s}$ (in Step~\ref{algo: nuni sc del yi} of \SC and Step~\ref{algo: nuni exch del yi} of \texttt{EXCHANGE}), it is the case that $y_{i_s} \notin O \setminus \{o_{i_s}\}$ since $D_{i_s}$ is good dense ball for $C_{i_s}$. 
\end{proof}

Now, consider Step~\ref{algo:settlecluster: type1} of \SC$(C_i)$, for some red cluster $C_i$. 
Let $a(r_i)$ denote the radius of anchor for $D_i$, where $D_i$ is the ball obtained in Step~\ref{algo:settle cluster findgoodenseball}. Now, if $C_i \in \rchble(\ball(y_i,\alpha r_i))$, then we call $C_i$ as \textit{Type-$1$ cluster}, otherwise we call it \textit{Type-$2$ cluster}. 
Note that when Type-$2$ cluster $C_i$ is processed in $\texttt{EXCHANGE}(D_i)$, it turns all the clusters in $\texttt{reachable}(D_i)$ black (some of them may have been already black). We call a cluster $C_t \in \texttt{reachable}(D_i)$, a \textit{Type-$3$ cluster}. We also say that the cluster $C_t \in \texttt{reachable}(C_i)$ is \textit{partitioned} by $C_i$.
Since, $r_t > r_i$, for every $C_t \in \texttt{reachable}(D_i)$, we have the following observation.

\begin{observation}\label{obs:type3}
None of \tp{1} or \tp{2} clusters are turned Type-$3$ by this process, and moreover, a Type-$3$ cluster remains a Type-$3$ cluster throughout the execution of the algorithm.    
\end{observation}

Now, we are ready to prove the guarantees of~\Cref{algo:nucsr}.
\begin{lemma}\label{lm: nuni bound}
\Cref{as:nuni} implies that Step~\ref{algo: nuni find feasible sol} successfully finds a feasible solution $(S,\sigma)$ using $\cS$ with cost $(3+2\sqrt{2}+\epsilon)\opt$ in time $n^{O(1)}$.
\end{lemma}
\begin{proof}
 Let $Y$ be the set of centers that the algorithm added to $\cS$. Then, note that $|Y| = k'$ since we add a center for every cluster in $\cC$. More precisely,  for \tp{1} cluster, we add a center at Step~\ref{algo: sc add cen2} of~\Cref{algo:settle cluster}, for \tp{2} cluster, we add a center at Step~\ref{algo: nuni exch del yi} of~\Cref{algo:exchange cluster}, and for \tp{3} cluster, we add a center at
Step~\ref{algo:nucsr add cen tp3} of~\Cref{algo:nucsr}.
We will show an assignment $\sigma': P \rightarrow Y$ such that $(Y,\sigma')$ is feasible for $\cI$ with $\cost((Y,\sigma')) \le \cost(\cS)$. Therefore, Step~\ref{algo: nuni find feasible sol} is successful and returns a feasible solution $(S,\sigma)$ for $\cI$ such that $\cost((S,\sigma)) \le \cost(\cS)$. 

Let $\cC_1 \subseteq \cC$ be the collection of \tp{1} clusters. Let $P_1 = \cup_{C \in \cC_1} C$ and let $P_2 = P \setminus P_1$. Since $P_1$ and $P_2$ are disjoint, we show $\sigma': P_1 \cup P_2 \rightarrow Y$. 
Consider a \tp{1} cluster $C_i \in \cC$ and let $y_i \in Y$ be the corresponding center in $Y$. 
Then, for every $p \in C_i$, we let $\sigma'(p)=y_i$.
Since, $y_i$ is the center of good dense ball $D_i$ returned by \FGDB, we have that $U_{y_i} \ge |C_i|$, as required. Furthermore, $\rad(\sigma'^{-1}(y_i)) \le (\alpha+2)r_i$, as $C_i \in \rchble(\ball(y_i,\alpha r_i))$.

Now consider the points in $P_2$ which are due to \tp{2} and \tp{3} clusters.
Consider a \tp{2} cluster $C_i$ and consider the good dense ball $D_i$ for $C_i$ centered at $y_i$ returned by \FGDB. Since $|D_i| \ge |C_i|$, let $E_i \subseteq D_i$ such that $|E_i| = |C_i|$, for $t \in [t_i]$. 
For each \tp{3} cluster $C_t \in \rchble(D_i)$, let $E_i^t = E_i \cap C_t$.
Since $|E_i|=|C_i|$, we can consider a partition $\{C^t_i\}_{C_t \in \rchble(D_i)}$ of $C_i$ such that $|C^t_i| = |E^t_i|$. Note that since $|C_i| = |E_i| = \cup_{C_t \in \rchble(D_i)} |E^t_i| =  \cup_{C_t \in \rchble(D_i)} |C^t_i| = |C_i|$, such a partition of $C_i$ is possible. Finally, for a \tp{3} cluster $C_t$, let $C'_t = \cup_{C_i \in \pof(C_t)} E^t_i$ and let $C''_t= \cup_{C_i \in \pof(C_t)} C^t_i$. Now, we assign the points of $P_2$ to $Y$ as follows.
For $p \in E_i$, we let $\sigma'(p)=y_i$. As before, we have $|U_{y_i}| \ge |C_i|$. Furthermore, $\rad(\sigma'^{-1}(y_i)) \le r_i$.
Note that $E^t_i \subseteq E_i$, for $C_t \in \rchble(D_i)$, is the subset of $C_t$ that has been assigned to $y_i$ instead of $y_t$. However, we can assign $C^t_i \subseteq C_i$ to $y_t$, noting that $|C^t_i| = |E^t_i|$, and hence $y_t$ serves exactly $|C_t|$ many points. More formally, for $p \in (C_t \setminus C'_t) \cup C''_t$, we let $\sigma(p)=y_t$. The number of points assigned to $y_t$ is $|(C_t \setminus C'_t) \cup C''_t|= C_t = U_{y_t}$ due to Step~\ref{algo:nucsr tp3 yt} of \Cref{algo:nucsr}. Finally, for $p \in \sigma'^{-1}(y_t) \cap C_t$,  we have $\delta(y_t,p) \le r_t$. Consider $p \in \sigma'^{-1}(y_t) \cap C_i$, where $C_i \in \pof(C_t)$. Then,
\[
\delta(p,y_t) \le d(y_t,y_i) + \delta(y_i,p) \le r_t+r_i + 3r_i + 2r_t = 4r_i+3r_t
\]
However, we know that $C_t$ was partitioned due to $D_i=\ball(y_i,r_i)$ because $C_i \notin \rchble(\ball(y_i,\alpha r_i))$ (Step~\ref{algo:settlecluster: type1} of \SC). This means that $\alpha r_i < r_i + 2a(r_i)$, where $a(r_i)$ is the radius of anchor of $D_i$. Since, $C_t \in \rchble(D_i)$, we have that $r_t \ge a(r_i)$. Hence, we have that $r_i < \nicefrac{2r_t}{\alpha-1}$. Therefore, $\delta(p,y_t) \le \nicefrac{8r_t}{\alpha-1} + 3r_t = (\alpha+2)r_t$. Thus, we have $\rad(\sigma'^{-1}(y_t)) \le (\alpha+2)r_t$.

\textbf{Cost of $(Y,\sigma)$.}
We have,   
\begin{align*}
(\cost(Y,\sigma'))^p &= \sum_{C_i \text{ is  \tp{1}}} \rad(\sigma'^{-1}(y_i))^p + \sum_{C_i \text{ is  \tp{2}}} \rad(\sigma'^{-1}(y_i))^p + \sum_{C_i \text{ is  \tp{3}}} \rad(\sigma'^{-1}(y_i))^p \\
&\le  \sum_{C_i \text{ is  \tp{1}}} ((\alpha+2)r_i)^p + \sum_{C_i \text{ is  \tp{1}}} (r_i)^p + \sum_{C_i \text{ is  \tp{1}}} ((\alpha+2)r_i)^p \\
&\le (\alpha+2)^p \sum_{C_i \in \cC} (r_i)^p \\
&\le (1+\eps)^p(\alpha+2)^p \sum_{C_i \in \cC} (r^*_i)^p \quad \text{using~\Cref{lem: eps apx}}
\end{align*}
In fact, we get $(\alpha+2+\eps)=(3+2\sqrt{2}+\eps)$-approximation when the objective is a monotone symmetric norm of radii. 

\textbf{Running time.}
We have the following claim that bounds the number of iterations of \Cref{algo:nucsr} by $n^{O(1)}$.

\begin{claim}
    During each iteration of while loop (Step~\ref{algo:nucsr while}) of  \Cref{algo:nucsr}, at least one red cluster is turned black. Hence, the while loop runs at most $k$ times.
\end{claim}
\begin{proof}
    From \Cref{lm:good dense ball}, we have that \emph{\FGDB}$(r_i,(P,F))$ never fails for every red cluster $C_i$ processed by the while loop of \Cref{algo:nucsr}. If $C_i$ is a Type-$1$ cluster, then $C_i$ is colored black in Step~\ref{algo:settlecluster type1 black} of \SC$(C_i)$. Otherwise, $C_i$ is Type-$2$ cluster and hence it is colored black in Step~\ref{algo:exchangecluster type2 black} of \texttt{EXCHANGE}$(D_i)$.
\end{proof}
\end{proof}

\subsubsection{Proof of \Cref{thm: nuni csr}}\label{ss: proof of nuni}
    Note that it is sufficient to show how to remove \Cref{as:nuni} in \Cref{lm: nuni bound}. We think of \Cref{algo:nucsr} as a linear program that  has access to a guess function and we will transform it to a branching program without guess function. We show how to replace the guess function.
    
    The algorithm makes the following guesses.
    \begin{enumerate}
        \item[E1] Guess an $\epsilon$-approximate radius profile of the clusters in $\cC$. This can be obtained by enumerating $(k/\epsilon)^{O(k)}  n^2$ choices using~\Cref{lem: eps apx}.

        \item[E2] Guesses in \FGDB. Each time \FGDB\ is invoked, the for loop runs $2k$ iterations and each iteration makes three guesses. First, it guesses the radius $r_j$ of the anchor of $\ball(y_i,r_i+2r_j)$, which can be simulated by creating $k$ branches corresponding to the elements of the radius profile $R$. 
        Second guess is whether or not $C_i \in \rchble(\ball(y_i,r_i+2r_j))$. For this, we create two branches---one  for each outcome. The third guess is whether or not $y_i$ is available for $C_i$, which again can be simulated by creating two branches for each outcome of the guess. 
        Hence, we create at most $(4k)^{2k}$ branches of execution for each invocation of \FGDB.
        
        \item[E3] Guesses in \SC. 
        Since we counted the guesses in \FGDB\ separately, the only guess in this subroutine is whether or not $C_i \in \rchble(\ball(y_i,\alpha r_i))$, which results in additional $2$ branches, for each invocation of \SC.

         \item[E4] Guesses in \texttt{EXCHANGE}. This subroutine requires the set $\rchble(D_i)$. This can be done by creating $2^k$ many branches of execution, one for each possible outcome.

         \item[E5]  Guess available $y_t \in \cap_{r_j \in \pof(C_t)} \ball_{F}(y_j,r_j+r_t)$ for $C_t$  at~\Cref{algo:nucsr tp3 yt} of~\Cref{algo:nucsr}.
         We need to show how to find such an available point $y_t$ for $C_t$. First note that $o_t \in\cap_{r_j \in \pof(C_t)} \ball_{F}(y_j,r_j+r_t)$, and hence such a point exists.
         Let $\hat{Y}_t$ be a set of $k$ maximum capacity points in $\cap_{r_j \in \pof(C_t)} \ball_{F}(y_j,r_j+r_t)$. We claim that $\hat{Y}_t$ contains such a point $y_t$. 
         If $o_t \in \hat{Y}_t$, then we are done since $o_t$ is available for $C_t$.
         For the other case when $o_t \notin \hat{Y}_t$, then note that $U_{y} \ge U_{o_t} \ge |C_t|$, for every $y \in \hat{Y}_t$. Since $|\hat{Y}_t| = k$ and $o_t \notin \hat{Y}_t$, there is a point $y_t \in \hat{Y_t}$ such that $y_t \notin O$ and $U_{y_t} \ge |C_t|$. Hence, in either case, there is  a point as required by ~\Cref{algo:nucsr tp3 yt} of~\Cref{algo:nucsr}. Since, $|\hat{Y}_t| =k$, we can branch on all $k$ possibilities.
    \end{enumerate}

    Since \Cref{algo:nucsr} calls \SC at most $k$ times and each execution of \SC creates $k^{O(k)}$ branches of execution, we have that the total running time of  \Cref{algo:nucsr} is bounded by $ 2^{O(k^2\log k  + k \log (k/\epsilon))}n^{O(1)}$.

\begin{remark}
    Since~\Cref{algo:nucsr} explicitly maintains a set $(F)$ from which it selects the center, it also yields $(3+\sqrt{8}+\eps)$ \fpt-approximation for the \emph{supplier (facility)} version of \csr, where points are from $Q \subseteq P$ and the solution of centers must come from another set $F \subseteq P$. For many fundamental problems, such as \kcen, $k$-\textsc{Median}, it is known that the supplier version is harder to approximate than the non-supplier version~\cite{HS85, GOYAL2023190, AnandL24}.
\end{remark}

\subsection{Extension to Cluster \csr}\label{ss: non cluster csr}
In this section, we highlight changes required to~\Cref{algo:nucsr} for cluster capacities. Let  $\cI=((P,\delta),k,\{U_1,\cdots, U_k\})$ be an instance of Non Uniform cluster \csr. The idea is to let~\Cref{algo:nucsr} run on $\cI$ with the following modifications at the places that depended on the node capacities.
\begin{itemize}
    \item[M1]~\Cref{def: available point} for available point:\\
    $y \in P$ is said to be \emph{available for $C_i$} if  $y \notin O \setminus \{o_i\}$.

     \item[M2]~\Cref{algo: nuni fgdb max} of \FGDB:\\
             Let $D_i = \ball_{P'}(y_i,r_i), y_i \in F'$ be the argument that maximizes the function\\ $\max_{y_j \in F'} |\ball_{P'}(y_j,r_i)|$;

    \item[M3]~\Cref{algo: find feasible sol} of~\Cref{algo:nucsr} needs the correct permutation of capacities for $\cS$. Then~\Cref{rem: matching cluster cap} allows us to use~\Cref{thm: feasibility by matching}:\\
            Guess the capacities of the elements in $\cS$;\\
             \textbf{If} $\cS$ is feasible  for $\cI$ \textbf{then} \Return $(S,\sigma)$ obtained from \Cref{thm: feasibility by matching};
\end{itemize}

Note that the capacities are now on the cluster, instead of on the nodes, and hence it makes sense to run~\Cref{algo:nucsr} with above modifications. (M1) is easy to see. For (M2), since we only want to make sure that $D_i$ contains as many points as $C_i$, it is sufficient to look for a densest ball of radius $r_i$, as the algorithm maintains the invariant that all the points of $C_i$ are intact (\Cref{lm:good dense ball}). Next, once we have $\cS$ at~\Cref{algo: find feasible sol} of~\Cref{algo:nucsr}, we need to also guess the capacities of the clusters in $\cS$, in order to use~\Cref{thm: feasibility by matching} (see \Cref{rem: matching cluster cap}). This will incur an additional multiplicative factor of $k!$ to the running time. Finally, we need to guess available $y_t \in \cap_{r_j \in \pof(C_t)} \ball_{F}(y_j,r_j+r_t)$ for $C_t$  at~\Cref{algo:nucsr tp3 yt} of~\Cref{algo:nucsr} (E5). In this case, we define $\hat{Y}_t$ to be an arbitrary set of $k$ facilities in $\cap_{r_j \in \pof(C_t)} \ball_{F}(y_j,r_j+r_t)$. It is easy to see that $\hat{Y}_t$ must contain an available point for $C_t$.

\subsection{Extension to Non-Uniform \csd}\label{ss: non csd}
In this section, we modify~\Cref{algo:nucsr} for Non-Uniform \csd. The pseudo-code is shown in~\Cref{algo:nucsd}.

\subsubsection{Overview of \Cref{algo:nucsd}}
The idea of \Cref{algo:nucsd} is similar to~\Cref{algo:nucsr}, where it starts processing the clusters in $\cC$ in non-decreasing order of diameters. 
As before, the key invariant of the algorithm is that whenever cluster $C_i \in \cC$ is being processed, there is a good dense ball for $C_i$ in the remaining points. 
Specifically, suppose we could find a ball $D_i = \ball(y_i,d_i)$, during the processing of $C_i$ such that $|D_i| \ge |C_i|$. Then, letting $C_j$ be an anchor for $D_i$, i.e., $C_j$ has the smallest diameter in $\rchble(D_i)$, note that if $C_i \notin \rchble(\ball(y_i,d_i+d_j))$, then we can (temporary) delete the $\ball(y_i,d_i+d_j)$, since this ball does not contain any point of $C_i$. Furthermore, note that in this case we end up deleting at least  one cluster (namely, $C_j$), and hence this process of temporary deleting balls can repeat at most $k$ times before the event that $C_i$ is in the reachable set of the extended ball happens. We call such ball a good dense ball.
Once, we find a good dense ball $D_i = \ball(y_i,d_i)$ for $C_i$, i.e.,  $C_i \in \rchble(\ball(y_i,d_i+d_j))$, then note that if $d_j \le d_i$, then all the points of $C_i$ are within distance $3d_i$ from $y_i$. This means we can create a cluster $C'_i \subseteq \ball(y_i,3d_i)$ of diameter $6d_i$ that contains points of $C_i$.

Unfortunately, the hard case is when $d_j > d_i$, and hence we can not charge $d_j$ anymore to pay for the cost of $d_i$. 
In this case, we create a new cluster $C'_i \subseteq D_i$ out of the points of $D_i$. However, in this case, we end up taking points from other clusters. But, note that this process does not take points from the clusters that have already been processed and have diameter smaller than $d_i$. This is because we process the clusters in the non-decreasing order and we are in the case when $d_j > d_i$. Hence, this process ends up taking points from the unprocessed clusters. Since, the algorithm needs the guarantee that the points of the unprocessed clusters are not disturbed, we mark these clusters processed. However, we need to make sure that all the affected clusters are taken care of. Towards this, we use a novel idea of exchanging points using the subroutine \texttt{EXCHANGE}. Since, we have created a new cluster $C'_i$ out of $D_i$, any cluster $C_{j'} \in \rchble(D_i)$ that has lost, say $n_{j'}>0$, points in this process, can instead claim $n_{j'}$ points from the original cluster $C_i$, by paying slightly more cost since, in this case, we can charge to $d_i$  as $d_i < d_{j'} \le d_j$. Here, we used the fact that $C_j$ is anchor for $D_i$, and hence $d_j$ is a smallest diameter in $\rchble(D_i)$. We call such clusters partitioned clusters.
More specifically, the diameter of the modified $C'_{j'}$ is at most $\diam(C_{j'}) + d_i+ \delta(y_i, p) \le  d_{j'} + d_i + 2d_i+d_j \le 5 d_{j'}$, for $p \in C_i$, as $d_j \le d_{j'}$. 
Note that a cluster $C_{j'}$ can be partitioned multiple times by different $C_i$'s, as shown in~\Cref{fig:dia}. However, in this case, the diameter of $C'_{j'}$ is at most $9d_{j'}$ as shown in~\Cref{fig:dia}. On the other hand, the diameter of the new cluster $C'_{i}$ is only $2d_i$. We use this gap between the costs to obtain improved approximation factor of $7$.

\begin{algorithm}[h]
\caption{$(7+\epsilon)$ \fpt approximation for Non Uniform \csd} \label{algo:nucsd}
  \KwData{Instance $\cI=((P,\delta),k,\{U_p\}_{p \in P})$ of Non Uniform \csd, $\epsilon>0$}
  \KwResult{A $7 + \epsilon)$ approximate solution}
    $F \leftarrow P, Q \leftarrow P$ \tcp*{$F$ is a set of potential centers}
    Let $C'_i \gets \emptyset$ for $i \in [k]$\;
    \guess $D=(d_1,\cdots,d_k)$, an $\epsilon$-approximation of optimal diameters such that $d_i \le d_{i+1}$ for $1 \le i \le k-1$, and let $(C_1,\dots,C_k)$ be the names of the corresponding clusters in optimal\;
    Let $\beta \leftarrow \nicefrac{5}{2}$\label{algo:nucsd beta}\tcp*{constant for extended balls}
    Color each cluster $C_i \in \cC$ red\;
    Set $\pof(C_t) = \emptyset$ for $C_t \in \cC$\;
    \While{$\exists$ a red cluster\label{algo:nucsd while}}{
        $C_i \leftarrow$ smallest radius red cluster according to $D$\;
        \SC$(C_i)$\;
    }

    \ForEach{black cluster $C_t$ such that $\pof(C_t) \ne \emptyset$}{
        Set $C'_t = \cup_{d_j \in \pof(C_t)} \ball_{Q}(y_j,2d_j+d_t)$\;
    }
    Let $\cC'$ be the collection of $C'_i$\;
    \lIf{$\cC'$ is feasible for $\cI$}{
        \Return $\tilde{\cC}$ obtained from \Cref{thm: feasibility by matching dia}\label{algo: find feasible sol dia nu}}    
   \textbf{fail}\;
\end{algorithm}

\begin{algorithm}[h]
\caption{\SC$(C_i)$} \label{algo:settle cluster csd}

    Let $D_i = \ball_Q(y_i,d_i)$ be the ball returned by \FGDB$(d_i,(Q,F))$\label{algo:settle cluster findgoodenseball csd} \;
    \guess \eIf{$C_i \in \rchble(\ball(y_i,\beta d_i))$\label{algo:settlecluster: type1 csd}}{
         Color $C_i$ black\label{algo:settlecluster type1 black csd}\;
                set $C'_i=\ball(y_i,(\beta + 1)d_i)$ to $\cS$\label{algo: sc add cen2 csd}\;
    }    
    {\texttt{EXCHANGE}$(D_i)$\tcp*{$d_j > (\beta-1)d_i$}}
    \Return
\end{algorithm}

\begin{algorithm}[h]
\caption{\FGDB$(d_i,(Q,F))$} \label{algo:good dense ball csd}
    $P' \leftarrow Q$\;
    \For{$k$ times}{
        Let $D_i = \ball_{P'}(y_i,d_i), y_i \in P'$ be the argument for $\max_{y_j \in P'}|\ball_{P'}(y_j,d_i)|$\label{algo: nuni fgdb max csd}\;

        \guess the diameter $d_j$ of the smallest diameter cluster $\rchble(D_i)$\;
        \guess \lIf{$C_i \notin \rchble(\ball(y_i,d_i+d_j))$}{
            $P' \leftarrow P' \setminus \ball(y_i,d_i+d_j))$\label{algo: nuni fgdb del ball csd}}
        \lElse{
        \Return $D_i$}
    }
     \textbf{fail}\;
\end{algorithm}

\begin{algorithm}[h]
\caption{\texttt{EXCHANGE}$(D_i)$} \label{algo:exchange cluster csd}

    Delete $D_i$ from $Q$\label{algo: nuni exch del di csd}\; 
    Color $C_i$ black\label{algo:exchangecluster type2 black csd}\;
    Set $C'_i=\ball(y_i,d_i)$\;
    \guess $\rchble(D_i)$\;
    \ForEach{$C_t \in \rchble(D_i)$}{
        Add $d_i$ to \pof$(C_t)$\;
        Color $C_t$ black\;
    }
    \Return
\end{algorithm}

\subsubsection{Analysis}
The analysis is exactly same as that of Non-Uniform Node \csr, except for changing definitions and calculations for the diameter case. We only highlight the required changes.
First, we say a cluster $C_i \in \cC$ is reachable from a ball $B:=\ball(y,r)$ if $\delta(y,p) \le r+d$, for $p \in C_i$. The definitions of $\rchble(B)$ and anchor for $B$ remain same.
Similarly, a ball $D_i:=\ball(y_i,d_i)$ is a good dense ball for $C_i$ if $|D_i| \ge |C_i|$ and $C_i \in \rchble(y_i, d_i+d_j)$, where $d_j$ is the diameter of anchor $C_j$ for $C_i.$

Note that~\Cref{lm:good dense ball} can be adapted to show that \FGDB (\Cref{algo:good dense ball csd}) returns a good dense ball for $C_i$ within $k$ iterations. Next, as before, we partition the clusters in $\cC$ into three types depending on Step~\ref{algo:settlecluster: type1 csd} of \SC$(C_i)$.
Let $a(d_i)$ denote the diameter of anchor for $D_i$, where $D_i$ is the ball obtained in Step~\ref{algo:settle cluster findgoodenseball csd}. Now, if $C_i \in \rchble(\ball(y_i,\beta d_i))$, then we call $C_i$ as \textit{Type-$1$ cluster}, otherwise we call it \textit{Type-$2$ cluster}. 
Note that in this case $a(d_i) > (\beta -1 )d_i$ since $\beta d_i < d_i + a(d_i)$. Furthermore, when Type-$2$ cluster $C_i$ is processed in $\texttt{EXCHANGE}(D_i)$, it turns all the clusters in $\texttt{reachable}(D_i)$ black (some of them may have been already black). 
We call a cluster $C_t \in \texttt{reachable}(D_i)$, a \textit{Type-$3$ cluster}. We also say that the cluster $C_t \in \texttt{reachable}(C_i)$ is \textit{partitioned} by $C_i$. Again, it is easy to see that~\Cref{obs:type3} also holds in this case.

Now, we define the clusters as follows. When $C_i$ is a \tp{1}, we define $C'_i=\ball(y_i,(\beta+1) d_i)$
When $C_i$ is a \tp{2} cluster, then we define $C'_i\subseteq \ball(y_i,d_i)$.
Finally, for a \tp{3} cluster $C_t$, we define $C'_t=\cup_{d_j \in \pof(C_t)} \ball_{Q}(y_j,2d_j+d_t)$. Let $\cC' = \{C'_i \mid \in [k]\}$. 

To show that~\Cref{algo:nucsd} returns a feasible solution  in~\Cref{algo: find feasible sol dia nu} using~\Cref{thm: feasibility by matching dia}, we will show that the set $\cC'$ constructed by the algorithm is a feasible set with cost $(7+\eps)\opt$, i.e., there exists a feasible solution $\tilde{\cC} = \{\tilde{C}_1,\cdots, \tilde{C}_{k}\}$ to $\cI$ such that $\tilde{C}_i \subseteq C'_i$, for $i \in [k']$. 
Note that, we only need to show an existence of such a feasible solution $\tilde{C}$, since in this case~\Cref{thm: feasibility by matching dia}, together with~\Cref{rem: matching cluster cap}, recovers some feasible solution whose cost is at most $\diam(\cC')$, which we is bounded by $(7+\eps)\opt$.
We define $\tilde{\cC}$ as follows. For a \tp{1} cluster $C_i \in \cC$, we define $\tilde{C}_i = C_i$. As  $C'_i = \ball(y_i,(\beta+1)d_i)$ and since $C_i \in \rchble(\ball(y_i,\beta d_i))$,  we have that $\tilde{C}_i \subseteq C'_i$, as required. Furthermore, $|\tilde{C}_i| = |C_i|$ and  the diameter of $\tilde{C}_i$ is at most the diameter of $C'_i$ which is bounded by $2(\beta+1)d_i = 7d_i$.
For a \tp{2} cluster $C_i$, we define $\tilde{C}_i \subseteq D_i = C'_i$ such that $|\tilde{C}_i| = |C_i|$, where $D_i = \ball(y_i,d_i)$ is a good dense ball for $C_i$. It is easy to see that the diameter of $\tilde{C}_i$ is  at most the diameter of $D_i$, which is bounded by $2d_i$.
Finally, when $C_i$ is a \tp{3} cluster (partitioned by other clusters; see~\Cref{fig:dia}), we define $\tilde{C}_i$ to be the points of $C_i$ that are not taken by any \tp{2} cluster union the points claimed back from the \tp{2} clusters. Thus, $|\tilde{C}_i| = |C_i|$ and $\tilde{C}_i \subseteq C'_i$. Furthermore,
we have $\diam(\tilde{C}_i) \le d_i + 2 \cdot \max_{C_t \in \pof(C_i)} (3d_t + d_i) = 3d_i + 6\cdot \max_{C_t \in \pof(C_i)} d_t \le 7d_i$, where we used the fact that $d_t < \frac{d_i}{\beta-1}$ since $\beta d_t <d_i+ d_t$. Finally, note that $P= \dot\cup_{i \in [k]} \tilde{C}_i$, and hence $\dcost(\tilde{\cC}) = \sum_{i \in [k]} \diam(\tilde{C}_i) \le  \sum_{i \in [k]}  2(\beta+1)d_i \le 7(1+\eps)\opt$, which also holds in the case of a monotone symmetric norm of diameters.
\begin{figure}
    \centering
    \includegraphics[width=0.8\linewidth]{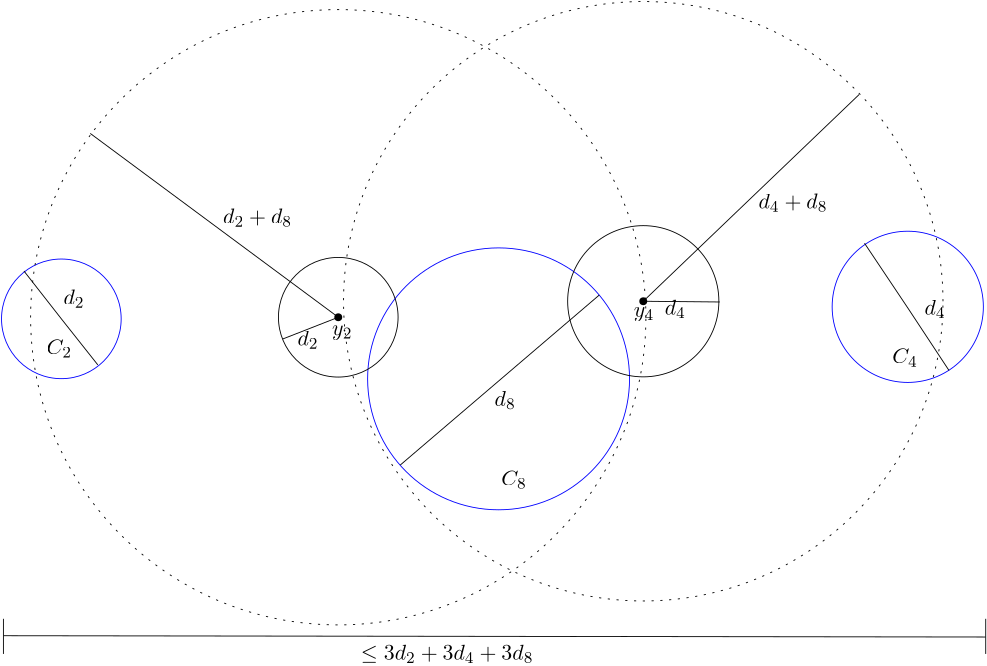}
    \caption{The good dense ball  $\ball(y_2,d_2)$ for $C_2$ partitions $C_8$. Similarly, $C_8$ is also partitioned by good dense ball $\ball(y_4,d_4)$ for $C_4$. However, note that the diameter of the new cluster $C'_8$ is bounded by $9d_8$.}
    \label{fig:dia}
\end{figure}
\section{Hardness of Approximation}\label{sec: hard}
In this section, we prove the following hardness of approximation result for uniform \csr (restated for convenience). Recall that, informally, \gapeth\ says that there exists an $\eps>0$ such that there is no sub-exponential time algorithm for \sat  can distinguish whether a given \sat formula has a satisfying assignment or every assignment satisfies at most $(1-\eps)$ fraction of the clauses (see~\Cref{hyp: gapeth1} for a formal statement).
\FPTLB*

The above theorem implies that there is no $f(k,\epsilon)n^{g(\epsilon)}$ time algorithm that can approximate uniform \csr to a factor better than $(1+1/e)$. Such algorithms are called PAS for Parameterized Approximation Scheme~\cite{a13060146} or  \fpt\ Approximation Scheme. Therefore, \Cref{thm:gapethhard} implies that, assuming \gapeth, there is no PAS for Uniform \csr. 

In fact, our reduction can be modified to rule out PTAS for Uniform \csr, assuming a weaker hypothesis of $\pp \neq \np$ (restated for convenience). 
\SORAPX*

Our hardness results are based on a polynomial time reduction from \textsc{Max $k$-Coverage} (\maxcov), where, given a universe $\cG$ of elements, a positive integer $k$, and $k$ collections $\{\cC_1,\dots,\cC_k\}$ of subsets of $G$, the task is to find $\{S_1,\dots,S_k\}$ such that $S_i\in \cC_i, i\in [k]$ that maximizes the number of elements covered. It is known that \maxcov is NP-hard to approximate to a factor better than $(1-1/e)$, even with a promise that there exists $k$ sets covering $\cG$ and when 
every set in the collection has precisely $|\cG|/k$ elements.

On a high level, given an instance of \maxcov, we create the set-element incidence graph  of the instance. Additionally, for every collection $\cC_i,i\in[k]$, we form a clique on the vertices corresponding to the sets in $\cC_i$. Let $G=(V,E)$ be the resulting graph.
The points of our \csr instance corresponds to the vertices in $V$, and the distance metric is given by the shortest path distance in $G$. Finally, we set the uniform capacity $U=|\cG|/k+m$. It is easy to see that when there are $k$ sets $\{S_1,\dots,S_k\}$ such that $S_i\in \cC_i, i\in [k]$ that cover $\cG$, then we can find a solution for \csr with cost $k$, by selecting the points corresponding to the sets $\{S_1,\dots,S_k\}$ as the centers. For the No case, observe that any solution to \csr must have clusters of size precisely $U$ due to our construction. This simple observation forces the centers to belong to different collections, yielding the result.

\subsection{Preliminaries}\label{ss: hard prelim}
\paragraph*{\textsf{Gap-ETH}}
To rule out \fpt\ algorithms, we will use the following assumption, called \gapeth, introduced by Manurangsi and Raghavendra~\cite{manurangsi_et_al:LIPIcs.ICALP.2017.78}, and independently by Dinur~\cite{e12154b895c549e0875972fe2fb4de91}. \gapeth\ has been instrumental in obtaining tight \fpt\ lower bounds for many fundamental problems~\cite{10.1137/18M1166869,10.5555/3381089.3381094}, including $k$-median, $k$-means~\cite{Cohen-AddadG0LL19}, $k$-center~\cite{GOYAL2023190} in clustering.

\begin{hypothesis}[(Randomized) Gap Exponential Time Hypothesis (\gapeth)~\cite{e12154b895c549e0875972fe2fb4de91,10.5555/3381089.3381094,manurangsi_et_al:LIPIcs.ICALP.2017.78}]\label{hyp: gapeth1}
    There exists constants $\epsilon', \tau >0$ such that no randomized algorithm when given an instance $\phi$ of \sat on $N$ variables and $O(N)$ clauses  can distinguish the following cases correctly with probability $2/3$ in time $O(2^{\tau N})$.
    \begin{itemize}
        \item there exists an assignment for $\phi$ that satisfies all the clauses
        \item every assignment satisfies $< (1-\epsilon')$ fraction of clauses in $\phi$.
    \end{itemize}
\end{hypothesis}

\paragraph*{\textsc{Max $k$-Coverage}}
The starting point of our reduction is \maxcov, where given a universe $\cG$, a collection $\cS$ of subsets of $\cG$, and a positive integer $k$, we want to find $k$ subsets from $\cS$ that maximizes the number of covered elements of $\cG$. For an element $g \in \cG$ and a set $S \subseteq \cG$, we say $g$ is  \emph{covered} by $S$, if $g \in S$.  Similarly, for subsets $S,T \subseteq \cG$, we say $S$ covers the elements of $T$ if $T \subseteq S$.
For our reduction, we need a more general setting.
\begin{definition}[\maxcov]
    Given a universe $\cG$, positive integer $k$, and $k$ collections $\{\cC_1,\dots,\cC_k\}$ of subsets of $\cG$, the \maxcov problem asks to find $k$ sets $\{S_1,\dots,S_k\}$ such that $S_i \in \cC, i\in[k]$ that maximizes the number of covered elements of $\cG$.  
\end{definition}
It is known that it is NP-hard to approximate \maxcov to factors better than $(1-1/e)$.
\begin{theorem}[\cite{10.1145/285055.285059}]\label{thm: maxcov np-hard}
    For every constant $\eps>0$, it is NP-hard to decide whether a given instance of \maxcov has a feasible solution that covers all the elements of $\cG$ or every feasible solution covers at most $(1-1/e+\eps)$ fraction of elements of $\cG$.
    Furthermore, the hardness result holds even when every set in $\cC_i,i\in[k]$ has exactly $|\cG|/k$ elements.
\end{theorem}

In fact, assuming \gapeth, Manurangsi[~\cite{10.5555/3381089.3381094}], showed the following hardness of approximation for \fpt~algorithms.
\begin{theorem}[\cite{10.5555/3381089.3381094}]\label{thm: maxcov enu-hard}
    For every constant $\eps>0$ and any function $f$, there is no $f(k)|\cG|^{o(k)}$ time algorithm that can decide whether a given instance of \maxcov has a feasible solution that covers all the elements of $\cG$ or every feasible solution covers at most $(1-1/e+\eps)$ fraction of elements of $\cG$, unless \gapeth~ fails. 
    Furthermore, the hardness result holds even when every set in $\cC_i,i\in[k]$ has exactly $|\cG|/k$ elements.
\end{theorem}


\subsection{Reduction from \maxcov to Uniform \csr}\label{ss: hardred}

In this section, we show the following result.
\begin{theorem}\label{thm: maxcov red}
    Given any constant $\eps>0$ and an instance $\cI=((\cG,\{\cC_1,\dots,\cC_k\}),k)$ of \maxcov with $|\cC_i| =m,i \in [k]$ such that $m \le \nicefrac{n'}{2k^2}$, where $n'=|\cG|$ and every set in $\cC_i,i \in [k]$ has exactly $\nicefrac{n'}{k}$ elements,  one can compute in $\poly(n')$ time an instance  $\cJ=((P,\delta),k,U)$ of uniform \csr  with $n:=|P| = n'+mk$ points and uniform capacity constrain  $U = n/k$, such that
	\begin{itemize}
		\item If $\exists  \cT=\{T_1,\dots,T_k \mid T_i \in \cC_i, \forall i\in [k]\}$ such that $\cT$ covers $\cG$, then there exist a feasible solution $(S,\sigma)$ with $\cost(S, \sigma) \le k$
		\item If $\forall \cT=\{T_1,\dots,T_k \mid T_i \in \cC_i, \forall i\in [k]\}$,  $\cT$ covers strictly less than $(1-1/e+\eps)n'$ elements of $\cG$, then   then for every feasible solution $(S,\sigma)$ it holds that $\cost(S,\sigma) > (1+1/e-\eps)k$.
	\end{itemize}
\end{theorem}
\begin{proof}
    For every element and every set in $\cI$, create a point in $P$. We call points corresponding to elements, as \emph{element points}, and points corresponding to sets, as \emph{set points}. Formally, for every element $e \in \cG$, create an element point $p_e$ in $P$, and for every set $T \in \bigcup_{i \in [k]} \cC_i$, create a set point $p_T$ in $P$. 
    Note that $n= |P| =n'+mk$.    
    Define the distance function $\delta$ as the shortest path metric of an unweighted graph $G=(V,E)$, where $V=P$, and  $E$ has the following edges. $(i)$ For every $\cC_i, i \in [k]$, there is an edge between every pair of set points $p_T,p_{T'} \in P$ such that the corresponding sets $T,T' \in \cC_i$. $(ii)$ For every set point $p_T \in P$, corresponding to set $T$, there are edges between $p_T$ and the element points corresponding to the elements contained in $T$. Finally, set $U = n/k=n'/k+m$.

 \begin{lemma}\label{lem: yescase}
If $\exists  \cT=\{T_1,\dots,T_k \mid T_i \in \cC_i, \forall i\in [k]\}$ such that $\cT$ covers $\cG$, then there exist a feasible solution $(S,\sigma)$ to the uniform \csr instance $\cJ$ with $\cost(S, \sigma) \le k$.
\end{lemma}   
\begin{proof}
    Without loss of generality,  assume for every $T_i \in \cT, i\in[k]$, it holds that $T_i \in \cC_i$.
    Consider the following solution $(S_\cT,\sigma)$ to $\cJ$, obtained from $\cT$. Let $S_\cT= \{p_{T_1},\dots,p_{T_k}\}$\atodo{these should be $T_i$?}\gtodo{No, $T_i$ are sets, but we want construct a solution for \csr. I changed the notations, though.} contains all the set points corresponding to the sets in $\cT$. 
    For $p_T \in P$, where $p_T$ is a set point corresponding to set $T \in \cC_i,i \in [k]$,    set $\sigma(p_T)=p_{T_i} \in S_\cT$  \atodo{Not clear enough. What is $\cC_j$? why can't you say that $o$ goes to $S_i$ s.t. $p\in S_i$ (and to note that it is unique).}\gtodo{What is $o$? and why do we need uniqueness? Also, it is unclear what it means by``goes". That's why I formulated it differently.}
    Note that this is feasible since $\cT$ contains a set from each $\cC_i$. 
    Furthermore, since $\cT$ covers $\cG$, associate each point $e \in \cG$ with some set $T_j \in \cT, j \in[k]$ covering $e$, by breaking ties arbitrarily. Call $T_j$ as the  \emph{responsible set} for $e$ in solution $\cT$.
    Now, for an element point $p_e \in P$ corresponding to element $e \in \cG$, set $\sigma(p_e)=p_{T_i} \in S_\cT$ such that $T_i$ is the responsible set for $e$ in $\cT$.
    By construction of $(S_{\cT},\sigma)$, the radius of each center in $p_{T_i} \in S_{\cT}$ is $1$, and by the structure of $\cI$, it holds that $\sigma^{-1}(p_{T_i}) = \nicefrac{n'}{k}+m=U$. 
    Hence, $(S_{\cT},\sigma)$ is a feasible solution to $\cJ$ with $\cost((S_\cT,\sigma))=k$.
\end{proof}

 \begin{lemma}\label{lem: nocase}
If $\forall \cT=\{T_1,\dots,T_k \mid T_i \in \cC_i, \forall i\in [k]\}$,  $\cT$ covers strictly less than $(1-1/e+\eps)n'$ elements of $\cG$, then   then for every feasible solution $(S,\sigma)$ it holds that $\cost(S,\sigma) > (1+1/e-\eps)k$.
\end{lemma}   
\begin{proof}
    We prove the contraposition. Suppose there exists $S \subseteq P, |S|\le k$ and a feasible assignment $\sigma: P \rightarrow S$ such that $\cost(S,\sigma) \le (1+1/e-\epsilon)k$. 
    We will show that there exists $\cT=\{T_1,\dots,T_k \mid T_i \in \cC_i, \forall i\in [k]\}$ such that $\cT$ covers at least  $(1-1/e+\eps)n'$ elements of $\cG$.

    First note that every cluster in the solution $(S,\sigma)$  has precisely $n/k=U$ points, since $n= n'+mk=U k$ and each point has a capacity of serving at most $U=n'/k+m$ points. Thus, we have that $|S|=k$ and there are no zero radius clusters in $(S,\sigma)$.    
    Now, let $S_1 \subseteq S$ be the centers that have radius $1$.  We claim that $|S_1| \ge (1-1/e+\eps)k$. Suppose  $|S_1| < (1-1/e+\eps)k$. Then, we have
    \[
        \cost((S,\sigma)) \ge  |S_1|\cdot 1 + (k-|S_1|)2 > (1+\nicefrac{1}{e}-\eps)k,
    \]
    which is a contraction to the cost of $(S,\sigma)$.
    Furthermore, since each cluster in the solution has $n'/k+m$ points, we claim that $S_1$ contains at most one set point from each $\cC_j$, and contains  no element points. For the former, note that any two set points from $\cC_j, j\in[k]$ can cover at most $2n'/k+m <2U$ points within distance $1$, and hence, one of them contains less than $U$ points. Similarly, for the latter, note that any element point can cover at most $mk+1 \le \nicefrac{n'}{2k}+1 < U = \nicefrac{n'}k+m$ points within distance $1$, since $m\ge 1$.
    Now, consider the sets $\cT=\{T_1,\dots, T_{k}\}$ corresponding to the set points of $S_1$. Since, each set point of $S$ covers  at most $m$ set points and in total precisely $U$ points, we have that $\cT$ covers at least $|S_1| (U-m) \ge (1-\nicefrac{1}{e}+\eps)n'/k$ elements of $\cG$, as desired.

\end{proof}
This concludes the proof of the theorem.
\end{proof}

Finally, to obtain the proofs of~\Cref{thm: nphard} and \ref{thm:gapethhard}, we use the reduction of~\Cref{thm: maxcov red} on the hard instances of \maxcov obtained from~\Cref{thm: maxcov np-hard} and~\ref{thm: maxcov enu-hard}, respectively. Note that the condition $m<\nicefrac{n'}{2k^2}$ in~\Cref{thm: maxcov red} can be easily achieved by duplicating $O(k^2m)$ times every element of $\cG$. More specifically, in this case we duplicate each element $2mk^2$ times. Then, letting $n''$ denote the number of elements in the new instance, we have $n'' = n'\cdot 2mk^2$. Now, applying the reduction of~\Cref{thm: maxcov red} on this new instance of \maxcov, it implies that  $U:=\nicefrac{n''}{k}+m$ is still a feasible constraint since each set now contains exactly  ${2mk^2 n'}/{k}=n''/k$ elements. However, we also have that $m=\nicefrac{n''}{2k^2n'}\le \nicefrac{n''}{2k^2}$, as required. Finally, observe that since each element occurs $2n'mk^2$ times, any solution to the new instance of \maxcov covering $\delta$ fraction of elements, also covers $\delta$ fraction of the elements in the original instance.

\paragraph*{Acknowledgments.}
We thank the anonymous reviewer for suggesting improved hardness of approximation construction.

\bibliographystyle{alpha}
\bibliography{papers}
\appendix

\section{Uniform Capacities}\label{sec: uni}
Our first main result of this section is the following (restated for convenience).
 \UnifromMonSymm*

As mentioned previously, \Cref{thm: uni} implies $(6+\eps)$-\fpt \ approximation algorithm for uniform \csd. Interestingly, our algorithmic framework for uniform \csr can be easily adapted for uniform \csd, incurring only an additional factor in the approximation, and thereby yielding $(4+\eps)$-\fpt\ approximation. 

\uniDiameter*
See \Cref{tb:results} for summary of our results for uniform \csr and \csd. 
After setting up the notations and definitions, in~\Cref{ss: csd  algo} we  begin with the proof of \Cref{thm: uni dia} as it is simpler and serves as a warm-up for the more complicated proof of \Cref{thm: uni}, which appears in~\Cref{ss: csr algo}.

\subsection{Technical Overview}\label{ss: overview uni apx}
We overview our algorithms for uniform capacities in more details here. 
For the sake of ease of exposition, the details here differ from the actual algorithms.
We remark that our algorithms  slightly differ from those in~\cite{jaiswal_et_al:LIPIcs.ITCS.2024.65}, and in fact \cite{jaiswal_et_al:LIPIcs.ITCS.2024.65} have superior running time. 
Nonetheless, we designed these algorithms before the publication~\cite{jaiswal_et_al:LIPIcs.ITCS.2024.65}, and choose to keep them as is.

Note that the uniform capacity is at least $U\ge\frac nk$, as otherwise, there will be not feasible solution.
Our algorithm divides the clusters in $\cC$ into two categories: heavy and light (denoted $\cC_H$ and $\cC_L$ respectively).
A cluster $C \in \cC$ is \emph{heavy} if $|C| > \nicefrac{n}{20k^3}$; otherwise, it is a \emph{light} cluster.
The intuition behind this partition is the following: some heavy cluster must exist, and given such $C\in\cC_H$, a randomly selected point from $P$ belongs to $C$ with probability at least $\poly(1/k)$. 
Therefore, with probability at least $k^{-O(k)}$, we can obtain a set $X \subseteq P$ containing a single point from each heavy cluster $C\in\cC_H$.
From the other hand, the total number of points belonging to light clusters is at most $\nicefrac{n}{20k^2}$, and thus it is often the case that all the points belonging to light clusters can be easily ``absorbed''.
We begin by highlighting our algorithm for the capacitated \sod, which also serves as a warm-up for the algorithm we will present for the capacitated \sor.

\paragraph*{\textsf{Sum of Diameters.}}
Let $\cL=\cup\cC_L$ be all the points in light clusters. Initially all the light clusters $\cC_L$ are unsatisfied. 
We consider the heavy clusters iteratively.
Consider a point $x_i \in X$ such that $x_i$ belongs to a heavy cluster $C_i \in \cC_H$. 
The ball $B_i := \ball(x_i, d_i)$ contains the entire cluster $C_i$, i.e., $C_i \subseteq B_i$. 
Let $d_j$ be the maximum diameter of an unsatisfied light cluster $C_j$ intersecting $B_i$ (if there is no such cluster, $d_j=0$). We call such $C_j$, the \emph{leader} of $B_i$.
Let $S_i=C_i\cup(\ball(x_i,d_i+d_j)\cap\cL)$ be the points in the heavy cluster $C_i$ and all the points in light clusters intersecting the ball $\ball(x_i,d_i+d_j)$. The set $S_i$ has diameter at most $2d_i + 2d_j$, and contains all the light clusters intersecting $B_i$.
If $|S_i| \leq U$, we will form a new cluster $C'_i$ corresponding to $S_i$ that includes points from at least one cluster (specifically, $C_i$) from $\cC$.
Next, consider the scenario where $|S_i|>U$.
As $|C_i|\le U$, $S_i$ contains points from at least two clusters, thus we can open two clusters to serve the points in $S_i$ without exceeding the limit of $k$ clusters.
We divide the points in $S_i$ into two clusters: (1) $\hat{C}_i$ containing all the points in $C_i$ with diameter $2d_i$, and (2) $C'_i$ containing all the points in $\ball(x_i,d_i+d_j)\cap\cL$ with diameter $2d_i+2d_j$ (there are at most $|\cL| \leq U$ such points).
The total cost of these two clusters is bounded by $4d_i + 2d_j$. 
At the end of processing $C_i$ all the light clusters intersecting $B_i := \ball(x_i, d_i)$ become satisfied, and we update $\cL\leftarrow\cL\setminus S_i$.

After we finished going over all the heavy clusters (represented by points in $X$), all the points in $\cS=\bigcup_{x_i\in X}S_i$ have been taken care of. 
The set $\cL$ consist of points belonging to light clusters, not yet taken care of.
We proceed greedily: while there is an unclustered point $z\in\cL$, guess its cluster $C_z\in\cC_L$, where $d_z$ is the diameter of $C_z$. As all remaining clusters are light, we can form a cluster $C'_z := \ball(z, d_z)\cap\cL$ without violating the uniform capacity constraint, while ensuring it contains at least one cluster $C_z$ from $\cC$. The diameter of $C'_z$ is at most $2d_z$. 
Next we remove $\ball(z, d_z)$ from $\cL$.
Continue until all the points are clustered, i.e. $\cL=\emptyset$. See \Cref{fig:AltAlgUnifromDiam} for an illustration.

\paragraph*{\textsf{Sum of Radii.}}
We can follow the same ideas from \sod to obtain a $4+\eps$-approximation for \sor with uniform capacities. For each $x_i \in X \cap C_i$ where $C_i \in \cC_H$, we have $B_i := \ball(x_i, 2r_i) \supseteq C_i$. Now, as before, consider the unsatisfied light clusters intersecting $B_i$, and consider the set $S_i=C_i\cup(\ball(x_i,d_i+d_j)\cap\cL)$, where $r_j$ is the maximum radius of such unsatisfied light cluster $C_j$ intersecting $B_i$ (leader of $B_i$).
If $|S_i| \le U$, then a single cluster centered in $x_i$ of radius $2r_i+2r_j$ can serve all the points in $S_i$.
Otherwise, $|S_i| > U$. As previously, we want to serve the points in $S_i$ using two clusters. However, we can use $x_i$ as a cluster center only once. 
Nevertheless, we can sample another point $\hat{x}_i$ from $C_i$ (as $C_i$ is heavy). Next, we open two clusters:
(1) $\ball(\hat{x}_i,2r_i)$ to serve all the points in $C_i$, and (2) $\ball(x_i,2r_i+2r_j)$ to serve all the points in $S_i\setminus C_i$. This gives a $4+\eps$-approximation. 
We introduce a concept of \emph{good dense ball}.

\textbf{Shaving a factor using Good Dense Balls.}
In the algorithm above, the expensive case responsible for the factor $4$ in the approximation is when the set $S_i$ contains more than $U$ points. Note that as $|\cL|<\nicefrac{n}{20k^2}$, it must be the case that the cluster $C$ is very big $|C_i|\ge\nicefrac{n}{2k}$.
We created a cluster centered at an arbitrary center $\hat{x}_i\in C_i$, and thus in order to cover all the points in $C_i$ it had to have radius $2r_i$- twice that of the optimal.
Suppose that instead, we could find a small ball $D_i$ that contains enough points from $S_i$ (which we call a good dense ball) so that $x_i$ can serve the points of $S_i \setminus D_i$ without violating the uniform capacity. Then, we would only pay the radius of $D_i$ for serving the points of $D_i$, in addition to the radius $2r_i + 2r_j$ for $x_i$. 
While~\cite{jaiswal_et_al:LIPIcs.ITCS.2024.65} rely on a matching argument to find low-cost balls, we use a concept of good dense balls and show that one can find them in \fpt \ time.
The key observation is that using the fact that $C_i$ is so big, in \fpt\ time we can find a dense ball $D_i$ containing at least $\nicefrac{n}{20k^2}$ points from $C_i$, with a radius of only $r_i$.
As $|S_i|\le|C_i|+|\cL|\le  U+\nicefrac{n}{20k^2}$, given such a dense ball $D_i$, we can serve the points in $S_i$ using $D_i$ and another ball centered at $x_i$ of radius $2r_i + 2r_j$, and thus a total of $3r_i + 2r_j$ (compared to the $4r_i + 2r_j$ cost in our previous approach). 

The crux of the above argument lies in efficiently finding such a good dense ball $D_i$.  First note that such a dense ball exist: indeed, the ball of radius $r_i$ centered at $o_i$ contains $|C_i|\ge \nicefrac{n}{2k}$ points from $C_i$.
We begin by finding the densest ball $\hat{D}_i$ of radius $r_i$ in $P$. If we are lucky, it happens that $|\hat{D}_i \cap C_i| \geq \nicefrac{n}{20k^2}$, meaning $\hat{D}_i$ is a good dense ball for $C_i$, and we are done. If not, we temporarily delete all the points from $\hat{D}_i$ and recursively find the densest ball of radius $r_i$ among the remaining points. The key observation is that such recursion always finds a good dense ball within a depth of $4k$.
To see this, note that each time we are unlucky, we delete fewer than $\nicefrac{n}{20k^2}$ points from $C_i$.
After $4k$ recursion levels, the number of points remaining in $C_i$ is at least $\nicefrac{n}{2k} - 4k \cdot \nicefrac{n}{20k^2} \geq \nicefrac{n}{4k}$. This implies that $|\hat{D}_i| \geq \nicefrac{n}{4k}$ at every recursion (as we are picking the densest possible ball), and hence the total number of points deleted after $4k$ levels of recursion is at least $4k \cdot \nicefrac{n}{4k} = n$, which is a contradiction since $C_i$ still contains $\nicefrac{n}{4k}$ points.

\begin{wrapfigure}{r}{0.25\textwidth}
	\begin{center}
		\vspace{-20pt}
		\includegraphics[width=0.24\textwidth]{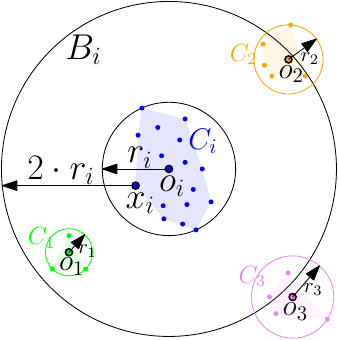}
		\vspace{-5pt}
	\end{center}
	\vspace{-15pt}
\end{wrapfigure}
\paragraph*{Improvement for $\ell_p$ norm objectives.}
For \csr with $\ell_p$ norm objective, our idea to get a better factor for $\ell_p$ norm is to  fine tune the definition of a leader. Recall that for  $B_i = \ball(x_i,2r_i)$,  where $x_i \in X$ belongs to a heavy cluster $C_i \in \cC_H$,  
the leader $C_j$ of $B_i$ was chosen to be the cluster with maximum radius $r_j$ among the unsatisfied light cluster intersecting $B_i$.
Instead, we now define the leader of $B_i$ to be an unsatisfied light cluster $C_j$ intersecting $B_i = \ball(x_i,2r_i)$ with maximum radius such that its center $o_j$ is in $B_i$. 
In the example on the right we consider a heavy cluster $C_i$, where we sampled a point $x_i\in C_i$ and consider the ball $B_i=\ball(x_i,2r_i)$. $B_i$ intersects $3$ light clusters ($C_1,C_2,C_3$) with radii $r_1<r_2<r_3$. The leader is $C_2$ because $o_2\in B_i$ while $o_3\notin B_i$.
Thus, for a light cluster $C_j$, if its center $o_j$ belongs to some ``heavy ball'' $\ball(x_i,2r_i)$ then $C_j$ will be covered in this first phase of the algorithm. 
From the other hand, if $C_j$ was not yet taken care of, then the eventual greedy process will take care of $C_j$.
We formalize this notion by a concept called \emph{Responsibility function} (see~\Cref{def: res fun csr}). 
Note that in this case, we open two clusters: (1) $D_i$ a cluster of radius at most $r_i$ containing at least $\nicefrac{n}{20k^2}$ points from $C_i$ (good dense ball),
and (2) a cluster $C'_i$ containing the remaining points of $C_i$, and all the points from light clusters within the ball $\ball(x_i,2r_i+r_j)$, with radius $2r_i+r_j$ (instead of $2r_i+2r_j$). Note that $C'_i$ fully contain the leader of $C_i$, and thus we fully covered at least two clusters (and hence have enough budget on centers).
Overall, we save a factor of $r_j$ in the radius of the cluster centered at $x_i$. 
Note that this solution yields a $3$ approximation w.r.t. the $\ell_\infty$ norm objective (\ckcen), which is our most significant contribution in this context (improving the $4$-factor from \cite{jaiswal_et_al:LIPIcs.ITCS.2024.65}).
More generally, we get improvement for any $\ell_p$ norm objective (see inequality \ref{eq:ellPratio}).

\begin{figure}[t]
  \centering
  \includegraphics[width=.8\linewidth]{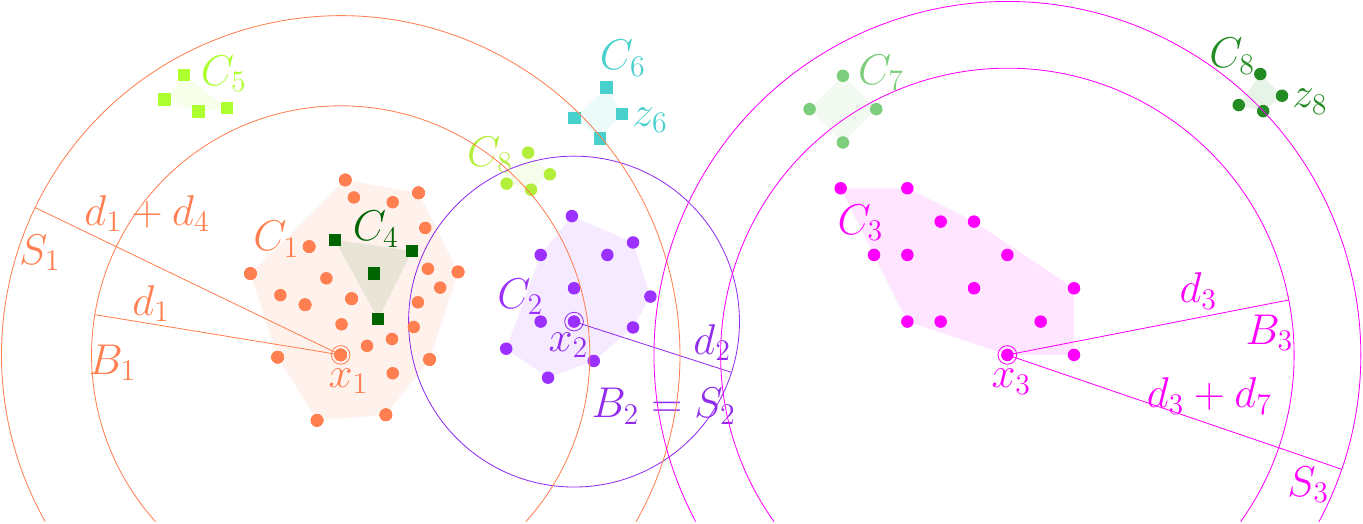}  
\caption{\footnotesize Illustration of our algorithm for uniform Capacitated \sod. The optimal solution consist of $3$ heavy clusters $\cC_H=\{C_1,C_2,C_3\}$, and $6$ light clusters $\cC_L=\{C_4,\dots,C_9\}$. Initially the algorithm guesses a set $X=\{x_1,x_2,x_3\}$ of a single point from each heavy cluster, and let $\cL$ be all the points in light clusters.
The ball $B_1=\ball(x_1,d_1)$ intersects the unsatisfied light clusters $C_4,C_8$. 
$S_1=C_1\cup(\ball(x_1,d_1+d_4)\cap\cL)$ contains the points in $C_1,C_4,C_5,C_8$, and subset of the points in $C_6$. It holds that $|S_1|>U$ and hence we create two clusters to take care of all these points: $C'_1$ for $C_1$, and $\hat{C}_1$ for $S_1\setminus C_1$. $C_4$ and $C_8$ become satisfied, and we update $\cL\leftarrow\cL\setminus S_1$.
Next we consider the ball $B_2=\ball(x_2,d_2)$. The ball $B_2$ does not intersects any unsatisfied cluster ($C_8$ is satisfied), and hence $S_2=B_2$ and we create a single cluster $C'_2$ containing $C_2$ points only.
Next we consider the ball $B_3=\ball(x_3,d_3)$ that intersects the unsatisfied light cluster $C_7$. The set $S_3=C_3\cup(\ball(x_3,d_3+d_7)\cap\cL)$ contains the clusters $C_3,C_7$ and some of the points in $C_8$. $|S_3|\le U$, and thus we cover all $S_3$ points using a single cluster $C'_3$.
Finally, the set $\cL$ consist of subsets of points from $C_6$ and $C_8$. We take care of those greedily by creating the clusters $C'_{z_6}$ and $C'_{z_8}$.
}
\label{fig:AltAlgUnifromDiam}
\end{figure}

\subsection{Notations and definitions}\label{ss: uni notdef}
We denote by $\cI=((P,\delta),k,U)$ an instance of uniform \csr or uniform \csd, depending upon the context. We denote by $\cC=\{C_1,\cdots,C_k\}$  an optimal (but fixed) clustering\footnote{For ease of analysis, we assume that $|\cC| =k$, otherwise, we can add zero radius clusters to make it $k$, without increasing the cost. Another way, which is used in \Cref{algo: 4 apx uni} and \Cref{algo: 3 apx uni}, is to guess $|\cC|$ which results into at most $k$ different choices.} of $\cI$.   
Without loss of generality, we assume that the clusters in $\cC$ are disjoint.
When $\cI$ is an instance of uniform \csr, we denote by $r^*_i$  and $o_i$ as the radius and the center of cluster $C_i \in \cC$, respectively, and let  $O = \{o_1,\cdots, o_k\}$. In this case, we let   $\text{OPT}= r^*_1 + \cdots+ r^*_k$. Note that, here $\cC$ is obtained from the assignment function $\sigma: P \rightarrow O$ of the optimal solution $(O,\sigma)$.
When $\cI$ is an instance of uniform \csd, we denote by $d^*_i$ as the diameter of cluster $C_i \in \cC$. In this case, we let  $\text{OPT}= d^*_1 + \cdots+ d^*_k$.
Finally, we assume we have $\epsilon$-approximation $\{r_1, \cdots, r_k\}$ of $\{r^*_1, \cdots, r^*_k\}$ (for e.g., using~\Cref{lem: eps apx}). Similarly, let  $\{d_1, \cdots, d_k\}$ denote  $\epsilon$-approximation of $\{d^*_1, \cdots, d^*_k\}$.

Once we have fixed an optimal clustering $\cC$ for $\cI$, we can partition the clusters in $\cC$ depending upon their size as follows.
\begin{definition} [Heavy, light, and almost-full clusters]
A cluster $C_i \in \cC$ is said to be a \emph{heavy} cluster if $|C_i| > \nicefrac{n}{20k^3}$, otherwise it is said to be a \emph{light} cluster. Furthermore, a heavy cluster $C_i \in \cC$ is said to be \emph{almost-full} if $|C_i| \ge \nicefrac{n}{2k}$.    
\end{definition}

\begin{definition}
    For $T \subseteq P$ and a cluster $C_i \in \cC$, we say $T$ \emph{hits} $C_i$ if $T \cap C_i \neq \emptyset$. Furthermore, for $\cC' \subseteq \cC$, we say that $T$ is a \emph{hitting set} for $\cC'$ if $T$ hits every cluster in $\cC$.
\end{definition}

\begin{remark}\label{rem:LargeKalgRadii}
        We assume that $k \le \sqrt[4]{n/30}$ (and hence, $n \ge 30$), otherwise $n <30k^4$, and hence the brute force algorithm along with \Cref{thm: feasibility by matching}, running in \fpt\ time in $k$ exactly solves uniform \csr. To see this, note that for \csr, we can guess the set of optimal centers and the corresponding radii correctly in time%
        \footnote{First guess the centers, there are ${n \choose k}$ options. Next for each center $x$, guess the furthest points $p$ assigned to $x$. There are $n$ options (and each such choice determines the corresponding cluster radius). Overall, there are ${n \choose k} n^{k}$ possible guesses. We can go over them one by one, checking for each guess whether it is feasible using \Cref{thm: feasibility by matching}, and thus find the optimum.\label{foot: bruteforce for sor}}
        ${n \choose k} n^{k+O(1)}$ which, in this case, is bounded by $2^{O(k \log k)}$. 

        For the \csd problem where $n <30k^4$, we simply run the same algorithm as above (for \csr) and return the obtained balls with matching as clusters (with the same matching) using \Cref{thm: feasibility by matching dia}. This will be a $2$-approximation (which is better than the guaranteed $4+\eps$ approximation). 
\end{remark}

\subsection{\csd}\label{ss: csd  algo}
The psuedo-code of our algorithm is presented in \Cref{algo: 4 apx uni}, which is based on the ideas presented in~\Cref{sec: tech}.
For the analysis of the algorithm, we make the following assumption. Later, when we prove~\Cref{thm: uni dia} in \Cref{{ss: proof of uni dia}}, we show how to remove this assumption.
\begin{assumption}\label{as: uni dia}
    All the guesses of \Cref{algo: 4 apx uni} are correct. 
\end{assumption}
\subsubsection{Analysis}

Consider $X \subseteq P$, the hitting set for heavy clusters $\cC_H$ obtained by the \Cref{algo: 4 apx uni} at Step~\ref{algo: uni hit heavy dia}, and
$Z \subseteq P$, the hitting set for subset $\cC'_L$ of light clusters in $\cC_L$ processed by the algorithm at Step~\ref{algo: uni hit light dia}. Let $\cC_{hit}$ be the set of clusters hit by $X \cup Z$. Then, note that each cluster in $\cC_{hit}$ is hit exactly by one point in $X \cup Z$.
We will use the following definitions that simplify the exposition.
\begin{definition}
    For a heavy cluster $C_i \in \cC$, let $x_i \in X$ such that $X \cap C_i = x_i$.
    We call $\ball(x_i,d_i)$, the \emph{ball of $C_i$} with respect to $X$ and denote it by $B_i$.
    Similarly, for a light cluster $C_t \in \cC$ hit by $Z$, we call $B_t :=  \ball(z_t,d_t)$,  the \emph{ball of $C_t$} with respect to $Z$, where $z_t = Z \cap C_t$.
\end{definition}

\begin{algorithm} [H]
\caption{$4+\epsilon$ approximation algorithm for uniform \csd} \label{algo: 4 apx uni}
  \KwData{Instance $\cI=((P,\delta),k,U)$ of Uniform \csd, $\epsilon>0$}
  \KwResult{A $4+\epsilon$ approximate solution $\tilde{\cC}$}
     \guess the number of clusters $k'$ in $\cC$\;
     
     \guess $\cC_H, \cC_F,\cC_L$, the set of heavy, almost-full, and light clusters in $\cC$, respectively\;
     Let $ C'_i, \hat{C}_i \leftarrow \emptyset, i\in[k'] $\;
     Let $X \leftarrow \emptyset, Z \leftarrow \emptyset$\;
     \guess $D= \{d_1,\cdots,d_{k'} \}$, an $\epsilon$-approximation of optimal diameters\;
     \lForEach{$C_i \in \cC_H$}{
        \guess a point $x_i$ from $C_i$ and add $x_i$ to $X$\label{algo: uni hit heavy dia}}
     Let $P' \leftarrow P \setminus \bigcup_{x_i \in X} \ball(x_i,d_i)$\;

     \While{$P' \neq \emptyset$}{
        pick a point $z \in P'$ and        add $z$ to $Z$\;\label{algo: uni hit light dia}
        \guess the cluster $C_z \in \cC$ containing $z$ and  
        $P'\leftarrow P' \setminus \ball(z,d_z)$\;
     }
    \guess the responsibility function $\Gamma : \cC_L \rightarrow \cC_{hit}$\tcp*{$\cC_{hit} := \{C \in \cC$  hit by $X \cup Z\}$}
        \ForEach{$C_i \in \cC_{hit}$}{
        let $a \in X \cup Z$ such that $a \in C_i$\;
        set $C'_i = \ball(a,d_i+\ell(d_i))$\label{algo: uni entry for hit dia}\tcp*{$\ell(d_i) := \max_{C_t \in \Gamma^{-1}(C_i)} \diam(C_t)$}
         \lIf{$C_i \in \cC_F$ and $\Gamma^{-1}(C_i) \neq \emptyset$}{
          set $\hat{C}_i = \ball(a,d_i)$\label{algo: uni special entry dia}}
    }
    Let $\cC'$ be the collection of $C'_i$ and $\hat{C}_i$ which are non-empty\;
    \lIf{$\cC'$ is feasible for $\cI$}{
        \Return $\tilde{\cC}$ obtained from \Cref{thm: feasibility by matching dia}\label{algo: find feasible sol dia}}
    \textbf{fail}\;
\end{algorithm}

\begin{definition}[Responsibility function, leader]\label{def: res fun csd}
    Consider a cluster $C_i \in \cC$ hit by $a_i \in X \cup Z$ and the corresponding ball $B_i$ of $C_i$ with respect to $X \cup Z$. For a light cluster $C_t \in \cC$, we say that $C_i$ is \emph{responsible} for $C_t$ in $X \cup Z$  if $\delta(p,a_i) \le d_i + d_t$ for every $p \in C_t$. 
    If there are multiple responsible clusters for $C_t$, then we break ties arbitrarily and pick one.
    Without loss of generality, we assume that a light cluster hit by $Z$ is responsible for itself. Also, note that by construction of $X$ and $Z$, we have that for every light cluster, there is a responsible cluster hit by $X \cup Z$.
    For a light cluster $C_t$, let $\Gamma(C_t)$ denote the responsible cluster  in $X \cup Z$ and hence,  $\Gamma^{-1}(C_i)$ denotes the set of responsibilities (light clusters) of cluster $C_i$ that is hit by $X \cup Z$. Finally, we call a cluster with maximum diameter in $\Gamma^{-1}(C_i)$ as a \emph{leader} of clusters in $\Gamma^{-1}(C_i)$ and denote its diameter by $\ell(d_i)$. If  $\Gamma^{-1}(C_i) =\emptyset$, we define $\ell(d_i)=0$.
\end{definition}
See~\Cref{fig:fig} for a pictorial illustration.
\begin{figure}[h]
	\begin{subfigure}{.55\textwidth}
		\centering
		\includegraphics[width=.8\linewidth]{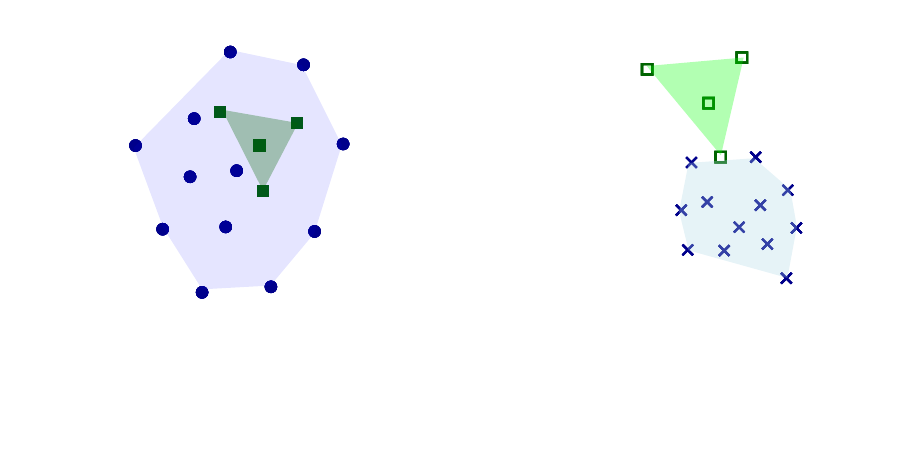}
		\caption{}
	\label{fig:sfig1}
\end{subfigure}%
\begin{subfigure}{.5\textwidth}
	\centering
	\includegraphics[width=.8\linewidth]{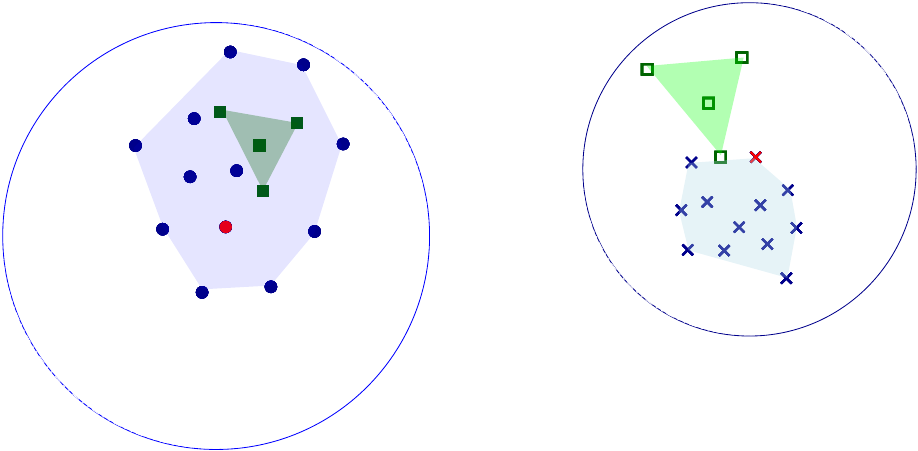}
	\caption{}
\label{fig:sfig2}
\end{subfigure}
\begin{subfigure}{.5\textwidth}
\centering
\includegraphics[width=.9\linewidth]{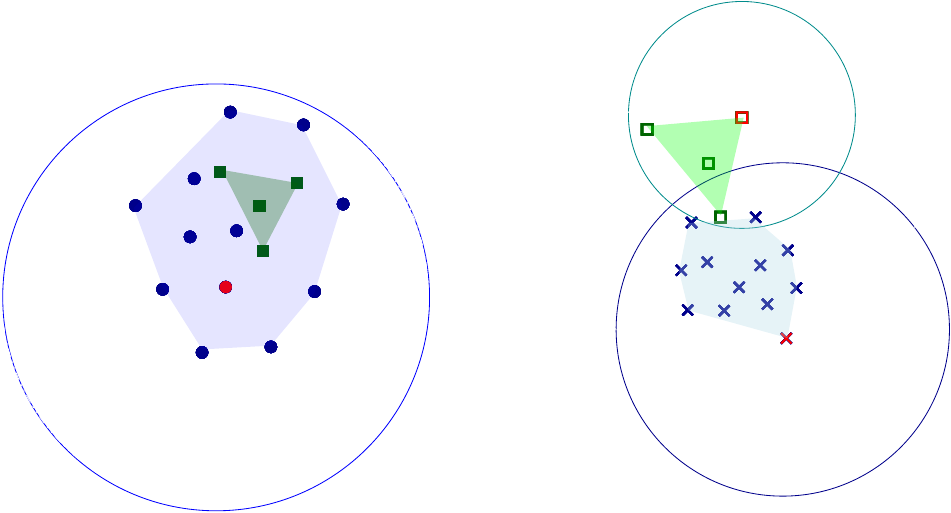}
\caption{}
\label{fig:sfig3}
\end{subfigure}
\begin{subfigure}{.5\textwidth}
\centering
\includegraphics[width=.6\linewidth]{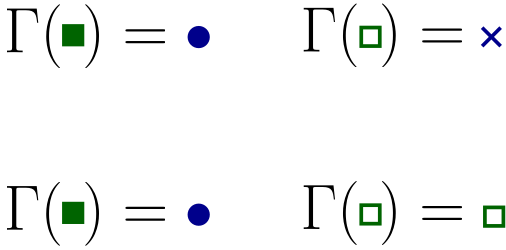}
\caption{}
\label{fig:sfig4}
\end{subfigure}
\caption{\footnotesize \Cref{fig:sfig1} shows an example with $k=4$ clusters -- each corresponding to a specific shape. Furthermore, the clusters with dark blue points correspond to heavy clusters, while cluster with green points correspond to light clusters. In~\Cref{fig:sfig2}, the red colored points are the sampled points from the heavy clusters and the red dotted balls correspond to balls centered at the sampled points of radius equal to the diameter of the corresponding cluster. Note that, in this case, the light clusters are hidden since they are completely covered by the red dotted balls.
\Cref{fig:sfig3} illustrates different choice of sampled points (marked in red) from the heavy clusters. As it can be seen, the red dotted balls of the heavy clusters do not cover all the points and hence, the light cluster of squares is exposed and a point from it is chosen to cover it.
In~\Cref{fig:sfig4}, the first row and second row shows the responsibility function $\Gamma$ due to the red points of \Cref{fig:sfig2} and \Cref{fig:sfig3}, respectively.
}
\label{fig:fig}
\end{figure}
The following lemma bounds the guarantee of \Cref{algo: 4 apx uni}.

\begin{lemma}\label{lm: uni bound dia}
\Cref{as: uni dia} implies that Step~\ref{algo: find feasible sol dia} successfully finds a feasible solution $\tilde{\cC}$ using $\cC'$ with cost $4(1+\epsilon)\opt$ in time $n^{O(1)}$.
\end{lemma}
\begin{proof}
First note that $|\cC'| \le k'$, since any $\hat{C}_i$ that is added to $\cC'$ can be charged to a cluster in $\Gamma^{-1}(C_i)$, as $\Gamma^{-1}(C_i) \neq \emptyset$.
We partition the clusters in $\cC_{hit}$ into two parts, as follows.
\begin{definition}
We call cluster $C_i \in \cC$ hit by $X \cup Z$, a \emph{\tps} cluster if $C_i \in \cC_F$ and $\Gamma^{-1}(C_i)  \neq \emptyset$.
Otherwise, we call $C_i$, a \emph{\tpr} cluster.
\end{definition}

For $C_i \in \cC_{hit}$, let $\cT_i = C_i \cup  \bigcup_{C_t \in \Gamma^{-1}(C_i)} C_t$. Then, note that $\{\cT_i\}_{C_i \in \cC_{hit}}$ partitions the set $P$. Also, note that $|\cT_i| \le U$ for a \tpr cluster $C_i \in \cC_{hit}$. This is because, in this case, $C_i \in \cC_L$ or $C_i \in \cC_H \setminus \cC_F$ or $\Gamma^{-1}(C_i) = \emptyset$, and hence $|\cT_i| \le U$.
Consider $\cC'$ constructed by the algorithm. Then, note that for every cluster $C_i \in \cC_{hit}$, there is a set $C'_i \in \cC'$. Further, if $C_i$ is a \tps cluster, then there is an additional set $\hat{C}_i \in \cC'$. 

We show that $\cC'$ is a feasible set for $\cI$ (see \Cref{def: feasible set csd}) by showing a feasible solution $\tilde{\cC}$ corresponding to $\cC'$, abusing the notation that is used in Step~\ref{algo: find feasible sol dia}. 
Therefore, Step~\ref{algo: find feasible sol dia} is successful and returns a feasible solution $\tilde{\cC}$ for $\cI$ such that $\dcost(\tilde{\cC}) \le \dcost(\cC')$. 
As argued before, we have $|\cC'| \le k' \le k$.
Hence, it is sufficient to assign points from each $\cT_i$'s to a specific cluster $\tilde{C}_i \in \tilde{\cC}$ without violating the uniform capacity, since  $\cT_i$'s form a partition of $P$ and $\cT_i \subseteq C'_i$ or $\cT_i \subseteq \hat{C}_i$.
Consider a \tpr cluster $C_i$, then define $\tilde{C}_i = \cT_i$.
Note that the number of points in the cluster $\tilde{C}_i$ is 
at most $|\cT_i| \le U$, since $C_i$ is a \tpr cluster. Also, $\diam(\tilde{C}_i) \le 2(d_i + \ell(d_i))$.

Now, consider a \tps cluster $C_i$ and the corresponding sets $C'_i$  and $\hat{C}_i$ in $\cC'$.
Let $C_j \in \cC$ be the leader of $\Gamma^{-1}(C_i)$.
Let $D'_i \subseteq C_i$ such that\footnote{Such $D'_i$ is feasible since $|C_i| \ge n/2k$.} $|D'_i| = \nicefrac{n}{20k^2}$. Then, define $\tilde{C}_i = D'_i$ and define $\tilde{C}_j = \cT_i \setminus D'_i$.
In this case, we have  $|\tilde{C}_i| \le\nicefrac{n}{20k^2} \le U$ and  $|\tilde{C}_j| \le |\cT_i \setminus D'_i| \le U$. Finally, $\diam(\tilde{C}_j) \le 2d_i$ and $\diam(\tilde{C}_i) \le 2d_i + 2d_j$.

\textbf{Cost of $\cC'$. } 
We first analyze the cost $\cC'$ for  $\ell_1$ objective function, which is uniform \csd.
We have, in this case,  $\dcost(\tilde{\cC}) \le \sum_{C'_i \in \cC'} \diam(C'_i) + \sum_{\hat{C}_i \in \cC'} \diam(\hat{C}_i) = \dcost(\cC')$. Furthermore,
\begin{align*}
\dcost(\cC') &= \sum_{C'_i \in \cC'} \diam(C'_i) + \sum_{\hat{C}_i \in \cC'} \diam(\hat{C}_i)
\le \sum_{C'_i \in \cC'} (2d_i + 2\ell(d_i)) + \sum_{\hat{C}_i \in \cC'}2d_i\\
&\le \sum_{C_i \in \cC_{hit}} (4d_i + 2\ell(d_i))
\le 4 \sum_{C_i \in \cC} d_i,
\end{align*}
where the last inequality follows since the leaders are pairwise disjoint  as $\Gamma^{-1}(C_i)$ are pairwise disjoint, for $C_i \in \cC_{hit}$. Therefore, $\dcost(\cC') \le 4 (1+\epsilon)\opt$, and hence the cost of $\tilde{\cC}$ returned by the algorithm at Step~\ref{algo: find feasible sol dia} is $\dcost(\tilde{\cC}) \le 4 (1+\epsilon)\opt$.

Now consider $\ell_p$ objective function, for $p\ge 1$. 
We have,  $\dcost(\tilde{\cC}) \le  (\sum_{C'_i \in \cC'} \diam(C'_i)^p + \sum_{\hat{C}_i \in \cC'} \diam(\hat{C}_i)^p )^{1/p} = \dcost(\cC')$. Furthermore,
\begin{align*}
    \frac{\dcost(\cC')}{(\sum_{i \in [k']} (d_i)^p)^{1/p}} &\le \left(\frac{\sum_{C'_i \in \cC'} (2d_i + 2\ell(d_i))^p + \sum_{\hat{C}_i \in \cC'} 2^p(d_i)^p}{\sum_{i \in [k']} (d_i)^p}\right)^{1/p}\\
     &= 2^p \left(\frac{\sum_{C'_i \in \cC'} (d_i + \ell(d_i))^p + \sum_{\hat{C}_i \in \cC'} (d_i)^p}{\sum_{i \in [k']} (d_i)^p}\right)^{1/p}\\
     &\le  2^p \left(\frac{\sum_{C'_i \in \cC'}2^{p-1} ((d_i)^p + (\ell(d_i))^p) + \sum_{\hat{C}_i \in \cC'} (d_i)^p}{\sum_{i \in [k']} (d_i)^p}\right)^{1/p}\\
     &\le  2^p (2^{p-1} +1) \left(\frac{\sum_{i \in [k']}  (d_i)^p}{\sum_{i \in [k']} (d_i)^p}\right)^{1/p}\\
    &\le 2  (2^{p-1}+1)^{1/p}
\end{align*}
Hence, we have 
\[
\dcost(\cC') \le  2  (2^{p-1}+1)^{1/p} (\sum_{i \in [k']} (d_i)^p)^{1/p} \le  2 (1+\eps) (2^{p-1}+1)^{1/p} (\sum_{i \in [k']} (d^*_i)^p)^{1/p}.\] 
When the objective is a monotone symmetric norm of the cluster diameters, then this gives $(4+\epsilon)$-approximation.

Finally, note that the running time of \Cref{algo: 3 apx uni} is $n^{O(1)}$, since the overall algorithm is linear except Step~\ref{algo: find feasible sol dia} which requires time  $n^{O(1)}$. 
\end{proof}
\subsubsection{Proof of \Cref{thm: uni dia}} \label{ss: proof of uni dia}
Now, we are ready to prove \Cref{thm: uni dia}.
    Note that it is sufficient to show how to remove \Cref{as: uni dia} in \Cref{lm: uni bound dia}. We think of \Cref{algo: 4 apx uni} as a linear program that  has access to a guess function and we will transform it to a branching program without guess function. We show how to replace the guess function.
    The algorithm makes the following guesses.
    \begin{enumerate}
        \item[E1] Guess $|O|$ and the clusters in $\cC_H,\cC_F,$ and $\cC_L$. This can be done by enumerating $k2^k2^k = 2^{O(k)}$ choices.
        \item[E2] Guess an $\epsilon$-approximation of optimal diameters. This can be obtained by enumerating $(k/\epsilon)^{O(k)} n^2$ choices given by~\Cref{lem: eps apx}. 
        \item[E3] Guess the responsibility function $\Gamma: \cC_L \rightarrow \cC_{hit}$. Although, this process depends on $(E4)$ ($X \cup Z$), but we can enumerate all $k^k$ choices for $\Gamma$ and later work with only those that are compatible with $(E4)$
        \item[E4] Guess a point from each heavy cluster. Towards this, we sample a point from every heavy cluster with probability at least $(20k^3)^{-k}$, since a heavy cluster contains at least $n/20k^3$ points. 
        We repeat  this sampling procedure $t=(20k^3)^k$ times and create a branch for each repetition. Hence, the probability that at least one branch of the executions correspond to hitting set for $\cC_H$ (good branch) is $\ge 1 - (1 - (20k^3)^{-k})^t \ge 1- e^{t (20k^3)^{-k}}>3/5$. Thus $(E4)$ results in $(20k^3)^k$ branches such that the probability that there exists a good branch in these executions is $3/5$.

    \end{enumerate}

    As mentioned above, we transform the linear execution of \Cref{algo: 4 apx uni} into a branching program by replacing each guess function with a (probabilistic) branching subroutine. Hence, we obtain an algorithm that successfully finds a $4(1+\epsilon)$-approximate solution $\tilde{\cC}$ with probability $3/5$ with running time $2^{O(k)}(k/\epsilon)^{O(k)}k^k(20k^3)^k n^{O(1)} = 2^{O(k \log (k/\epsilon))}n^{O(1)}$.

\subsection{\csr} \label{ss: csr algo}
The psuedo-code of our algorithm is presented in \Cref{algo: 3 apx uni}. The basic framework is similar to \Cref{algo: 4 apx uni} for uniform \csd. First, find a hitting set for heavy clusters and uncovered light clusters. Then, for regular clusters, assign the points from the corresponding clusters to the hitting set. For special clusters, open another cluster (center) to serve the extra points. It is not hard to see that \Cref{algo: 4 apx uni} can be easily modified to get $(4+\epsilon)$-approximation even for uniform \csr. However, to get factor $(3+\epsilon)$-approximation, we use the following crucial concept of good dense ball.

\begin{definition}[Good Dense Ball for $C_i$]\label{def: uni good dense ball}
    For a cluster $C_i \in \cC, p \in P,$ and a positive real $r$, we say that $\ball(p,r)$ is a \emph{good dense ball for $C_i$} if $|\ball(p_,r) \cap C_i| \ge \nicefrac{n}{20k^2}$ and $p \notin O \setminus \{o_i\}$.
\end{definition}
For the analysis of the algorithm, we make the following assumption. Later, when we prove~\Cref{thm: uni} in \Cref{{ss: proof of uni}}, we show how to remove this assumption.
\begin{assumption}\label{as: uni}
    All the guesses of \Cref{algo: 3 apx uni} are correct. 
     Moreover, the points guessed by the algorithm at Step~\ref{algo: uni hit heavy} from heavy clusters are disjoint from the optimal centers $O$, i.e., $X \cap O = \emptyset$.
\end{assumption}

\begin{algorithm} [h!]
\caption{$3+\epsilon$ approximation algorithm for uniform \csr} \label{algo: 3 apx uni}
  \KwData{Instance $\cI=((P,\delta),k,U)$ of Uniform \csr. $\epsilon>0$}
  \KwResult{A $3+\epsilon$ approximate solution $(S,\sigma)$}
     \guess the number of clusters $k'$ in $\cC$\;
     
     \guess $\cC_H, \cC_F,\cC_L$ the set of heavy, almost-full, and light clusters in $\cC$, respectively;

     Let $\cS = (\emptyset,0), F\leftarrow P, X \leftarrow \emptyset, Z \leftarrow \emptyset,$\;
     \guess $R= \{r_1,\cdots,r_{k'} \}$, an $\eps$-approximation of optimal radii\;
     \lForEach{$C_i \in \cC_H$}{
        \guess a point $x_i$ from $C_i$ and add $x_i$ to $X$\label{algo: uni hit heavy}}
     Let $P' \leftarrow P \setminus \bigcup_{x_i \in X} \ball(x_i,2r_i)$\;
     \While{$P' \neq \emptyset$}{
        pick a point $z \in P'$\;
        add $z$ to $Z$\;
        \guess the cluster $C_z \in \cC$ containing $z$ and  $P'\leftarrow P' \setminus \ball(z,2r_z)$\;\label{algo: uni hit light}
     }
    \guess the responsibility function $\Gamma : \cC_L \rightarrow \cC_{hit}$\tcp*{$\cC_{hit} := \{C \in \cC$  hit by $X \cup Z\}$}
    \ForEach{$a_i \in X \cup Z$}{
        let $C_i$ be the cluster hit by $a_i$\;
         add $(a_i, 2r_i + \ell(r_i))$ to $\cS$\label{algo: uni entry for hit}\tcp*{$\ell(r_i) := \max_{C_t \in \Gamma^{-1}(C_i)} \rad(o_t)$}}

    $F \leftarrow F \setminus (X \cup Z)$\label{algo: uni rem xz}\;
    Mark all $C_i \in \cC_F$ such that $\Gamma^{-1}(C_i) \neq \emptyset$\;
     \While{$\exists$ a marked cluster $C_i$\label{algo: uni while}}{
            $\ball(y_i,r_i) \leftarrow$ \FGDB$(\ball(x_i,2r_i),r_i)$\;\label{algo: uni goodense ball}
            add $(y_i,r_i)$ to $\cS$\; \label{algo: uni special entry}
            unmark $C_i$ and $F \leftarrow F \setminus \{y_i\}$\;\label{algo: uni rem y}
     }
    \lIf{$\cS$ is feasible  for $\cI$}{
        \Return $(S,\sigma)$ obtained from \Cref{thm: feasibility by matching}\label{algo: find feasible sol}}
    \textbf{fail}\;
\end{algorithm}

\begin{algorithm}[h!]
\caption{\FGDB$(\ball(x_i,2r_i),r_i)$} \label{algo:uni good dense ball}

    $P'' \leftarrow P$\;
    \For{$4k$ times}{
        Let $D_i = \ball_{P''}(y_i,r_i), y_i \in F$ be a densest ball contained within $\ball(x_i,2r_i)$\;\label{algo: goodense ball find D_i}
        \guess \lIf{$|D_i \cap C_i| \ge \nicefrac{n}{20k^2}$ and $y_i \notin O \setminus \{o_i\}$}{\Return $D_i$}       
        \lElse{ $P'' \leftarrow P'' \setminus D_i$}
    }
     \textbf{fail}\;
\end{algorithm}

\subsubsection{Analysis}
The analysis is similar to that of uniform \csd.
Consider $X \subseteq P$, the hitting set for heavy clusters $\cC_H$ obtained by the \Cref{algo: 3 apx uni}  at Step~\ref{algo: uni hit heavy}, and
$Z \subseteq P$, the hitting set for subset $\cC'_L$ of light clusters in $\cC_L$ processed by the algorithm at Step~\ref{algo: uni hit light}. We have the similar definitions from the analysis of uniform \csd.

\begin{definition}
    For a heavy cluster $C_i \in \cC$, let $x_i \in X$ such that $X \cap C_i = x_i$.
    We call $\ball(x_i,2r_i)$, the \emph{ball of $C_i$} with respect to $X$ and denote it by $B_i$.
    Similarly, for a light cluster $C_t \in \cC$ hit by $Z$, we call $B_t :=  \ball(z_t,2r_t)$,  the \emph{ball of $C_t$} with respect to $Z$, where $z_t = Z \cap C_t$.
\end{definition}

\begin{definition}[Responsibility function, leaders]\label{def: res fun csr}
    Consider a cluster $C_i \in \cC$ hit by $X \cup Z$ and the corresponding ball $B_i$ of $C_i$ with respect to $X \cup Z$. For a  light cluster $C_t \in \cC$, we say that $C_i$ is \emph{responsible} for $C_t$ in $X \cup Z$  if $o_t \in B_i$. 
    If there are multiple responsible clusters for $C_t$, then we break ties arbitrarily and pick one.
    Without loss of generality, we assume that a light cluster hit by $Z$ is responsible for itself. Also, note that by construction of $X$ and $Z$, we have that for every light cluster, there is a responsible cluster hit by $X \cup Z$.
    For a light cluster $C_t$, let $\Gamma(C_t)$ denote the responsible cluster  in $X \cup Z$ and hence,  $\Gamma^{-1}(C_i)$ denotes the set of responsibilities (light clusters) of cluster $C_i$ that is hit by $X \cup Z$. Finally, we say a cluster with maximum radius in $\Gamma^{-1}(C_i)$ as a \emph{leader} of clusters in $\Gamma^{-1}(C_i)$ and denote its radius by $\ell(r_i)$. If  $\Gamma^{-1}(C_i) =\emptyset$, we define $\ell(r_i)=0$.
\end{definition}
The following claim asserts that 
\FGDB (\Cref{algo:uni good dense ball}) runs at most $4k$ times before finding a good dense ball for $C_i \in \cC$, when invoked on input $(\ball(x_i,2r_i),r_i)$.

\begin{claim}\label{cl: bound on findgoodenseball}
Let $\cC_p  = \{C_{i_1},C_{i_2},\cdots\}$ be the almost-full clusters processed by \Cref{algo: 3 apx uni} in that order, and let $x_{i_j} \in C_{i_j} \cap X$.
    Then, \Cref{as: uni} implies that on input $(\ball(x_{i_j},2r_{i_j}),r_{i_j})$, the subroutine  $\FGDB$ finds a good dense ball for $C_{i_j} \in \cC_p$ within $4k$ iterations. 
\end{claim}
\begin{proof}
We prove that \Cref{as: uni} implies the following during the execution of while loop at Step~\ref{algo: uni while}. For  $j \in [|\cC_p|]$, when the while loop at Step~\ref{algo: uni while} is processing $C_{i_j}$,  it holds that
    \begin{enumerate}
        \item  On input $(\ball(x_{i_j},2r_{i_j}),r_{i_j})$, the subroutine  $\FGDB$ finds a good dense ball for $C_{i_j}$ within $4k$ iterations.
        \item For $j' > j$, we have $o_{i_{j'}} \in F$, after Step~\ref{algo: uni rem y}. That is, $o_{i_j'}$ is available in $F$ when $C_{i_{j'}}$ is processed in the future iterations.
    \end{enumerate}
Let $S_{i_j} = \ball(x_{i_j},2r_{i_j})$ and note that $C_{i_j} \subseteq S_{i_j}$. 
Also, let $F_j \subseteq F$ denote the points in $F$ when $C_{i_j}$ is processed at Step~\ref{algo: uni while}.
Consider $j=1$ and note that $o_{i_1} \notin X \cup Z$ due to \Cref{as: uni}, and hence $o_{i_1} \in F_1$. Now, consider the execution of $\FGDB(\ball(x_{i_1},r_{i_1}),r_{i_1})$ for the first $4k$ iterations. 
Let $d^\ell_1$ be the size of the densest ball within $S_{i_1} \cap P''$ centered at a point in $((S_{i_1} \cap F_1) \setminus O) \cup \{o_{i_1}\}$ after iteration $\ell \le 4k$.
Then, note that $d^\ell_1 >  \max\{\frac{n}{2k} - \frac{n\ell}{20k^2},0\} \ge 1$, since the number of points remaining in $C_{i_1}$ after $\ell$ iterations is at least $\max\{\frac{n}{2k} - \frac{n\ell}{20k^2},0\} \ge 1$ and  $o_i\in F_1$, due to \Cref{as: uni}.
Thus, after every iteration $\ell \in [4k]$, the algorithm deletes at least  ${\frac{n}{2k} - \frac{n\ell}{20k^2}} > \frac{3n}{10k}\ge 1$ many points from $S_{i_1}$. Hence, after $\ell=4k$ iterations, the total number of points deleted from $S_{i_1}$ is
\[
> \sum_{t=1}^{4k} \left(\frac{n}{2k} - \frac{nt}{20k^2}\right) > \frac{3n}{10k} \cdot 4k > n.
\]
On the other hand, the number of points remaining in $C_{i_1}$ is at least $\frac{n}{2k} - \frac{4n}{20k} > \frac{n}{20k^2} \ge 1$, resulting in a contradiction. Furthermore, note that the center $y_{i_1}$ of the good dense ball for $C_{i_1}$ returned by \FGDB is such that $y_{i_1} \notin O \setminus \{o_{i_1}\}$, and hence  for $j' > 1$, we have $o_{i_{j'}} \in F$, after Step~\ref{algo: uni rem y}. This concludes the base case of the induction.
Now we assume that the above two points are true for all $j < |\cC_p|$, and consider the processing of $C_{i_{j}}$ at Step~\ref{algo: uni while}.
From induction hypothesis ($2$), we have that $o_{i_j} \in F_j$. And hence, the above argument implies that  $d^\ell_j >  \max\{\frac{n}{2k} - \frac{n\ell}{20k^2},0\} \ge 1$, since the number of points remaining in $C_{i_j}$ after $\ell$ iterations is at least $\max\{\frac{n}{2k} - \frac{n\ell}{20k^2},0\} \ge 1$.
Thus, after every iteration $\ell \in [4k]$, the algorithm deletes at least  ${\frac{n}{2k} - \frac{n\ell}{20k^2}} > \frac{3n}{10k}\ge 1$ many points from $S_{i_j}$. Hence, after $\ell=4k$ iterations, the total number of points deleted from $S_{i_j}$ is
\[
> \sum_{t=1}^{4k} \left(\frac{n}{2k} - \frac{nt}{20k^2}\right) > \frac{3n}{10k} \cdot 4k > n.
\]
On the other hand, the number of points remaining in $C_{i_j}$ is at least $\frac{n}{2k} - \frac{4n}{20k} > \frac{n}{20k^2} \ge 1$, resulting in a contradiction.  Finally, note that the center $y_{i_j}$ of the good dense ball for $C_{i_j}$ returned by \FGDB is such that $y_{i_j} \notin O \setminus \{o_{i_j}\}$, and hence  for $j' > j$, we have $o_{i_{j'}} \in F$, after Step~\ref{algo: uni rem y}.

\end{proof}
Now we can bound the guarantee of \Cref{algo: 3 apx uni}.
\begin{lemma}\label{lm: uni bound}
\Cref{as: uni} implies that Step~\ref{algo: find feasible sol} successfully finds a feasible solution $(S,\sigma)$ using $\cS$ with cost $3(1+\epsilon)\opt$ in time $n^{O(1)}$.
\end{lemma}
\begin{proof}

We partition the clusters in $\cC_{hit}$ into two parts, as follows.
\begin{definition}
We call cluster $C_i \in \cC$ hit by $X \cup Z$, a \emph{\tps} cluster if $C_i \in \cC_F$ and $\Gamma^{-1}(C_i)  \neq \emptyset$.
Otherwise, we call $C_i$, a \emph{\tpr} cluster.
\end{definition}

For $C_i \in \cC_{hit}$, let $\cT_i = C_i \cup  \bigcup_{C_t \in \Gamma^{-1}(C_i)} C_t$. Then, note that $\{\cT_i\}_{C_i \in \cC_{hit}}$ partitions the set $P$. Also, note that $|\cT_i| \le U$ for a \tpr cluster $C_i \in \cC_{hit}$. This is because, in this case, $C_i \in \cC_L$ or $C_i \in \cC_H \setminus \cC_F$ or $\Gamma^{-1}(C_i) = \emptyset$, and hence $|\cT_i| \le U$.
Consider $\cS$ constructed by the algorithm. Then, note that for every cluster $C_i \in \cC_{hit}$, there is a entry $(a_i, 2r_i+\ell(r_i))$ in $\cS$, added by the algorithm at Step~\ref{algo: uni entry for hit}. Further, if $C_i$ is a \tps cluster, then there is an additional entry $(y_i,r_i)$ in $\cS$, added at Step~\ref{algo: uni special entry}. 
Let $Y$ be the set of points $y_i$ corresponding to the additional entries of \tps clusters. Note that by construction, the sets $X$, $Y$, and $Z$ are pairwise disjoint.

\textbf{Intuition.} Consider  $a_i \in X \cup Z$ and let $C_i \in \cC_{hit}$ hit by $a_i$.
The key idea is that if $C_i$ is a \tpr cluster, then $a_i$ can serve all the points in $\cT_i$ without violating the uniform capacity and within radius $2r_i + \ell(r_i)$, since $|\cT_i| \le U$, as reasoned above. If $C_i$ is \tps, then it can happen that the number of points in $\cT_i$ is larger than $U$, violating the uniform capacity. But, in this case, as mentioned above,  there is an additional entry $(y_i,r_i)$ in $\cS$. Further, \Cref{cl: bound on findgoodenseball} guarantees that $\ball(y_i,r_i)$ is a good dense ball for $C_i$, and hence we have $|\ball(y_i,r_i) \cap C_i| \ge \nicefrac{n}{20k^2}$. Therefore, in this case $y_i \in Y$ can serve $\nicefrac{n}{20k^2}$ points $D'_i$ from $C_i$ within radius $r_i$, while $a_i$ can serve the remaining points in $\cT_i$. Note that, $|\cT_i \setminus D'_i| \le |\cT_i| - |D'_i| \le |C_i|$. Since, $\cT_i$'s form a partition of $P$, the total cost of this assignment is at most $3\sum_{i \in [k']} r_i \le 3(1+\epsilon)\opt$, as desired. Also, note that the number of centers opened is at most $k' \le k$.

Now, we prove the above intuition formally. Let $S = X \cup Y \cup Z$ and note that $|S| \le k'$ since $|X|+|Z| \le k'$ and $|Y| \le k' - (|X|+|Z|)$. We will show an assignment $\sigma': P \rightarrow S$ such that $(S,\sigma')$ is feasible for $\cI$ with $\cost((S,\sigma')) \le \cost(\cS)$. Therefore, Step~\ref{algo: find feasible sol} is successful and returns a feasible solution $(S,\sigma)$ for $\cI$ such that $\cost((S,\sigma)) \le \cost(\cS)$. 
Towards this, it is sufficient to assign points from each $\cT_i$'s to a specific center in $S$ without violating the uniform capacity, since  $\cT_i$'s form a partition of $P$.
Consider a \tpr cluster $C_i$ and let $s_i \in S$ be such that $s_i \in C_i$. 
For $p \in \cT_i $, define, \[\sigma'(p) = s_i.\]
Note that $|\sigma'^{-1}(s_i)| \le |\cT_i| \le U$, since $C_i$ is a \tpr cluster. Also, $\rad(s_i) \le 2r_i + \ell(r_i)$.
Now, consider a \tps cluster $C_i$ and let $(s_i,2r_i + \ell(r_i))$ and $(y_i,r_i)$ be the entries corresponding to $C_i$ in $\cS$, where $s_i \in X$ and $y_i \in Y$.
Recall that $\ball(y_i,r_i)$ returned by the subroutine \FGDB$(\ball(s_i,2r_i),r_i)$ at Step~\ref{algo: uni goodense ball} is a good dense ball for $C_i$, due to \Cref{cl: bound on findgoodenseball}. Hence, we have $|\ball(y_i,r_i) \cap C_i| \ge \nicefrac{n}{20k^2}$. Let $D'_i \subseteq D_i \cap C_i$ such that $|D'_i| = \nicefrac{n}{20k^2}$. For $p \in \cT_i$, define
\[
\sigma'(p) = \begin{cases}
 y_i \qquad \text{if } p \in D'_i,\\
 s_i \qquad \text{otherwise}.
\end{cases}
\]
In this case, we have  $|\sigma'^{-1}(y_i)| \le\nicefrac{n}{20k^2} \le U$ and  $|\sigma'^{-1}(s_i)| \le |\cT_i \setminus D'_i| \le U$. Finally, $\rad(y_i) \le r_i$ and $\rad(s_i) \le 2r_i + \ell(r_i)$.

\textbf{Cost of $(S,\sigma')$. } 
We first analyze the cost $(S,\sigma)$ for  $\ell_1$ objective function, which is uniform \csr.
We have, in this case,  $\cost((S,\sigma')) \le \sum_{s_i \in X \cup Z} \rad(s_i) + \sum_{y_i \in Y} \rad(y_i) = \cost(\cS)$. Furthermore,
\begin{align*}
\cost(\cS)) &= \sum_{s_i \in X \cup Z} \rad(s_i) + \sum_{y_i \in Y} \rad(y_i)
\le \sum_{C_i \in \cC_{hit}} (2r_i + \ell(r_i)) + \sum_{C_j \in \cC_{hit}: C_j \text{is \tps}} r_j\\
&\le \sum_{C_i \in \cC_{hit}} (3r_i + \ell(r_i))
\le 3 \sum_{C_i \in \cC} r_i,
\end{align*}
where the last inequality follows since the leaders are pairwise disjoint  as $\Gamma^{-1}(C_i)$ are pairwise disjoint, for $C_i \in \cC_{hit}$. Therefore, $\cost(\cS) \le 3 (1+\epsilon)\opt$, and hence the cost of $(S,\sigma)$ returned by the algorithm at Step~\ref{algo: find feasible sol} is $\cost((S,\sigma)) \le 3 (1+\epsilon)\opt$.

Now consider $\ell_p$ objective function, for $p\ge 1$. 
We have,   $\cost((S,\sigma')) \le (\sum_{s_i \in X \cup Z} \rad(s_i)^p + \sum_{y_i \in Y} \rad(y_i)^p )^{1/p} = \cost(\cS)$. Furthermore,
\begin{align}
    \frac{\cost(\cS)}{(\sum_{i \in [k']} (r_i)^p)^{1/p}} &\le \left(\frac{\sum_{C_i \in \cC_{hit}} (2r_i + \ell(r_i))^p + \sum_{C_j \text{is \tps}} r_j^p}{\sum_{i \in [k']} (r_i)^p}\right)^{1/p}\nonumber\\
    &\le  \max_{C_i \text{ is special}} \left(\frac{(2r_i + \ell(r_i))^p + r_i^p}{(r_i)^p + (\ell(r_i))^p}\right)^{1/p}\nonumber\\
    &\le (2\cdot 3^{p-1}+1)^{1/p}\label{eq:ellPratio}
\end{align}
Please refer \Cref{apx: apx ratio p} for details about the last step. Hence, we have
\[
\cost(\cS) \le (2\cdot 3^{p-1}+1)^{1/p}(\sum_{i \in [k']} (r_i)^p)^{1/p} \le (1+\eps)(2\cdot 3^{p-1}+1)^{1/p}(\sum_{i \in [k']} (r^*_i)^p)^{1/p}.
\]
In fact, for specific values of $p$, we obtain better factors than the one obtained from the above closed form.
\begin{remark}\label{rem: better factor for p}
    For specific values of $p$, we obtain better factors than the above mentioned closed form. For instance, for  $p=2 \text{ and } 3$, we obtain factors $\approx 2.414$ and $\approx2.488$, respectively, instead of $\approx2.65$ and $\approx2.67$ obtained from the closed form. For comparison, optimizing the final expression of~\cite{jaiswal_et_al:LIPIcs.ITCS.2024.65},  $ \left(\frac{\sum_{C_i \in \cC_{hit}} (2r_i + 2\ell(r_i))^p + \sum_{C_j \text{is \tps}} r_j^p}{\sum_{i \in [k']} (r^*_i)^p}\right)^{1/p}$, for $p=2 \text{ and } 3$, yields factors $\approx 2.92$ and $\approx3.191$.
\end{remark}

When the objective is a monotone symmetric norm of the cluster radii, then this gives $(3+\epsilon)$-approximation.

Finally, note that the running time of \Cref{algo: 3 apx uni} is $n^{O(1)}$, since each iteration of the \emph{while} loop unmarks at least one marked cluster, and hence runs at most $k$ times.
 
\end{proof}

\subsubsection{Proof of \Cref{thm: uni}}\label{ss: proof of uni}
Now, we are ready to prove \Cref{thm: uni}.
    Again, as in \Cref{ss: csd  algo}, note that it is sufficient to show how to remove \Cref{as: uni} in \Cref{lm: uni bound}. We think of \Cref{algo: 3 apx uni} as a linear program that  has access to a guess function and we will transform it to a branching program without guess function. We show how to replace the guess function.
    The algorithm makes the following guesses.
    \begin{enumerate}
        \item[E1] Guess $|O|$ and the clusters in $\cC_H,\cC_F,$ and $\cC_L$. This can be done by enumerating $k2^k2^k = 2^{O(k)}$ choices.
        \item[E2] Guess an $\epsilon$-approximate radius profile of the clusters in $\cC$. This can be obtained by enumerating $(k/\epsilon)^{O(k)}  n^2$ choices using~\Cref{lem: eps apx}. 
        \item[E3] Guess the responsibility function $\Gamma: \cC_L \rightarrow \cC_{hit}$. Although, this process depends on $(E4)$ ($X \cup Z$), but we can enumerate all $k^k$ choices for $\Gamma$ and later work with only those that are compatible with $(E4)$
        \item[E4] Guess a point from each heavy cluster. We sample a point from every heavy cluster with probability at least $(20k^3)^{-k}$, since a heavy cluster contains at least $n/20k^3$ points. However, for \Cref{as: uni}, we want that $X \cap O = \emptyset$.
        But, since $|O| \le k$, we have
        \[
        \Pr_{x_i \sim P}[x_i \in C_i \text{ and } x_i \notin O] \ge \frac{n/20k^3 - k}{n} \ge \frac{1}{60k^3},
        \]
        since $k \le \sqrt[4]{n/30}$. 
        Therefore, with probability at least $(60k^3)^{-k}$, we sample $X$ that hits every heavy cluster such that $X \cap O = \emptyset$. 
        Let $k'' \le k$ be the number of heavy clusters in $\cC$.
        We say a $k''$-tupple $A \in P^{k''}$ is \emph{good} if $|A \cap C_i|=1$ for all $C_i \in \cC_H$ and $X \cap O = \emptyset$.
        We repeat the above sampling $t=(60k^3)^k$ times and create a branch for each repetition. Hence, the probability that at least one branch of the executions correspond to good $k''$-tuples is $\ge 1 - (1 - (60k^3)^{-k})^t \ge 1- e^{t (60k^3)^{-k}}>3/5$. Thus $(E4)$ results in $(60k^3)^k$ branches such that the probability that there exists a good $k''$-tuple in these executions is $3/5$.

        \item[E5] Guesses in \FGDB. Note that there are two sub-conditions in the \emph{if} statement - event $(A)$ when $|D_i \cap C_i| \ge n/20k^2$ and event $(B)$ when $y_i \notin O \setminus \{o_i\}$. Hence, for each $D_i$ obtained at Step~\ref{algo: goodense ball find D_i}, there are $4$ possible outcomes - $(A,B), (\neg A, B),\allowbreak (A,\neg B),\allowbreak (\neg A, \neg B)$. Thus, for each iteration of the \emph{for} loop, we create $4$ branches. Since, there are $4k$ iterations, we have at most $4^{4k}$ many different executions such that at least one of them correctly computes a good dense ball for $C_i$ (\Cref{cl: bound on findgoodenseball}). Furthermore, \FGDB\ is invoked at most $k$ times, and hence, we have at most $4^{4k^2}$ many executions such that at least one of them computes a good dense ball for every $C_i$ for which \FGDB is invoked.
    \end{enumerate}

    As mentioned above, we transform the linear execution of \Cref{algo: 3 apx uni} (and \Cref{algo:good dense ball}) into a branching program by replacing each guess function with a (probabilistic) branching subroutine. Hence, we obtain an algorithm that successfully finds a $3(1+\epsilon)$-approximate solution $(S,\sigma)$ with probability $3/5$ with running time $2^{O(k)}(k/\epsilon)^{O(k)}k^k(60k^3)^k 4^{4k^2} n^{O(1)} = 2^{O(k^2 + k \log (k/\epsilon))}n^{O(1)}$.

\section{Omitted Proof}\label{apx: apx ratio p}
We want to show the following.
\begin{lemma} 
 For $x,y \in \mathbb{R}_{\ge 0}$  and $p \ge 1$,
    \[\left(\frac{(2x+y)^p + x^p}{x^p+ y^p}\right) \le (2\cdot 3^{p-1}+1)
    \]
\end{lemma}
\begin{proof}
Consider the case when $y \ge x$.  In this case, we will ``charge" the cost of $x$ to $y$ as follows.
\begin{align*}
    (2x+y)^p + x^p &\le  (\nicefrac{x}{2} + \nicefrac{3x}{2}+y)^p + x^p \\
    &\le (\nicefrac{3y}{2} + \nicefrac{3x}{2})^p + x^p \qquad \text{since } x \le y\\
    &\le \left(\frac{3}{2}\right)^p 2^{p-1} (x^p+y^p)  + x^p \\
    &\le \left(\frac{3^p}{2} +1 \right) (x^p + y^p)\\
    &\le \left(2\cdot3^{p-1} +1 \right) (x^p + y^p),
\end{align*}
where in the third step we used the Minkowski inequality, which states that $(a+b)^p \le 2^{p-1}(a^p+b^p)$, for $a,b \ge 0$.
 
 Now, suppose $x > y$. If $y=0$, then, $(2x+y)^p + x^p \le (2^p+1) x^p \le (2\cdot 3^{p-1}+1)$, for $p\ge 1$. Assume $y >0$, and let $\beta = x/y > 1$. We need the following claim.
    \begin{claim} \label{apx: cl: beta}
        $(2\beta + 1)^p \le 2\cdot 3^{p-1}(\beta^p +1)-1$, for $p \ge 1$.
    \end{claim}
    Before proving the claim, we finish the proof of the lemma as follows.
    \[
      (2x+y)^p + x^p = y^p (2\beta + 1)^p + x^p \le y^p (2\cdot 3^{p-1}(\beta^p +1)-1) + x^p \le  (2\cdot 3^{p-1} + 1) (x^p + y^p),
    \]
    as desired.
    \begin{proof}[Proof of \Cref{apx: cl: beta}]
        Note that for $p=1$, the claim holds with equality. Therefore, assume $p >1$.     Consider the function $f(\beta) = (2\beta + 1)^{p} - 2\cdot 3^{p-1}(\beta^{p} +1)+1$. 
    We will show that $f(\beta)$ is non-increasing for $\beta \ge 1$. Since $f(1)\le 0$, this implies, $f(\beta) \le 0$ for $\beta \ge 1$, finishing the proof. First note that $(2\beta+1) < 3\beta$ since $\beta >1$. 
    Now, consider the derivative of $f(\beta)$,
    \begin{align*}
        f'(\beta) &= 2p(2\beta + 1)^{p-1} - 2p\cdot 3^{p-1} \beta^{p-1}\\
        &< 2p(3\beta)^{p-1}- 2p\cdot 3^{p-1} \beta^{p-1}\\
        &=0.
    \end{align*}        
    \end{proof}

\end{proof}

\end{document}